\documentclass[11pt,a4paper]{article}

\usepackage[a4paper,top=2.10cm,bottom=2.18cm,left=2.28cm,right=2.28cm,headheight=15pt]{geometry}
\usepackage{fontspec}

\usepackage{amsmath,amssymb,amsthm,mathtools,bm,mathrsfs}
\usepackage{booktabs,longtable,tabularx,array,multirow}
\usepackage{enumitem}
\usepackage{graphicx}
\usepackage{tikz}
\usetikzlibrary{arrows.meta,positioning,calc,fit,shapes.geometric,matrix}
\usepackage{xcolor}
\usepackage{microtype}
\usepackage{fancyhdr}
\usepackage{caption}
\usepackage{hyperref}
\usepackage[nameinlink,noabbrev]{cleveref}
\usepackage{bookmark}
\usepackage{titlesec}
\usepackage{setspace}

\hypersetup{
  hidelinks,
  pdfauthor={Xianwei Meng},
  pdftitle={SN-ASMO: Satellite-Navigation Array Spatial-Manifold Precise Observation Theory},
  pdfsubject={Observation Formation, Unified U(1) Geometry, Intrinsic Information, and Preservation of the RTK Integer Structure},
  pdfkeywords={GNSS,array manifold,U(1),phase transport,Fisher information,Fubini-Study,RTK,integer ambiguity}
}

\titleformat{\section}{\Large\bfseries}{\thesection}{0.65em}{}
\titleformat{\subsection}{\large\bfseries}{\thesubsection}{0.65em}{}
\titleformat{\subsubsection}{\normalsize\bfseries}{\thesubsubsection}{0.65em}{}
\titlespacing*{\section}{0pt}{1.45ex}{0.85ex}
\titlespacing*{\subsection}{0pt}{1.20ex}{0.65ex}
\titlespacing*{\subsubsection}{0pt}{1.0ex}{0.50ex}
\setlist[itemize]{leftmargin=2.25em,itemsep=0.25em,topsep=0.35em}
\setlist[enumerate]{leftmargin=2.45em,itemsep=0.30em,topsep=0.38em}
\allowdisplaybreaks[3]
\newtheoremstyle{snstyle}{0.72em}{0.72em}{\normalfont}{}{\bfseries}{.}{0.55em}{}
\theoremstyle{snstyle}
\newtheorem{principle}{Principle}[section]
\newtheorem{definition}{Definition}[section]
\newtheorem{assumption}{Assumption}[section]
\newtheorem{theorem}{Theorem}[section]
\newtheorem{proposition}{Proposition}[section]
\newtheorem{corollary}{Corollary}[section]
\newtheorem{remark}{Remark}[section]

\newcommand{\R}{\mathbb{R}}
\newcommand{\C}{\mathbb{C}}
\newcommand{\Z}{\mathbb{Z}}
\newcommand{\Torus}{\mathbb{T}}
\newcommand{\Uone}{\mathrm{U}(1)}
\newcommand{\SOthree}{\mathrm{SO}(3)}
\newcommand{\CP}{\mathbb{CP}}
\newcommand{\CN}{\mathcal{CN}}
\newcommand{\E}{\mathbb{E}}
\newcommand{\calX}{\mathcal{X}}
\newcommand{\calE}{\mathcal{E}}
\newcommand{\calZ}{\mathcal{Z}}
\newcommand{\calY}{\mathcal{Y}}
\newcommand{\calH}{\mathcal{H}}
\newcommand{\calV}{\mathcal{V}}
\newcommand{\calM}{\mathcal{M}}
\newcommand{\calC}{\mathcal{C}}
\newcommand{\calD}{\mathcal{D}}
\newcommand{\calS}{\mathcal{S}}
\newcommand{\calI}{\mathcal{I}}
\newcommand{\calR}{\mathcal{R}}

\newcommand{\calN}{\mathcal{N}}
\newcommand{\calG}{\mathcal{G}}

\newcommand{\dd}{\mathrm{d}}
\newcommand{\jj}{\mathrm{j}}
\newcommand{\ee}{\mathrm{e}}
\newcommand{\ph}{\operatorname{ph}}
\newcommand{\Arg}{\operatorname{Arg}}
\newcommand{\rank}{\operatorname{rank}}
\newcommand{\Range}{\operatorname{Range}}
\newcommand{\Ker}{\operatorname{Ker}}
\newcommand{\Span}{\operatorname{span}}
\newcommand{\tr}{\operatorname{tr}}
\newcommand{\diag}{\operatorname{diag}}
\newcommand{\cov}{\operatorname{cov}}
\newcommand{\Var}{\operatorname{Var}}
\newcommand{\Rea}{\operatorname{Re}}
\newcommand{\Ima}{\operatorname{Im}}

\newcommand{\crossmat}[1]{[#1]_{\times}}
\newcommand{\norm}[1]{\left\lVert#1\right\rVert}
\newcommand{\abs}[1]{\left\lvert#1\right\rvert}

\newcommand{\trans}{^{\mathsf T}}
\newcommand{\herm}{^{\mathsf H}}

\newcommand{\bmxi}{\bm\xi}
\newcommand{\bmeta}{\bm\eta}
\newcommand{\bmu}{\bm u}
\newcommand{\bmw}{\bm w}
\newcommand{\bmc}{\bm c}
\newcommand{\bma}{\bm a}
\newcommand{\bmg}{\bm g}
\newcommand{\bmz}{\bm z}
\newcommand{\bmn}{\bm n}
\newcommand{\bmtheta}{\bm\theta}
\newcommand{\bmvartheta}{\bm\vartheta}
\newcommand{\nav}{\mathrm{nav}}
\newcommand{\obs}{\mathrm{obs}}

\newcommand{\PI}{\mathrm{PI}}
\newcommand{\hw}{\mathrm{hw}}
\newcommand{\DD}{\mathrm{DD}}

\crefname{equation}{Eq.}{Eqs.}
\Crefname{equation}{Equation}{Equations}
\crefname{theorem}{theorem}{theorems}
\crefname{proposition}{proposition}{propositions}
\crefname{definition}{definition}{definitions}
\crefname{assumption}{assumption}{assumptions}
\crefname{corollary}{corollary}{corollaries}
\crefname{principle}{principle}{principles}
\crefname{remark}{remark}{remarks}
\crefname{figure}{figure}{figures}
\crefname{table}{table}{tables}
\crefname{section}{section}{sections}

\begin{document}
\hypersetup{pageanchor=false}
\begin{titlepage}
\thispagestyle{empty}
\centering
\vspace*{1.55cm}
\vspace{0.30cm}
\vspace{0.75cm}
\vspace{1.25cm}
{\LARGE\bfseries SN-ASMO: Satellite-Navigation Array\par}
\vspace{0.18cm}
{\LARGE\bfseries Spatial-Manifold Precise Observation Theory\par}
\vspace{0.40cm}
{\large Dual-Axis Formation, Phase Transport, and RTK\par}
\vfill
{\Large Xianwei Meng\par}
\vspace{0.35cm}
{\large Hefei University of Technology; Anhui Zhongke Yujiang Technology Co., Ltd.\par}
\vspace{0.25cm}
{\large Corresponding author: \href{mailto:mxianwei@hfut.edu.cn}{mxianwei@hfut.edu.cn}\par}
\vfill
{\large Manuscript THE-INNOVATION-D-26-01250\par}
\vspace{0.20cm}
{\large Final submission manuscript\par}
\vspace{0.45cm}
{\normalsize August 2026\par}
\end{titlepage}

\doublespacing
\pagenumbering{Roman}
\setcounter{page}{1}
\hypersetup{pageanchor=true}

\section*{Public Summary}
\addcontentsline{toc}{section}{Public Summary}
\begin{itemize}
\item Synchronized weights applied to the same GNSS samples show that receiver configuration is an irreducible input to observation formation.
\item The complete behavioral quotient then yields the minimal physical-state and observer-state dual axis; the second axis is derived rather than assumed.
\item On nonzero-response paths, exact same-source phase transport preserves carrier continuity through interference-suppressing weight changes and admitted high dynamics.
\item Under explicit phase-closure, tracking, accuracy, and integrity conditions, high-precision RTK remains compatible with anti-jamming.
\end{itemize}

\section*{Abstract}
\addcontentsline{toc}{section}{Abstract}
Adaptive GNSS arrays must suppress interference while accounting for platform motion and weight-dependent changes in the carrier statistic used by tracking and RTK. We derive the required two-state structure. Applying synchronized admissible weights to one same-source element-sample block at fixed physical state can yield unequal complex conditional means, ruling out dependence on physical state alone. The complete behavioral quotient produces the behavior-minimal observer-state space \(\calC_{\rm obs}\) and admissible joint domain \(\calE\subseteq\calX\times\calC_{\rm obs}\), instantiating the Dual-Axis Axiom of Observation in GNSS. On \(g_s\neq0\), the normalized response defines the one-form \(\Ima(\dd g_s/g_s)\) and, under standard radian normalization and response-tracking compatibility, the compatible \(\Uone\) connection; its horizontal equation is the phase-transport fundamental equation. PTFE path transport requires an explicit nonzero path; same-source inverse transport additionally requires synchronization and phase anchoring. We derive reference covariance, correlated-noise propagation, conditional Fisher increments, and joint residual bias--covariance propagation to RTK integer risk. This analysis separates transport from estimation, relative alignment from absolute closure, and circular phase from integer cycle count. Given a feasible common weight sequence and the stated front-end, synchronization, nonzero-response, phase-closure, tracking, regularity, accuracy, and integrity assumptions, interference suppression and admitted high-dynamic reception are conditionally compatible with high-precision RTK within one receiver chain.

\noindent\textbf{Keywords:} satellite navigation; array spatial manifold; precise observation; observer; \(\Uone\) group; quotient geometry; phase transport; Fisher information; Fubini--Study geometry; RTK; integer ambiguity

\clearpage
\hypersetup{pageanchor=false}
\pagenumbering{arabic}
\setcounter{page}{1}
\hypersetup{pageanchor=true}
\section{Introduction}
\subsection{Background}
Classical GNSS observation theory begins with formed pseudorange, carrier-phase, Doppler, and signal-to-noise measurements, then estimates position, clock offset, atmospheric delay, and integer ambiguity. Carrier phase is modeled through geometric range, clock and propagation terms, hardware biases, and integer ambiguities; antenna and channel effects enter through phase-center models, calibration parameters, and measurement noise. This formulation is effective when the receiver formation configuration is fixed or modeled externally.

Adaptive interference suppression changes this setting. An array receiver updates its spatial or space--time weights as the interference environment changes, while platform translation, rotation, and deformation alter the direction of arrival in array coordinates, inter-element phase, and channel response. The receiver configuration therefore enters the formation of the composite statistic. At a fixed physical state, different admissible weights can produce different complex outputs; under motion, fixed weights sample a changing array response.

Platform motion enters through two typed pathways. The \textbf{propagation pathway} changes geometric range, line-of-sight velocity, and navigation Doppler through \(\alpha_s^{\nav}\). The \textbf{array-response pathway} changes the direction of arrival in array coordinates and hence \(g_s\) through \(D_{\bmxi}g_s\). Weight, channel, synchronization, and processing changes enter separately through \(D_{\bmeta}g_s\). Without this separation, attitude- or weight-induced response phase can be misassigned to propagation range and carried into tracking, differencing, and integer fixing.

Previous studies have addressed array-induced code and carrier-phase bias, bias estimation, phase-stable processing, RTK under interference, and multi-antenna array-aided RTK\cite{delorenzo2006,church2007,church2009,obrien2008,xie2021,li2023,wang2024,wang2025,li2025artk,bamberg2023}. They establish the engineering context but address mainly particular algorithms, constraints, calibration structures, or positioning methods. The remaining interface problem is how to represent physical content, coordinate freedom, measurable relations, statistical information, and integer structure consistently when receiver configuration participates in observation formation.

\subsection{Gaps in Existing Theory}
We organize this problem around four unresolved interfaces.

\begin{enumerate}
\item \textbf{Gap in the observation-formation mechanism.} Conventional models often treat the receiver as the terminal device of the observation chain. Under array reception, however, the array hardware, channels, weights, synchronization, and correlators all participate in forming the observation. If the observer state is omitted from the formation map, physical-state changes and receiver-configuration-induced output changes are confounded within that map.
\item \textbf{Gap in phase structure and identifiability.} The composite array phase is often represented using a selected reference element, a normalization convention, or a calibration quantity, but these coordinate objects are not identical to the absolute composite phase relative to the navigation reference point. What purely relative branches can and cannot measure, and what external conditions are required for absolute closure, must be determined by identifiability rather than intuition.
\item \textbf{Gap between geometry and statistics.} Complex-logarithmic differentials, quotient spaces, Fisher information, and Fubini--Study geometry can describe array responses. Without a correct determination of genuine nuisance freedoms, tangent-space constraints, random weights, and correlated-noise models, however, formal geometric elegance does not guarantee valid statistical conclusions.
\item \textbf{Gap in preservation of the RTK integer structure.} A residual composite array-response phase can generate both a deterministic non-integer bias and additional random noise. Analyzing covariance or nominal fix rate alone cannot guarantee correct fixing. Bias and covariance must be propagated in parallel to the integer Voronoi region and false-fix risk, and the hierarchy between the \(\Uone\) phase branch and the tracking-loop integer cycle count must be explicit.
\end{enumerate}

The common origin of these gaps is that conventional theory focuses on estimating unknown states from given observations without first explaining how the physical world and the observer jointly generate those observations.

\subsection{How the Dual Axis Grows from GNSS Observation Facts}
The Dual-Axis Axiom of Observation is not imposed as a receiver-model assumption. Its set-theoretic foundation constructs complete behavioral quotients from typed observation records and yields a behavior-minimal two-state interface \cite{meng2026dualaxis}. The application-specific task of SN-ASMO is therefore to provide a GNSS nontriviality witness showing that the observer input cannot be eliminated.

Before any coordinate model is chosen, the observation protocol distinguishes two empirically typed roles: the GNSS propagation/platform condition and the receiver configuration that forms the record. This typing is supplied by the experimental protocol and physical model, not inferred from numerical output values. The formation role includes the array and channel condition, coherent reference, digital weights, synchronization, and processing rule. Two configurations represent the same \(\bmeta\) precisely when, for every \(\bmxi\), their admissibility agrees and, whenever both are admissible, their complete output behavior or conditional law agrees. Thus \(\calC_{\rm obs}\) is the behavioral quotient of formation configurations, not a catalogue of implementation labels.

By construction, \(\calC_{\rm obs}\) is behavior-minimal: distinct classes differ in complete GNSS formation behavior on at least one admissible physical state, and no coarser behavior-preserving quotient can merge them. The foundational result establishes the existence, unique descent, and task-relative uniqueness of this quotient \cite{meng2026dualaxis}; the proposition below establishes its nontriviality for GNSS. The descended application interface is
\begin{equation}
\calE\subseteq\calX\times\calC_{\rm obs},
\qquad
\calH:\calE\to\calZ.
\label{eq:dual-axis-descent}
\end{equation}
The product expresses two logical input types; it does not assert that every pair is executable, statistically independent, orthogonal, or symmetric.

\begin{proposition}[GNSS strong witness for a non-degenerate observer axis]
Fix a physical state \(\bmxi_0\), a satellite \(s\), and one coherent element-sample block. Apply two synchronized, admissible digital weights \(\bmw_0\) and \(\bmw_1\) to that same block while holding all other array, channel, synchronization, and processing coordinates fixed. Let \(\bmeta_1\) and \(\bmeta_2\) be the behavioral classes represented by these two configurations, respectively, so that \((\bmxi_0,\bmeta_k)\in\calE\) for \(k=1,2\). Suppose that
\begin{equation}
\alpha_s^{\nav}(\bmxi_0)
(\bmw_0-\bmw_1)\herm\bmc_s^{(O)}(\bmxi_0)\neq0.
\label{eq:gnss-strong-witness}
\end{equation}
Then the two configurations belong to distinct classes in \(\calC_{\rm obs}\), and the complete formation map on \(\calE\) cannot factor through \(\calX\) alone.
\label{prop:gnss-dual-axis-witness}
\end{proposition}
\begin{proof}
The corresponding noiseless complex outputs obey
\begin{equation}
z_{s,0}^{(0)}-z_{s,1}^{(0)}
=\alpha_s^{\nav}(\bmxi_0)
(\bmw_0-\bmw_1)\herm\bmc_s^{(O)}(\bmxi_0)\neq0.
\label{eq:witness-output-difference}
\end{equation}
The unequal outputs place the two configurations in different behavioral classes. If \(\calH\) factored through \(\calX\) alone, both outputs at \(\bmxi_0\) would be equal, contradicting \Cref{eq:witness-output-difference}. For stochastic observations, the same argument applies to conditional means or, more generally, complete conditional laws; two isolated noise realizations are insufficient.
\end{proof}

Physical-axis nondegeneracy is established separately. For a fixed \(\bmeta_0\), it is enough to exhibit \(\bmxi_0,\bmxi_1\) with \((\bmxi_i,\bmeta_0)\in\calE\), \(i=0,1\), such that
\begin{equation}
\alpha_s^{\nav}(\bmxi_0)g_s(\bmxi_0;\bmeta_0)
\neq
\alpha_s^{\nav}(\bmxi_1)g_s(\bmxi_1;\bmeta_0).
\label{eq:physical-axis-witness}
\end{equation}
A change in range, propagation phase, or direction of arrival supplies this witness whenever it changes the displayed product. Together, the two witnesses yield two nondegenerate input axes. They imply neither statistical independence, orthogonality, symmetry, nor global Cartesian-product admissibility. In the three-layer description, \(\bmxi\) and \(\bmeta\) are inputs and \(\bmz\) is the output; \(\bmz\) is not a third input axis.

The resulting deterministic observation-formation relation is
\begin{equation}
\bmz^{(0)}=\calH(\bmxi;\bmeta),
\qquad (\bmxi,\bmeta)\in\calE.
\label{eq:joint-map-intro}
\end{equation}
When an additive stochastic model is justified, it is written separately as \(\bmz=\bmz^{(0)}+\bmn\). Thus geometry is assigned to the formation map and randomness to the conditional measurement law.

\begin{principle}[Preservation of observation physics]
At fixed \(\bmxi\), changing \(\bmeta\) changes how the record is formed but does not alter the navigation-propagation process represented by that physical state. The model must therefore distinguish variations of \(\alpha_s^{\nav}\) and the physical dependence of \(g_s\) from variations induced through \(D_{\bmeta}g_s\). Any quotient, transport, inverse transport, or statistical elimination applied to the latter must preserve propagation geometry, physical motion, continuous carrier phase, and the integer-ambiguity structure.
\label{prin:physics-preservation}
\end{principle}

For the witness pair at \(\bmxi_0\), where \(\bmeta_1\) and \(\bmeta_2\) are the classes represented by the configurations using \(\bmw_0\) and \(\bmw_1\), respectively,
\begin{equation}
\calH(\bmxi_0;\bmeta_1)\neq\calH(\bmxi_0;\bmeta_2).
\label{eq:observer-change}
\end{equation}
Thus a fixed physical state can yield different outputs as the observer state varies, and a change in \(\bmz\) cannot by itself be attributed to a change in \(\bmxi\). SN-ASMO therefore seeks covariance relations, exactly quotientable freedoms, and navigation quantities that remain invariant under admissible observer-state changes.

The theoretical chain developed in this paper is
\begin{equation}
\begin{aligned}
\text{typed GNSS formation records}
&\longrightarrow \text{fixed-physical-state strong witness}\\
&\longrightarrow \text{behavioral quotient and minimal observer state}\\
&\longrightarrow \calE\subseteq\calX\times\calC_{\rm obs},
\quad \calH:\calE\to\calZ\\
&\longrightarrow \text{nonzero complex response and phase transport}\\
&\longrightarrow \text{intrinsic information and RTK cycle continuity}.
\end{aligned}
\label{eq:theory-chain}
\end{equation}

In the basic model retaining only direction of arrival and array weights, extract from \(\bmxi\) the unit direction vector \(\bmu_s\) and from \(\bmeta\) the weight vector \(\bmw\). Then
\begin{equation}
g_s=\calH_A(\bmu_s;\bmw)=\bmw\herm\bmc(\bmu_s).
\label{eq:basic-array-response}
\end{equation}
\Cref{eq:basic-array-response} is only a specialization of the complete mapping; it does not identify \(\bmxi\) with \(\bmu_s\), \(\bmeta\) with \(\bmw\), or \(\bmz\) with \(g_s\).

\subsection{Main Contributions}
This work makes five contributions.

\begin{enumerate}
\item \textbf{Dual-axis observation formation.} A synchronized same-source weight comparison provides the GNSS nontriviality witness, and the behavioral quotient yields \(\calC_{\rm obs}\) and the admissible domain \(\calE\). Under the stated local regularity conditions, the joint map then admits typed physical- and observer-state differentials.
\item \textbf{Response-derived phase transport.} On \(g_s\neq0\), the normalized response determines the response one-form; with the standard radian normalization and response-tracking compatibility fixed, the form determines the compatible \(\Uone\) connection and the phase-transport fundamental equation \cite{meng2026ptfe}. Path transport, endpoint algebra, relative alignment, and absolute phase closure are distinguished explicitly.
\item \textbf{Intrinsic geometry and information.} The paper establishes reference covariance, regular and singular domains, genuine-nuisance quotients, intrinsic Jacobians, and conditional Fisher increments for correlated observations and random weights.
\item \textbf{RTK cycle preservation and integer risk.} Circular phase, continuous lifts, and PLL cycle counts are separated, while residual mean and covariance are propagated jointly to float ambiguities, integer Voronoi regions, and false-fix risk.
\item \textbf{Physical tests.} The theory is paired with reproducible tests for synchronized weight changes, high-dynamic motion, timing mismatch, near-null response, conditional information, and RTK bias risk.
\end{enumerate}

\subsection{Boundary with Related Work and Scope of Novelty}
Prior work has established array-induced carrier-phase bias, bias estimation, phase-stable processing, and joint spatial filtering and RTK under interference \cite{church2009,obrien2008,li2023,wang2024,bamberg2023}. We therefore claim no novelty for these ingredients, projective coordinates, Fisher-information methods, or RTK integer theory considered separately. The contribution lies in their typed composition: the GNSS witness and behavioral quotient determine the joint observation interface, from which reference covariance, response-derived transport, intrinsic information, and RTK preservation conditions are derived under explicit domain assumptions. E-optimal multiweight design, closed-loop path planning, and recursive safety control remain outside the present scope.

\begin{table}[htbp]
\centering
\small
\caption{Boundary between related research directions and SN-ASMO}
\label{tab:boundary}
\begin{tabularx}{\textwidth}{>{\raggedright\arraybackslash}p{0.20\textwidth}*{4}{>{\raggedright\arraybackslash}X}}
\toprule
Research direction & Array bias and compensation & Joint observation formation & Genuine-nuisance quotient & RTK integer risk\\
\midrule
Array-bias modeling and calibration & Primary focus & Local treatment & Usually not central & Local impact analysis\\
Phase-stable or low-distortion spatial processing & Primary focus & Partial treatment & Local constraints & Partial treatment\\
Array-aided RTK/ARTK & Partial treatment & Partial treatment & Not a unified object & Primary focus\\
SN-ASMO (this paper) & Included in the joint model & Explicit joint map & Explicit quotient under a stated action & Parallel bias--covariance propagation\\
\bottomrule
\end{tabularx}
\end{table}

\section{Results: Dual-Axis State and Joint Observation Model}
\subsection{Physical-State Manifold}
\begin{definition}[Physical-state manifold]
Let \(\calX\) be the physical-state manifold of the satellite-navigation array system, with a particular state denoted by \(\bmxi\in\calX\). In a local coordinate chart, it may be written as
\begin{equation}
\bmxi=(\bm r,\bm v,\bm C_a^n,\bmxi_{\rm sat},\bmxi_{\rm env}),
\qquad \bm C_a^n\in\SOthree,
\label{eq:physical-state}
\end{equation}
where \(\bm r\) and \(\bm v\) are the position and velocity of the platform's physical reference point; \(\bm C_a^n\) transforms vectors from array coordinates to navigation coordinates; \(\bmxi_{\rm sat}\) collects satellite-orbit and clock states; and \(\bmxi_{\rm env}\) collects the atmosphere, multipath, and the objective interference field.
\end{definition}

The actual platform attitude belongs to the physical state; the array-body geometry, channel calibration, weights, and processing rules belong to the observer state. For satellite \(s\), the line-of-sight unit vector in navigation coordinates and the direction-of-arrival unit vector in array coordinates are
\begin{equation}
\bmu_s^n=\bmu_s^n(\bmxi),
\qquad
\bmu_s^a=\bm C_n^a\bmu_s^n,
\qquad
\bm C_n^a=(\bm C_a^n)\trans,
\qquad
\norm{\bmu_s^n}=\norm{\bmu_s^a}=1.
\label{eq:los-frames}
\end{equation}
The direction of arrival and attitude are not two independent sets of physical degrees of freedom; \(\bmu_s^a\) already contains the attitude action.

The same physical state enters the observation through two channels with distinct physical roles:
\begin{equation}
\bmxi\longmapsto \alpha_s^{\nav}(\bmxi),
\qquad
(\bmxi,\bmeta)\longmapsto g_s(\bmxi;\bmeta),
\label{eq:two-channels}
\end{equation}
where
\begin{equation}
\alpha_s^{\nav}(\bmxi)
=A_s(\bmxi)\ee^{\jj\phi_s^{\nav}(O;\bmxi)}
\label{eq:nav-factor}
\end{equation}
carries the amplitude and phase of normal navigation propagation, and \(g_s\) is the composite array response under the given observer state. The noiseless joint single-satellite model is
\begin{equation}
 z_s^{(0)}(\bmxi;\bmeta)
=\alpha_s^{\nav}(\bmxi)g_s(\bmxi;\bmeta).
\label{eq:single-joint-model}
\end{equation}
Its complete local differential is
\begin{equation}
\dd z_s^{(0)}
=\underbrace{\left[
g_sD_{\bmxi}\alpha_s^{\nav}
+\alpha_s^{\nav}D_{\bmxi}g_s
\right]\dd\bmxi}_{\text{physical-axis variation}}
+\underbrace{\alpha_s^{\nav}D_{\bmeta}g_s\,\dd\bmeta}_{\text{observer-axis variation}}.
\label{eq:complete-dual-axis-differential}
\end{equation}
This equation is the local first-order counterpart of the finite GNSS witness. The physical axis affects both ordinary propagation and the array response through direction of arrival, platform attitude, and environment; only the \(D_{\bmeta}g_s\) term is observer-state-induced. The two input axes must therefore not be confused with the two physical pathways inside the first bracket.

\subsection{Observer-State Space}
\begin{definition}[Observer state]
Let \(\calC_{\rm obs}\) be the behavior-minimal observer-state space established by the complete behavioral quotient described in \Cref{eq:dual-axis-descent}, and let \(\bmeta\in\calC_{\rm obs}\) denote one state. In local functional coordinates,
\begin{equation}
\bmeta=(\bmeta_{\rm arr},\bmeta_{\rm ch},\bmeta_{\rm proc}),
\label{eq:observer-state}
\end{equation}
where \(\bmeta_{\rm arr}\) includes the actual array-body geometry, element responses, polarization state, and mutual coupling; \(\bmeta_{\rm ch}\) contains the actual channel amplitude, phase, and delay states; and \(\bmeta_{\rm proc}\) includes weights, spatial or space--time filtering structure, operating frequency, coherent reference, code/carrier NCO settings, integration window, sample timing, and data-bit wipeoff rules. Calibration or model estimates of these quantities are written with hats and are not identified with the true state.
\end{definition}

This division does not mechanically split the physical hardware into two devices called ``physics'' and ``observer.'' It distinguishes theoretical identities: the physical state states what condition the objective system is in, whereas the observer state states through what structure and processing rules that condition is converted into a measurement.

The scalar factorization in \Cref{eq:single-joint-model} uses one fixed common coordinate for the NCO and correlator. In this coordinate, \(\alpha_s^{\nav}(\bmxi)\) is the objective pre-NCO propagation factor and \(g_s=\bmw\herm\bmc_s^{(O)}\) is the array response. If the NCO or correlator changes, its receiver phase must enter the generalized response part of \(\calH\); it may not be reassigned silently to \(\alpha_s^{\nav}\).

\subsection{Time-Multiplexed Single Weights and Simultaneous Multiple Weights}
With one weight vector \(\bmw[n]\) at epoch \(n\), the composite response of satellite \(s\) is
\begin{equation}
 g_s[n]=\bmw[n]\herm\bmc_s^{(O)}[n].
\label{eq:single-weight-epoch}
\end{equation}
If adjacent epochs use different weights, their observer states are different; in a dynamic scene, one generally also has \(\bmxi[n]\neq\bmxi[n-1]\). An inter-epoch difference therefore contains both physical-state evolution and observer variation and cannot be interpreted entirely as geometric Doppler without a model.

If \(K\) digital weight vectors are applied synchronously to the same element samples, define
\begin{equation}
\bm W[n]=[\bmw_1[n],\ldots,\bmw_K[n]]\in\C^{M\times K},
\label{eq:multiweight-matrix}
\end{equation}
the joint response and observation vectors as
\begin{equation}
\bmg_s[n]=\bm W[n]\herm\bmc_s^{(O)}[n]
=[g_{s,1}[n],\ldots,g_{s,K}[n]]\trans,
\label{eq:multiweight-response}
\end{equation}
\begin{equation}
\bmz_s[n]=\alpha_s^{\nav}[n]\bmg_s[n]+\bmn_s[n].
\label{eq:multiweight-observation}
\end{equation}
These branches share the same physical signal, clock, NCOs, raw element samples, and part of the noise sources. They are joint statistical branches rather than \(K\) independent physical observations. The architecture studied here always retains one principal working output for ordinary acquisition, tracking, and navigation. Auxiliary weights support phase transport, state identification, weight evaluation, and inverse-transport estimation; they do not create separate physical integer ambiguities.

\subsection{Array Spatial Manifold, Physical Reference Point, and Internal Reference Element}
Two reference objects must be distinguished. The physical reference point \(O\), used to define geometric range and normal propagation phase, belongs to the propagation model. The internal reference element \(\mu\), selected only to represent relative element responses, is a local coordinate choice for the array response.

Suppose the array contains \(M\) elements and \(\calD\subseteq S^2\) is the direction domain. Given an array-observer substate \(\bmeta_A\) and physical reference point \(O\), define the complete element-response map
\begin{equation}
\bmc^{(O)}(\cdot;\bmeta_A):\calD\to\C^M,
\qquad
\bmu^a\mapsto\bmc^{(O)}(\bmu^a;\bmeta_A),
\label{eq:array-map}
\end{equation}
whose response image
\begin{equation}
\calM_A^{(O)}(\bmeta_A)
=\{\bmc^{(O)}(\bmu^a;\bmeta_A):\bmu^a\in\calD\}
\subseteq\C^M
\label{eq:array-manifold}
\end{equation}
is denoted by \(\calM_A^{(O)}(\bmeta_A)\). On a constant-rank regular domain this image carries the corresponding immersed array-spatial-manifold structure; before that condition is established it is only a response image and may contain folds, self-intersections, or strata. The response point associated with satellite \(s\) is
\begin{equation}
\bmc_s^{(O)}(\bmxi;\bmeta_A)
=\bmc^{(O)}(\bmu_s^a(\bmxi);\bmeta_A).
\label{eq:manifold-point}
\end{equation}

If the response of element \(\mu\) satisfies \(c_{\mu,s}^{(O)}\neq0\), define the local reference-element coordinate
\begin{equation}
\bma_s^{[\mu]}
=\frac{\bmc_s^{(O)}}{c_{\mu,s}^{(O)}},
\qquad
a_{\mu,s}^{[\mu]}=1,
\label{eq:internal-chart}
\end{equation}
and set
\begin{equation}
\beta_{\mu,s}^{(O)}=c_{\mu,s}^{(O)},
\qquad
\bmc_s^{(O)}=\beta_{\mu,s}^{(O)}\bma_s^{[\mu]}.
\label{eq:chart-factor}
\end{equation}
Any nonzero element may serve as the internal reference element; changing \(\mu\) merely changes the local chart of the same complete response.

Under ideal narrowband far-field conditions, if elements \(m\) and \(\mu\) are at positions \(\bm b_m^a\) and \(\bm b_\mu^a\) in array coordinates, then
\begin{equation}
 a_{m,s}^{[\mu]}
=\Gamma_{m\mu,s}
\exp\!\left[-\jj\kappa(\bm b_m^a-\bm b_\mu^a)\trans\bmu_s^a\right],
\qquad \kappa=\frac{2\pi}{\lambda},
\label{eq:ideal-relative-response}
\end{equation}
where \(\Gamma_{m\mu,s}\) subsumes element patterns, polarization, mutual coupling, and calibrated channel response. Common-weight combining gives
\begin{equation}
 g_s=\bmw\herm\bmc_s^{(O)}
=\beta_{\mu,s}^{(O)}\bmw\herm\bma_s^{[\mu]}.
\label{eq:scalar-response-chart}
\end{equation}
The quantity \(\Arg(\bmw\herm\bma_s^{[\mu]})\) is the nonlinear argument of a complex superposition; in general it cannot be decomposed into two independent global scalars called a ``weight phase'' and an ``array phase.''

\begin{principle}[Non-physicality of reference selection]
An internal reference element is used only to establish relative element-response coordinates. It is not a physical channel retained separately in the PI composite output, nor is it selected automatically by an RTK differencing equation. A legitimate chart change may only redistribute a common complex factor and a normalized response; the complete element response, composite response, restored navigation phase, and RTK integer structure must not depend on the index of the internal reference element.
\label{prin:reference-nonentity}
\end{principle}

\subsection{Joint Total Space, Local Trivialization, and Legitimate Partial Derivatives}
The two state types are \(\calX\) and \(\calC_{\rm obs}\), but executability, synchronization, calibration validity, and tracking lock can constrain their joint admissibility. Consistently with \Cref{eq:dual-axis-descent}, define the joint total space
\begin{equation}
\calE\subseteq\calX\times\calC_{\rm obs},
\qquad
\pi:\calE\to\calX,
\quad
\pi(\bmxi,\bmeta)=\bmxi.
\label{eq:joint-total-space}
\end{equation}
The inclusion need not be equality: dual-axis typing does not imply that every pair \((\bmxi,\bmeta)\) is implementable.

\begin{assumption}[Local trivialization]
The admissible domain \(\calE\) is a smooth manifold and the map \(\pi:\calE\to\calX\) is a smooth surjective submersion. For every point of interest \(e_0\in\calE\), there exist a neighborhood \(U\subseteq\calX\), a typical fiber \(F\), and a local trivialization
\begin{equation}
\mathscr{T}_U:\pi^{-1}(U)\longrightarrow U\times F.
\label{eq:local-trivialization}
\end{equation}
Every occurrence of \(D_{\bmxi}\calH\) and \(D_{\bmeta}\calH\) in this paper refers to partial derivatives in a selected local trivialization. Under a change of trivialization, they transform covariantly according to the chain rule.
\label{ass:local-trivialization}
\end{assumption}

At \(e\in\calE\), the vertical subspace is
\begin{equation}
\mathsf{Ver}_e=\Ker(\dd\pi_e).
\label{eq:vertical-space}
\end{equation}
After selecting an Ehresmann horizontal distribution \cite{kobayashi1963},
\begin{equation}
T_e\calE=\mathsf{Hor}_e\oplus\mathsf{Ver}_e.
\label{eq:horizontal-vertical-split}
\end{equation}
For \(\dot e=\dot e^{\rm hor}+\dot e^{\rm ver}\), the differential of the joint observation is intrinsically
\begin{equation}
\dd\calH_e[\dot e]
=\dd\calH_e[\dot e^{\rm hor}]
+\dd\calH_e[\dot e^{\rm ver}].
\label{eq:intrinsic-total-differential}
\end{equation}
Only in local product coordinates does \Cref{eq:intrinsic-total-differential} take the familiar form
\begin{equation}
\dd\bmz
=D_{\bmxi}\calH\,\dd\bmxi
+D_{\bmeta}\calH\,\dd\bmeta.
\label{eq:local-total-differential}
\end{equation}
Thus, the first-order physical/observer decomposition depends on an explicit local comparison rule; it is not a canonical global direct-sum decomposition on an arbitrary total space.

For a fixed local observer coordinate \(\bmeta_0\in F\), define the local section
\begin{equation}
\sigma_{\bmeta_0}(\bmxi)=\mathscr{T}_U^{-1}(\bmxi,\bmeta_0),
\end{equation}
and the restricted mapping
\begin{equation}
\calH_{\bmeta_0}=\calH\circ\sigma_{\bmeta_0}:U\to\calZ.
\label{eq:restricted-local-map}
\end{equation}
The statement ``vary the physical state while holding the observer state fixed'' is defined by this local section and is not misrepresented as an unconditional global comparison between different fibers.

\subsection{Exact Correlation Integral and Conditions for a Single Complex-Scalar Model}
Before correlation, the element-level complex baseband signal may be written as
\begin{equation}
\bm x(t)
=\sum_{s\in\calS}
A_s(t)\ee^{\jj\phi_s^{\nav}(O,t)}
\bmc_s^{(O)}(t)d_s(t-\tau_s)
+\bm j(t)+\bm n(t),
\label{eq:array-signal}
\end{equation}
where \(d_s\) denotes known code and data modulation, \(\bm j\) is interference, and multipath and unmodeled satellite terms may be included in the effective residual.

For the \(n\)th coherent integration interval \(\calI_n\) of satellite \(s\), let \(\bmw(t)\) and \(q_{s,n}(t)\) denote the weight vector and correlation kernel. The exact correlation output is
\begin{equation}
 z_{s,n}
=\int_{\calI_n}
q_{s,n}^{*}(t)\bmw(t)\herm\bm x(t)\,\dd t.
\label{eq:exact-correlator}
\end{equation}
After data-bit or pilot wipeoff, denote the residual navigation phase of the desired satellite by \(\delta\phi_s^{\nav}(t)\), and define
\begin{equation}
F_s(t)
=A_s(t)\ee^{\jj\delta\phi_s^{\nav}(t)}g_s(t),
\qquad
g_s(t)=\bmw(t)\herm\bmc_s^{(O)}(t).
\label{eq:window-composite}
\end{equation}
If, about the window center \(t_n\), the correlation kernel satisfies
\begin{equation}
 m_0=\int_{\calI_n}q_{s,n}^{*}(t)\,\dd t,
\qquad
m_1=\int_{\calI_n}(t-t_n)q_{s,n}^{*}(t)\,\dd t=0,
\label{eq:window-moments}
\end{equation}
then a Taylor expansion gives
\begin{equation}
 z_{s,n}=m_0F_s(t_n)+n_{s,n}^{\rm eff}+r_{s,n}^{\rm frz},
\label{eq:frozen-scalar-model}
\end{equation}
with deterministic remainder bound
\begin{equation}
\abs{r_{s,n}^{\rm frz}}
\le
\frac12
\left(\int_{\calI_n}\abs{q_{s,n}(t)}\abs{t-t_n}^2\,\dd t\right)
\sup_{t\in\calI_n}\abs{\ddot F_s(t)}.
\label{eq:freeze-remainder-bound}
\end{equation}
Thus, the commonly used single complex-scalar model is not automatically valid under arbitrary rapid weight variation; it is an approximation to \Cref{eq:exact-correlator} under a coherent-window freezing condition. A convenient first-order dimensionless engineering audit index is
\begin{equation}
\epsilon_{\rm frz}
=T_{\rm coh}
\sup_{t\in\calI_n}
\left(
\abs{\frac{\dot A_s}{A_s}}
+\abs{\dot{\delta\phi}_s^{\nav}}
+\abs{\frac{\dot g_s}{g_s}}
\right)
\ll1.
\label{eq:freeze-index}
\end{equation}

\begin{assumption}[Front-end linearity and the boundary of recoverable information]
At the interference levels under study, active antenna components, RF/IF chains, AGC, ADC, and digital array channels do not undergo unmodeled saturation, clipping, or irreversible quantization distortion. If information has already been lost through nonlinearity before weighting, a subsequent \(\Uone\) inverse action cannot recover it.
\label{ass:front-end-linearity}
\end{assumption}

\begin{assumption}[Narrowband or suitably bandwidth-resolved processing]
The single-carrier conclusions apply to narrowband channels representable by \(\bmc_s^{(O)}\). Under wideband beam squint, group delay, code-correlation-shape distortion, near-field spherical waves, or distributed scattering, frequency and delay must enter \(\calH\); the carrier-phase results in this paper do not replace broadband propagation analysis.
\label{ass:narrowband}
\end{assumption}

\subsection{Non-Identifiability and Closure of the Composite Additional Phase}
\begin{theorem}[Non-identifiability of the absolute composite phase]
In the noiseless same-source model
\begin{equation}
\bmz_s=\alpha_s^{\nav}\bmg_s,
\label{eq:bilinear-factor}
\end{equation}
let \(\alpha_s^{\nav}\in\C^\times\) and \(\bmg_s\in\C^K\setminus\{\bm0\}\). Suppose the admissible factorization is closed under the \(\C^\times\) action below and no model, calibration, or informative absolute prior restricts that action. Then, for every admissible \(a\in\C^\times\), the transformation
\begin{equation}
(\alpha_s^{\nav},\bmg_s)
\longmapsto
(\alpha_s^{\nav}a^{-1},a\bmg_s)
\label{eq:factor-ambiguity}
\end{equation}
leaves \(\bmz_s\) unchanged. Under this common-complex-scale auxiliary model, the data identify only the projective class \([\bmg_s]\in\CP^{K-1}\), not the common complex scale of \(\bmg_s\). A model, prior, or calibration that breaks the action lies outside this non-identifiability statement and may provide absolute closure.
\label{thm:absolute-unidentifiable}
\end{theorem}
\begin{proof}
The transformation in \Cref{eq:factor-ambiguity} leaves the bilinear product unchanged; its \(\C^{\times}\) orbit is precisely the unidentifiable freedom.
\end{proof}

Knowledge of only the normalized response \(\bma_s^{[\mu]}\) is likewise insufficient for absolute composite-phase restoration, because \(\bmc_s^{(O)}=\beta_{\mu,s}^{(O)}\bma_s^{[\mu]}\), and the phase of \(\beta_{\mu,s}^{(O)}\) contains the common array phase relative to the physical reference point.

Absolute closure must be supplied by at least one of the following sources:
\begin{enumerate}
\item \textbf{Model closure:} compute \(g_s\) from calibrated \(\bmc_s^{(O)}\) and the current weights;
\item \textbf{Calibrated-branch closure:} know the composite phase of a same-source anchor branch relative to \(O\), and transport it to the working branch through a relative measurement;
\item \textbf{Joint-state closure:} treat the absolute composite phase or its continuous lift as a dynamical state and estimate it jointly from a model and an informative absolute initial prior or anchor, with relative observations propagating subsequent evolution. Relative data and smoothness alone leave an unknown constant phase.
\end{enumerate}
None of these approaches requires an individual raw reference element to track the carrier independently under strong interference.

\subsection{Basic Assumptions and the Minimal Theoretical Skeleton}
In addition to the preceding assumptions, the following regularity conditions are used throughout.

\begin{assumption}[Reference coordinates and nonzero response]
Using the chart associated with internal reference element \(\mu\) requires \(c_{\mu,s}^{(O)}\neq0\); defining the composite phase requires \(g_s\neq0\). The former is a chart condition, whereas the latter is a physical nondegeneracy condition of the output.
\label{ass:chart-and-output}
\end{assumption}

\begin{assumption}[Smoothness]
On the regular domain, \(\bmc^{(O)}\), the weight trajectory, and the joint observation map are at least continuously differentiable. Whenever curvature, Hessians, or freezing remainders are used, the required second derivatives exist.
\label{ass:smoothness}
\end{assumption}

\begin{assumption}[Synchronized same-source branches]
Whenever ``direct same-source cancellation'' is claimed, the parallel auxiliary branches for a given satellite share the same element-sample identity, clock, code/carrier NCOs, data-bit wipeoff, and coherent integration operator. Independent local oscillators, unknown branch phases, or uncalibratable delays must enlarge both the nuisance group and the statistical model.
\label{ass:same-source}
\end{assumption}

\begin{assumption}[Conditional weight statistics]
The local proper-complex-Gaussian noise, covariance, and Fisher-information formulas are conditioned on a given weight set, or require the data used to estimate the weights to be independent of the current inference block. If adaptive weights are estimated from the same noisy, interfered data block, a joint model \(p(\bmz,\widehat{\bm W}\mid\bmxi,\bmeta)\) must be used, or the weight errors and their cross-covariances must be propagated explicitly.
\label{ass:conditional-weight}
\end{assumption}

\begin{definition}[SN-ASMO observation-geometric system]
A regular local SN-ASMO system is denoted by
\begin{equation}
\calG_{\rm SN}=(\calX,\calE,\pi,\calZ,\calH),
\label{eq:sn-system}
\end{equation}
and is supplemented, according to the task, by a physical reference point \(O\), an internal-reference atlas, an array-response mapping, the genuine nuisance group \(G_{\rm nui}\), a cross-state comparison rule, singular sets, and a statistical model.
\end{definition}

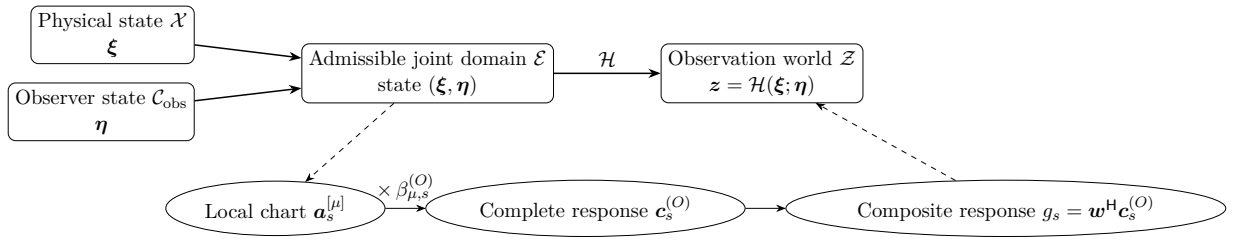
\begin{figure}[htbp]
\centering
\resizebox{0.98\textwidth}{!}{%
\begin{tikzpicture}[node distance=1.75cm and 1.45cm,>=Stealth]
\node[draw,rounded corners,fill=white,minimum width=3.4cm,minimum height=1.05cm,align=center] (e) {Admissible joint domain $\calE$\\state $(\bmxi,\bmeta)$};
\node[draw,rounded corners,fill=white,minimum width=3.0cm,minimum height=0.95cm,left=2.0cm of e,yshift=0.72cm,align=center] (x) {Physical state $\calX$\\$\bmxi$};
\node[draw,rounded corners,fill=white,minimum width=3.0cm,minimum height=0.95cm,left=2.0cm of e,yshift=-0.72cm,align=center] (s) {Observer state $\calC_{\rm obs}$\\$\bmeta$};
\node[draw,rounded corners,fill=white,minimum width=3.1cm,minimum height=1.05cm,right=2.0cm of e,align=center] (z) {Observation world $\calZ$\\$\bmz=\calH(\bmxi;\bmeta)$};
\draw[->,thick] (x) -- (e);
\draw[->,thick] (s) -- (e);
\draw[->,thick] (e) -- node[above] {$\calH$} (z);
\node[draw,ellipse,below=1.45cm of e,xshift=-2.85cm] (a) {Local chart $\bma_s^{[\mu]}$};
\node[draw,ellipse,right=0.75cm of a] (c) {Complete response $\bmc_s^{(O)}$};
\node[draw,ellipse,right=0.75cm of c] (g) {Composite response $g_s=\bmw\herm\bmc_s^{(O)}$};
\draw[->] (a) -- node[above] {$\times\,\beta_{\mu,s}^{(O)}$} (c);
\draw[->] (c) -- (g);
\draw[dashed,->] (e) -- (a);
\draw[dashed,->] (g) -- (z);
\end{tikzpicture}%
}
\caption{Dual-axis SN-ASMO formation chain. Physical and observer states jointly enter the admissible domain; the observation world is the output. The internal reference element changes only local coordinates, not the complete response, composite response, or observation.}
\label{fig:object-chain}
\end{figure}
\section[Complete Observation, Reference Covariance, and Composite Phase]{Complete Observation, Reference Covariance,\\and Preservation of the Composite Phase}
\subsection{Complete Single-Satellite Signal and Composite Amplitude--Phase Decomposition}
With respect to the physical reference point \(O\), the narrowband coherent signal of satellite \(s\) at the element level is
\begin{equation}
\bm x_s(t)
=A_s(t)\ee^{\jj\phi_s^{\nav}(O,t)}\bmc_s^{(O)}(t)+\bmn_s(t).
\label{eq:complete-array-signal}
\end{equation}
After common-weight combining, and neglecting the coherent-integration freezing remainder, one obtains
\begin{equation}
 z_s(t)
=A_s(t)\ee^{\jj\phi_s^{\nav}(O,t)}g_s(t)+n_s^{\rm out}(t),
\qquad
g_s(t)=\bmw(t)\herm\bmc_s^{(O)}(t).
\label{eq:complete-scalar-output}
\end{equation}
The internal reference element provides only the factorization
\begin{equation}
 g_s=\beta_{\mu,s}^{(O)}\gamma_{\mu,s},
\qquad
\gamma_{\mu,s}=\bmw\herm\bma_s^{[\mu]},
\label{eq:local-response-factorization}
\end{equation}
and does not alter \(g_s\) itself.

On the regular domain where \(g_s\neq0\), its unique polar decomposition is
\begin{equation}
 g_s=\varrho_s^A h_s^A,
\qquad
\varrho_s^A=\abs{g_s}>0,
\qquad
 h_s^A=\frac{g_s}{\abs{g_s}}=\ee^{\jj\psi_s^A}\in\Uone.
\label{eq:polar-composite-response}
\end{equation}
For every \(z\in\C^\times\), define its phase element by
\begin{equation}
\ph(z):=\frac{z}{\abs{z}}\in\Uone.
\label{eq:phase-element-definition}
\end{equation}
Here, \(\psi_s^A\) is the principal-value phase. Along a continuous trajectory that does not cross a zero, its real-valued continuous lift is denoted by \(\widetilde\psi_s^A\). The continuous phase of the noiseless output satisfies
\begin{equation}
\widetilde\phi_s^{\obs}
=\widetilde\phi_s^{\nav}(O)+\widetilde\psi_s^A.
\label{eq:observed-phase-sum}
\end{equation}
This definition reduces all element responses and weight superposition to one final composite additional phase and does not misinterpret intermediate coordinate factors as multiple independently compensable physical phases.

\subsection{Exact and Estimated Phases Must Be Treated Separately}
\begin{theorem}[Inverse action of the exact composite phase]
Suppose \(g_s\neq0\) and the phase element \(h_s^A\) used for inverse transport equals the true composite-response phase. Define
\begin{equation}
\widetilde z_s=(h_s^A)^{-1}z_s.
\label{eq:exact-inverse-action}
\end{equation}
Then
\begin{equation}
\widetilde z_s
=A_s\varrho_s^A\ee^{\jj\phi_s^{\nav}(O)}
+(h_s^A)^{-1}n_s^{\rm out}.
\label{eq:exact-restored-output}
\end{equation}
Thus, the final composite additional phase in the noiseless component is removed exactly, while the normal navigation propagation phase is preserved. For every noise sample, the unit-modulus rotation preserves its absolute value. In the vector case, a deterministic diagonal unitary transformation preserves the covariance eigenvalues, trace, and rank.
\label{thm:exact-inverse}
\end{theorem}
\begin{proof}
Substitute \(g_s=\varrho_s^Ah_s^A\) into \Cref{eq:complete-scalar-output}. The conclusion depends only on the complete composite response and not on its factorization through an internal reference element.
\end{proof}

\begin{theorem}[Residual propagation with an estimated composite phase]
If the phase element actually used is
\begin{equation}
\widehat h_s^A=h_s^A\ee^{\jj\varepsilon_{h,s}},
\label{eq:estimated-phase-element}
\end{equation}
then
\begin{equation}
(\widehat h_s^A)^{-1}z_s
=A_s\varrho_s^A
\ee^{\jj(\phi_s^{\nav}-\varepsilon_{h,s})}
+(\widehat h_s^A)^{-1}n_s^{\rm out}.
\label{eq:estimated-restoration}
\end{equation}
The signal retains the multiplicative phase error \(-\varepsilon_{h,s}\). If \(\widehat h_s^A\) is estimated from the same noisy data, then \(\varepsilon_{h,s}\) and \(n_s^{\rm out}\) are generally correlated, so the inverse-transport residual cannot be interpreted unconditionally as an independent unit-modulus rotation of the original noise.
\label{thm:estimated-inverse}
\end{theorem}
\begin{proof}
Substitute \((\widehat h_s^A)^{-1}=\ee^{-\jj\varepsilon_{h,s}}(h_s^A)^{-1}\) into \Cref{eq:complete-scalar-output} and use \(g_s=\varrho_s^Ah_s^A\).
\end{proof}

\Cref{thm:exact-inverse,thm:estimated-inverse} distinguish the exact physical limit from the conditions of an implementable estimate. Every invariance claim in the remainder of the paper states whether the corresponding group element is a known deterministic quantity or a random estimate derived from data.

\subsection{Internal-Reference-Element Chart Changes and Invariance of the Composite Response}
When changing from the chart associated with element \(\mu\) to that associated with element \(\nu\), assuming \(a_{\nu,s}^{[\mu]}\neq0\),
\begin{equation}
\bma_s^{[\nu]}
=\frac{\bma_s^{[\mu]}}{a_{\nu,s}^{[\mu]}},
\label{eq:chart-change-a}
\end{equation}
\begin{equation}
\beta_{\nu,s}^{(O)}
=\beta_{\mu,s}^{(O)}a_{\nu,s}^{[\mu]},
\qquad
\gamma_{\nu,s}
=\frac{\gamma_{\mu,s}}{a_{\nu,s}^{[\mu]}}.
\label{eq:chart-change-factors}
\end{equation}

\begin{theorem}[Covariance under an internal-reference-element change and invariance of the composite output]
Under \Cref{eq:chart-change-a,eq:chart-change-factors},
\begin{equation}
\beta_{\nu,s}^{(O)}\bma_s^{[\nu]}
=\beta_{\mu,s}^{(O)}\bma_s^{[\mu]}
=\bmc_s^{(O)},
\label{eq:full-response-invariant}
\end{equation}
\begin{equation}
\beta_{\nu,s}^{(O)}\gamma_{\nu,s}
=\beta_{\mu,s}^{(O)}\gamma_{\mu,s}
=g_s.
\label{eq:scalar-response-invariant}
\end{equation}
Consequently, the composite magnitude, composite additional phase, pre-discriminator restoration quantity, and RTK integer structure are independent of the internal-reference-element index.
\label{thm:chart-invariance}
\end{theorem}
\begin{proof}
The result follows by direct multiplication. The transition functions among any three valid charts satisfy the cocycle condition, so repeated chart changes still represent the same complete response.
\end{proof}

Taking the complex-logarithmic differential of \(g_s=\beta_{\mu,s}^{(O)}\gamma_{\mu,s}\) yields
\begin{equation}
\frac{\dd g_s}{g_s}
=\frac{\dd\beta_{\mu,s}^{(O)}}{\beta_{\mu,s}^{(O)}}
+\frac{\dd\gamma_{\mu,s}}{\gamma_{\mu,s}}.
\label{eq:log-factorization}
\end{equation}
The two terms on the right change separately under a chart transition, but their sum does not. Thus, the composite amplitude--phase dynamics are intrinsic, whereas the individual factor terms are not.

\subsection{Physical-Reference-Point Covariance and the Antenna Phase-Center Interface}
If the physical reference point is translated from \(O\) to \(O'=O+\Delta\bm r\), then under a narrowband far-field plane-wave model with the same amplitude normalization and channel convention, the complete element response undergoes the common pure-phase transformation
\begin{equation}
\bmc_s^{(O')}
=\ee^{\jj\chi_s}\bmc_s^{(O)},
\qquad
\chi_s=\kappa(\bmu_s^n)\trans\Delta\bm r.
\label{eq:physical-reference-response}
\end{equation}
The paired relation necessary and sufficient to preserve the complete noiseless element signal is
\begin{equation}
\phi_s^{\nav}(O')
=\phi_s^{\nav}(O)-\chi_s
\pmod{2\pi}.
\label{eq:physical-reference-phase}
\end{equation}
Therefore,
\begin{equation}
\ee^{\jj\phi_s^{\nav}(O')}\bmc_s^{(O')}
=\ee^{\jj\phi_s^{\nav}(O)}\bmc_s^{(O)}.
\label{eq:physical-reference-invariant}
\end{equation}
A physical-reference-point transformation is a covariance relation between the propagation phase and the complete array response. It is neither an internal-reference-element chart change nor a common-oscillator phase nuisance.

The IGS ANTEX model uses the mechanical antenna reference point and antenna reference frame as its geometric basis and represents the direction-dependent carrier-range correction as the sum of a phase-center offset and phase variation \cite{antex2026}. The interface with SN-ASMO is as follows. Conventional PCO/PCV models describe a fixed or piecewise calibrated directional antenna response, whereas \(g_s=\bmw\herm\bmc_s^{(O)}\) additionally permits weights, channels, and attitude states to participate in the composite response. The dynamic equivalent phase center defined later is a local first-order interpretation of the composite phase. It neither replaces ANTEX calibration nor permits unmodeled weight dynamics to be disguised as a static PCO/PCV correction.

In broadband or near-field conditions, or whenever a reference-point translation produces non-negligible amplitude and delay structure, the pure-phase relation in \Cref{eq:physical-reference-response} must be replaced by a general nonzero complex operator or frequency--delay response.

\subsection{Identity Separation among Three Reference Transformations}
Three common transformations act on different objects:
\begin{enumerate}
\item \textbf{Change of internal-reference-element chart:} a \(\C^{\times}\) transition function redistributes \(\beta_{\mu,s}^{(O)}\) and \(\bma_s^{[\mu]}\), while the complete \(\bmc_s^{(O)}\) and \(g_s\) remain strictly invariant;
\item \textbf{Change of physical reference point:} the satellite-dependent phase \(\chi_s\) is paired between the propagation factor and complete array response, leaving the complete physical signal invariant;
\item \textbf{Change of common phase coordinates:} a diagonal \(\Uone\) action is applied to same-source branches of the same satellite and may constitute a nuisance in an auxiliary relative-response task.
\end{enumerate}
Calling all three transformations merely a ``reference-phase change'' conceals their different groups, action objects, and observable meanings.

\subsection{Typed Paired Phase-Coordinate Transformation and the Complete Invariant}
The noiseless single-satellite observation is the product \(z_s^{(0)}=\alpha_s^{\nav}g_s\). A phase-coordinate change that preserves the declared navigation identity uses an admissible physical-reference function \(\chi=\chi(\bmxi)\), independent of observer-state coordinates, and acts as a pair:
\begin{equation}
(\alpha_s^{\nav},g_s)
\longmapsto
\bigl(\ee^{-\jj\chi(\bmxi)}\alpha_s^{\nav},
\ee^{\jj\chi(\bmxi)}g_s\bigr),
\label{eq:paired-gauge}
\end{equation}
so that \(z_s^{(0)}\) remains strictly invariant. Multiplying only \(g_s\) by \(\ee^{\jj\chi}\) without changing the propagation-phase coordinate changes the physical observation and is not a pure coordinate transformation. An arbitrary \(\chi(\bmxi,\bmeta)\) still yields an algebraic refactorization of the product, but its transformed first factor generally depends on \(\bmeta\) and may no longer be called the pure navigation-propagation factor.

Define the propagation-phase one-form and the array-response-induced one-form by
\begin{equation}
\omega_s^{\nav}=\dd\widetilde\phi_s^{\nav},
\qquad
\omega_s^A=\Ima\!\left(\frac{\dd g_s}{g_s}\right).
\label{eq:paired-one-forms}
\end{equation}
Under \Cref{eq:paired-gauge},
\begin{equation}
(\omega_s^{\nav})'
=\omega_s^{\nav}-\dd\chi,
\qquad
(\omega_s^A)'
=\omega_s^A+\dd\chi,
\label{eq:paired-form-transform}
\end{equation}
and hence the complete differential phase
\begin{equation}
\omega_s^{\obs}=\omega_s^{\nav}+\omega_s^A
\label{eq:total-phase-invariant}
\end{equation}
is gauge invariant. SN-ASMO preserves the complete physical content represented by \Cref{eq:physical-reference-invariant,eq:total-phase-invariant}, not the numerical value of any one decomposition coordinate.
\section{G0: Complete Complex Response, Provenance, Regularity, and the Safe Domain}
\subsection{Exact Total Differential of the Composite Response}
For
\begin{equation}
 g_s=\bmw\herm\bmc_s^{(O)},
\end{equation}
the exact total differential is
\begin{equation}
\dd g_s=(\dd\bmw)\herm\bmc_s^{(O)}
+\bmw\herm\dd\bmc_s^{(O)}.
\label{eq:dg-exact}
\end{equation}
The first term arises from a change in the weights and the second from a change in the complete element-level response. This is a decomposition of local differential contributions, not a global decomposition of the final composite phase into two independent scalar phases.

If
\begin{equation}
\bmc_s^{(O)}
=\bmc^{(O)}(\bmu_s^a;\bmeta_A),
\end{equation}
then, in local coordinates,
\begin{equation}
\dd\bmc_s^{(O)}
=\bm J_{u,s}\,\delta\bmu_s^a
+\bm J_{A,s}\,\delta\bmeta_A,
\label{eq:dc-local}
\end{equation}
where
\begin{equation}
\bm J_{u,s}=\frac{\partial\bmc_s^{(O)}}{\partial(\bmu_s^a)\trans},
\qquad
\bm J_{A,s}=\frac{\partial\bmc_s^{(O)}}{\partial\bmeta_A\trans}.
\label{eq:response-jacobians}
\end{equation}
The matrix \(\bm J_{u,s}\) is an ambient-space Jacobian; a physically admissible direction perturbation must lie in \(T_{\bmu_s^a}S^2\).

\subsection{Tangent-Space Constraints on \texorpdfstring{$S^2$}{S2} and \texorpdfstring{$\SOthree$}{SO(3)}}
Since \(\norm{\bmu_s^a}=1\),
\begin{equation}
(\bmu_s^a)\trans\delta\bmu_s^a=0,
\qquad
\delta\bmu_s^a\in T_{\bmu_s^a}S^2.
\label{eq:sphere-tangent}
\end{equation}
Define the spherical tangent projector
\begin{equation}
\bm P_{u,s}=\bm I-\bmu_s^a(\bmu_s^a)\trans.
\label{eq:sphere-projector}
\end{equation}
When ambient coordinates are used, the legitimate direction Jacobian is
\begin{equation}
\bm J_{u,s}^{\rm tan}=\bm J_{u,s}\bm P_{u,s}.
\label{eq:tangent-direction-jacobian}
\end{equation}
Treating the three components of \(\bmu_s^a\) as three independent degrees of freedom invents a radial direction and overestimates both local rank and Fisher information.

The following attitude-perturbation convention is adopted. The matrix \(\bm C_a^n\) maps array-frame coordinates to navigation-frame coordinates, and a navigation-frame small rotation \(\delta\bmtheta^n\) acts as a left perturbation:
\begin{equation}
\bm C_a^n(\delta\bmtheta^n)
=\exp\!\bigl(\crossmat{\delta\bmtheta^n}\bigr)\bm C_a^n.
\label{eq:attitude-left-perturb}
\end{equation}
Consequently,
\begin{equation}
\delta\bm C_n^a
=-\bm C_n^a\crossmat{\delta\bmtheta^n},
\label{eq:inverse-attitude-perturb}
\end{equation}
and, from \(\bmu_s^a=\bm C_n^a\bmu_s^n\),
\begin{equation}
\delta\bmu_s^a
=\bm C_n^a\crossmat{\bmu_s^n}\delta\bmtheta^n
+\bm C_n^a\delta\bmu_s^n.
\label{eq:los-attitude-chain}
\end{equation}
Letting \(\delta\bmtheta^a=\bm C_n^a\delta\bmtheta^n\), equivalently,
\begin{equation}
\delta\bmu_s^a
=\crossmat{\bmu_s^a}\delta\bmtheta^a
+\bm C_n^a\delta\bmu_s^n.
\label{eq:los-attitude-array}
\end{equation}
Equations~\Cref{eq:los-attitude-chain,eq:los-attitude-array} give the common chain rule for attitude and line-of-sight variations. Other left- or right-perturbation conventions change the corresponding signs, but different conventions must not be mixed within one derivation.

Substitution into \Cref{eq:dg-exact,eq:dc-local} gives
\begin{equation}
\dd g_s
=(\dd\bmw)\herm\bmc_s^{(O)}
+\bmw\herm\bm J_{u,s}\delta\bmu_s^a
+\bmw\herm\bm J_{A,s}\delta\bmeta_A.
\label{eq:dg-expanded}
\end{equation}
The actual additional-phase differential will be extracted uniformly in G1 through \(\Ima(\dd g_s/g_s)\).

\subsection{Composite-Phase Condition Numbers and Rate Bounds}
For a small perturbation
\begin{equation}
\delta g_s
=(\delta\bmw)\herm\bmc_s^{(O)}
+\bmw\herm\delta\bmc_s^{(O)},
\end{equation}
when \(g_s\neq0\) and the first-order approximation is valid,
\begin{equation}
\delta\psi_s^A
\approx\Ima\!\left(\frac{\delta g_s}{g_s}\right),
\label{eq:phase-first-order}
\end{equation}
which gives the deterministic bound
\begin{equation}
\abs{\delta\psi_s^A}
\le
\frac{\norm{\bmc_s^{(O)}}\norm{\delta\bmw}
+\norm{\bmw}\norm{\delta\bmc_s^{(O)}}}{\abs{g_s}}
+O(\norm{\delta}^2).
\label{eq:phase-conditioning-bound}
\end{equation}
Define
\begin{equation}
\kappa_{w,s}=\frac{\norm{\bmc_s^{(O)}}}{\abs{g_s}},
\qquad
\kappa_{c,s}=\frac{\norm{\bmw}}{\abs{g_s}}.
\label{eq:phase-condition-numbers}
\end{equation}
Then
\begin{equation}
\abs{\delta\psi_s^A}
\lesssim
\kappa_{w,s}\norm{\delta\bmw}
+\kappa_{c,s}\norm{\delta\bmc_s^{(O)}}.
\label{eq:phase-condition-summary}
\end{equation}
The quantities \(\kappa_{w,s}\) and \(\kappa_{c,s}\) are local deterministic sensitivities, not statistical standard deviations. Under a random-error model, the corresponding Jacobians must be combined with covariance matrices to obtain RMS quantities.

Along a time trajectory,
\begin{equation}
\abs{\dot{\widetilde\psi}_s^A}
\le
\frac{\norm{\dot{\bmw}}\norm{\bmc_s^{(O)}}
+\norm{\bmw}\norm{\dot{\bmc}_s^{(O)}}}{\abs{g_s}}.
\label{eq:phase-rate-bound}
\end{equation}
This quantitative bound shows that the condition \(\abs{g_s}\ge\varepsilon_s\) alone is not sufficient. A large-norm superdirective weight vector can strongly amplify calibration errors and weight perturbations even before the response approaches an exact null. A safe domain should constrain response magnitude, weight norm, phase condition numbers, and admissible update rate simultaneously.

\subsection{Constant-Rank Local Response Geometry}
For a fixed local observer coordinate \(\bmeta_0\), consider the restricted response map
\begin{equation}
\calH_{A,\bmeta_0}:U\subseteq\calX\to\calZ_A,
\qquad
\calH_{A,\bmeta_0}(\bmxi)=\bmg(\bmxi;\bmeta_0).
\label{eq:restricted-array-map}
\end{equation}

\begin{theorem}[Constant-rank local response theorem]
If \(\calH_{A,\bmeta_0}\) has constant rank \(r\) on a neighborhood of \(\bmxi_0\), then there exist local coordinates on its domain and codomain in which
\begin{equation}
\calH_{A,\bmeta_0}(v_1,\ldots,v_n)
=(v_1,\ldots,v_r,0,\ldots,0).
\label{eq:constant-rank-normal-form}
\end{equation}
The response image therefore has an \(r\)-dimensional immersed-submanifold structure near \(\calH_{A,\bmeta_0}(\bmxi_0)\), and the local equal-response set has dimension \(n-r\).
\label{thm:constant-rank}
\end{theorem}
\begin{proof}
The result follows directly from the constant-rank theorem \cite{lee2013}.
\end{proof}
Rank \(r\) at a single point does not imply a stable local structure; the rank must remain constant on a neighborhood. The global response image may self-intersect, fold, or form a stratified set. Accordingly, the term ``response manifold'' is used only on constant-rank regular domains.

\subsection{Singular Sets and Their Physical Consequences}
Let \(e=(\bmxi,\bmeta)\in\calE\). Define five classes of singular sets:
\begin{align}
\Sigma_s^0&=\{e:g_s(e)=0\},
\label{eq:zero-singularity}\\
\Sigma^{\rm rank}&=\{e:\rank\dd\calH_A(e)<r_{\max}\},
\label{eq:rank-singularity}\\
\Sigma_r^{\rm chart}&=\{e:g_{s,r}(e)=0\},
\label{eq:chart-singularity}\\
\Sigma^G&=\{e:\text{the nuisance-orbit dimension or stabilizer type changes}\},
\label{eq:group-singularity}\\
\Sigma^{\rm stat}&=\{e:\text{likelihood support, score, covariance, or Fisher regularity fails}\}.
\label{eq:stat-singularity}
\end{align}
Their physical meanings are as follows:
\begin{enumerate}
\item \(\Sigma_s^0\): the composite response for satellite \(s\) cancels completely; \(\Arg g_s\), \(\dd g_s/g_s\), and the continuous lift are undefined, and carrier tracking may lose lock;
\item \(\Sigma^{\rm rank}\): the observation loses first-order sensitivity to some state directions, reducing the rank of the intrinsic Jacobian and effective Fisher information;
\item \(\Sigma_r^{\rm chart}\): a denominator of a ratio coordinate vanishes; if the joint response vector remains nonzero, another chart can be used, so this is not a physical failure;
\item \(\Sigma^G\): the orbit type of the genuine nuisance group changes, and the unified quotient can become a stratified space;
\item \(\Sigma^{\rm stat}\): the geometric response may remain regular, but a Gaussian approximation or Fisher-regular model fails and must be replaced by a nonregular, heavy-tailed, or full-likelihood treatment.
\end{enumerate}
These sets must not be conflated under the generic label ``signal failure.''

Let \(\calS_{\rm tr}\) be the tracked-satellite set and define
\begin{equation}
\Sigma^0=\bigcup_{s\in\calS_{\rm tr}}\Sigma_s^0,
\end{equation}
\begin{equation}
\calD_{\rm reg}
=\calE\setminus
(\Sigma^0\cup\Sigma^{\rm rank}\cup\Sigma^G\cup\Sigma^{\rm stat}),
\label{eq:regular-domain}
\end{equation}
\begin{equation}
\calD_{\varepsilon}
=\{e:\abs{g_s(e)}\ge\varepsilon_s,\ \forall s\in\calS_{\rm tr}\},
\label{eq:amplitude-safe-domain}
\end{equation}
and the composite safe domain
\begin{equation}
\calD_{\rm safe}
=\calD_{\rm reg}\cap\calD_{\varepsilon}
\cap\{e:\kappa_{w,s}\le\bar\kappa_{w,s},\ \kappa_{c,s}\le\bar\kappa_{c,s}\}.
\label{eq:full-safe-domain}
\end{equation}

For \(z_s=\alpha_s^{\nav}g_s+n_s\) with \(n_s\sim\CN(0,\sigma_s^2)\), at high SNR with known magnitude or after local elimination of magnitude, the phase variance has the typical lower bound
\begin{equation}
\Var(\widehat\phi_s)
\gtrsim\frac{\sigma_s^2}{2\abs{\alpha_s^{\nav}g_s}^2}.
\label{eq:phase-variance-near-zero}
\end{equation}
Phase information therefore degenerates as \(\abs g\to0\). The domain \(\calD_{\rm safe}\) is simultaneously an interface for phase definability, reliable linearization, finite variance, and stable tracking.

\subsection{Continuous Phase Lifts, Winding on Closed Loops, and Global Existence}
Over any time interval, a nonzero continuous response always admits a continuous real-valued phase lift. On a multidimensional state domain \(D\subseteq\calD_{\rm safe}\), however, nonvanishing and connectedness alone do not guarantee a globally single-valued phase.

\begin{theorem}[Criterion for a global real-valued phase lift]
Let \(D\) be a connected smooth domain and \(g:D\to\C^{\times}\) continuously differentiable. A global continuous function \(\widetilde\psi:D\to\R\) satisfying \(g/\abs g=\ee^{\jj\widetilde\psi}\) exists if and only if, for every closed loop \(\gamma\subset D\),
\begin{equation}
\frac{1}{2\pi}
\oint_{\gamma}
\Ima\!\left(\frac{\dd g}{g}\right)=0.
\label{eq:global-lift-condition}
\end{equation}
If \(D\) is simply connected and \(g\neq0\), the condition holds automatically.
\label{thm:global-lift}
\end{theorem}
The covering-space proof is given in \Cref{subsec:lift-proof}.

The winding number of a general closed loop is
\begin{equation}
 k_\gamma
=\frac{1}{2\pi\jj}
\oint_\gamma\frac{\dd g}{g}
\in\Z.
\label{eq:winding-number}
\end{equation}
The endpoint ratio of the normalized group element around the loop is \(\ee^{\jj2\pi k_\gamma}=1\). Winding information therefore resides in the real-valued continuous lift and integer cycle count, not in a nontrivial \(\Uone\) endpoint group element. If the path crosses \(g=0\), a continuous lift generally ceases to exist.
\section{Response-Normalized Dynamics and the Dynamic Equivalent Phase Center}
\subsection{Complex-Logarithmic Response One-Form}
On a local region where \(g_s\neq0\), define the complex-logarithmic response one-form
\begin{equation}
\calC_s^A
=\frac{\dd g_s}{g_s}
=\dd\ln\varrho_s^A
+\jj\dd\widetilde\psi_s^A.
\label{eq:complex-log-form}
\end{equation}
Hence,
\begin{equation}
\Rea\calC_s^A=\dd\ln\varrho_s^A,
\qquad
\Ima\calC_s^A=\dd\widetilde\psi_s^A.
\label{eq:log-amplitude-phase}
\end{equation}
Compared with the unnormalized differential \(\dd g_s\), the ratio \(\dd g_s/g_s\) describes scale-normalized relative variations in amplitude and phase. By \Cref{eq:log-factorization}, it is invariant under a change of internal-reference-element chart.

Within a local trivialization,
\begin{equation}
\calC_s^A
=\frac{D_{\bmxi}g_s\,\dd\bmxi}{g_s}
+\frac{D_{\bmeta}g_s\,\dd\bmeta}{g_s}.
\label{eq:log-state-observer-split}
\end{equation}
The first term is the action of the physical state on the composite response through the element-response channel, whereas the second is the action of observer-state variation. The ordinary propagation phase is not contained in \(g_s\); it is carried separately by \(\alpha_s^{\nav}\).

Under the linear narrowband model,
\begin{equation}
\calC_s^A
=\frac{(\dd\bmw)\herm\bmc_s^{(O)}}{g_s}
+\frac{\bmw\herm\dd\bmc_s^{(O)}}{g_s}.
\label{eq:log-two-source}
\end{equation}
The real and imaginary parts give, respectively, the first-order variations of the composite amplitude and the additional phase. The two terms are local contributions to one final complex response; they do not possess independent global phase identities.

\subsection{Time Trajectories and Instantaneous Frequency Offset}
Along a time trajectory \(t\mapsto(\bmxi(t),\bmeta(t))\),
\begin{equation}
\frac{\dot g_s}{g_s}
=\frac{\dot{\bmw}\herm\bmc_s^{(O)}}{g_s}
+\frac{\bmw\herm\dot{\bmc}_s^{(O)}}{g_s}.
\label{eq:time-log-response}
\end{equation}
It follows that
\begin{equation}
\dot{\widetilde\psi}_s^A
=\Ima\!\left(\frac{\dot g_s}{g_s}\right)
=\dot{\widetilde\psi}_s^{(w)}
+\dot{\widetilde\psi}_s^{(c)},
\label{eq:phase-rate-decomposition}
\end{equation}
where
\begin{equation}
\dot{\widetilde\psi}_s^{(w)}
=\Ima\!\left(\frac{\dot{\bmw}\herm\bmc_s^{(O)}}{g_s}\right),
\qquad
\dot{\widetilde\psi}_s^{(c)}
=\Ima\!\left(\frac{\bmw\herm\dot{\bmc}_s^{(O)}}{g_s}\right).
\label{eq:phase-rate-components}
\end{equation}
Under the positive-phase convention adopted here, the instantaneous frequency offset induced by array processing is
\begin{equation}
 f_s^A=\frac{1}{2\pi}\dot{\widetilde\psi}_s^A.
\label{eq:array-frequency-offset}
\end{equation}
It may enter a PLL or FLL, but it cannot automatically be interpreted as true geometric Doppler.

If \(\dot{\bmw}\) is caused primarily by adaptive spatial reconstruction against suppressive interference, define
\begin{equation}
\omega_s^{\rm jam}
=\Ima\!\left(\frac{\dd\bmw\herm\bmc_s^{(O)}}{g_s}\right).
\label{eq:jam-phase-form}
\end{equation}
If a selected local model declares
\(\dd\bmc_s^{(O)}=\dd\bmc_{s,\rm dyn}^{(O)}+\dd\bmc_{s,\rm other}^{(O)}\),
with the first term collecting the stated attitude, deformation, channel, or environmental dynamics, define
\begin{equation}
\omega_s^{\rm dyn}
=\Ima\!\left(\frac{\bmw\herm\dd\bmc_{s,\rm dyn}^{(O)}}{g_s}\right).
\label{eq:dynamic-phase-form}
\end{equation}
The labels \(\omega_s^{\rm jam}\), \(\omega_s^{\rm dyn}\), and any further contributions are local bookkeeping terms. They may be summed only after an exhaustive, non-overlapping decomposition of \(\omega_s^A\) has been declared. Independently of such a decomposition, the total composite-response phase transport along a nonzero path is
\begin{equation}
T_s^{A}(\gamma)
=\exp\!\left(\jj\int_\gamma\omega_s^A\right)\in\Uone.
\label{eq:observer-transport-path}
\end{equation}
When such a decomposition is declared, its contributions compose by the Abelian group law; this does not imply statistical independence. The \(D_{\bmxi}g_s\) contribution is physical-axis-induced, whereas only the \(D_{\bmeta}g_s\) contribution is observer-state-induced.

\subsection{Two Physical Pathways of High-Dynamic Motion and Order-of-Magnitude Audit}
Choose any internal reference element \(\mu\). The array-frame baseline from that element to element \(m\) is
\begin{equation}
\bm b_{m\mu}^a=\bm b_m^a-\bm b_\mu^a,
\end{equation}
and the corresponding baseline in the navigation frame is
\begin{equation}
\bm r_{m\mu}^n=\bm C_a^n\bm b_{m\mu}^a.
\end{equation}
Under the ideal far-field model, the relative element phase is
\begin{equation}
\phi_{m\mu,s}^A
=-\kappa(\bmu_s^n)\trans\bm r_{m\mu}^n.
\label{eq:element-relative-phase}
\end{equation}
Let \(\bm\omega_{a/n}^n\) denote the angular velocity of the array with respect to the navigation frame, and adopt the convention
\begin{equation}
\dot{\bm C}_a^n
=\crossmat{\bm\omega_{a/n}^n}\bm C_a^n.
\label{eq:attitude-rate-convention}
\end{equation}
Then
\begin{equation}
\dot\phi_{m\mu,s}^A
=-\kappa\left[
(\dot{\bmu}_s^n)\trans\bm r_{m\mu}^n
+(\bm r_{m\mu}^n\times\bmu_s^n)\trans\bm\omega_{a/n}^n
\right].
\label{eq:element-phase-dynamics}
\end{equation}
The first term is due to variation of the line-of-sight direction in the navigation frame, and the second is due to rotation of the array body. Together they constitute the physical sources of the same quantity \(\dot{\bmu}_s^a\).

Let the satellite position and velocity be \(\bm p_s,\bm v_s\), and let the position and velocity of the physical reference point be \(\bm p_O,\bm v_O\). With
\begin{equation}
\rho_s=\norm{\bm p_s-\bm p_O},
\qquad
\bmu_s^n=\frac{\bm p_s-\bm p_O}{\rho_s},
\end{equation}
one obtains
\begin{equation}
\dot{\bmu}_s^n
=\frac{1}{\rho_s}
(\bm I-\bmu_s^n(\bmu_s^n)\trans)
(\bm v_s-\bm v_O).
\label{eq:los-rate}
\end{equation}
The order of the relative element-phase sensitivity caused by line-of-sight variation is \(O(\kappa B/\rho_s)\), whereas that caused by array rotation is \(O(\kappa B)\). Thus, for a compact far-field GNSS array, rapid attitude rotation generally has a much stronger direct effect on inter-element relative phase than absolute translation. Translation nevertheless remains the dominant contributor to true propagation range and ordinary Doppler, which are retained in \(\phi_s^{\nav}\).

\subsection{Inter-Satellite Array Phase and the PI Special Case}
For receiver \(q\), current satellite \(s\), and reference satellite \(r\), fix the inter-satellite differencing convention
\begin{equation}
\Delta_{sr}f_q=f_{q,s}-f_{q,r}.
\label{eq:inter-satellite-difference}
\end{equation}
The principal value of the inter-satellite composite array-response phase is
\begin{equation}
\Delta_{sr}\psi_q^A
=\Arg(g_{q,s}g_{q,r}^{*})
=\psi_{q,s}^A-\psi_{q,r}^A\pmod{2\pi}.
\label{eq:inter-satellite-array-phase}
\end{equation}
On a continuous nonzero arc, use the real-valued lift \(\Delta_{sr}\widetilde\psi_q^A\). Its dynamics satisfy
\begin{equation}
\Delta_{sr}\dot{\widetilde\psi}_q^A
=\Ima\!\left(
\frac{\dot g_{q,s}}{g_{q,s}}
-\frac{\dot g_{q,r}}{g_{q,r}}
\right).
\label{eq:inter-satellite-phase-dynamics}
\end{equation}
Even when all satellites share the same weight vector and the same weight-update rate, the weight-induced terms generally do not cancel because different satellites correspond to different points on the array manifold.

In the Power Inversion (PI) special case\cite{compton1979}, if \(\bm R_x\succ0\) and the constraint vector is \(\bmc_{\PI}\), then
\begin{equation}
\bmw_{\PI}
=\frac{\bm R_x^{-1}\bmc_{\PI}}
{\bmc_{\PI}\herm\bm R_x^{-1}\bmc_{\PI}},
\qquad
\bmc_{\PI}\herm\bmw_{\PI}=1.
\label{eq:pi-weight}
\end{equation}
If a finite sample size causes rank deficiency or severe ill-conditioning, diagonal loading \(\bm R_x+\delta\bm I\) or a generalized inverse should be used, and feasibility of the constraint must be re-examined. The PI constraint vector belongs to the weight-optimization rule; it need not correspond to the internal reference element and does not impose a distortionless constraint on any particular satellite. Consequently,
\begin{equation}
 g_{q,s}^{\PI}
=(\bmw_q^{\PI})\herm\bmc_{q,s}^{(O_q)}
\end{equation}
is generally not equal to one. The dynamics of the inter-satellite PI phase difference are
\begin{equation}
\Delta_{sr}\dot{\widetilde\psi}_q^{\PI}
=\Ima\!\left(
\frac{\dot g_{q,s}^{\PI}}{g_{q,s}^{\PI}}
-\frac{\dot g_{q,r}^{\PI}}{g_{q,r}^{\PI}}
\right).
\label{eq:inter-satellite-pi-dynamics}
\end{equation}

\subsection{Weight-Dependent Dynamic Equivalent Phase Center}
A conventional antenna phase center locally interprets a direction-dependent carrier-range correction as an equivalent displacement from a mechanical reference point. For the composite array response, define the additional equivalent range
\begin{equation}
 d_s^A(\bmu_s^a;\bmeta)
=\frac{\widetilde\psi_s^A}{\kappa}.
\label{eq:array-equivalent-range}
\end{equation}
Define the local dynamic equivalent phase center on the tangent space of the sphere by
\begin{equation}
\bm r_{{\rm pc},s}^{\rm tan,a}
=-\bm P_{u,s}\nabla_{\bmu_s^a}d_s^A.
\label{eq:dynamic-phase-center}
\end{equation}
It describes only the local first-order variation of \(d_s^A\) with direction. The radial component along the line of sight cannot be determined from a spherical directional derivative alone.

The scalar \(d_s^A\) depends on the chosen continuous lift: changing the lift by \(2\pi k\) changes \(d_s^A\) by \(k\lambda\). Its tangential gradient is unaffected by a constant branch change on a connected chart. The dynamic equivalent phase center is therefore a local first-order object, not a globally unique mechanical point.

When taking the directional gradient, hold the current weights fixed. Under the ideal element-response model
\begin{equation}
 c_{m,s}^{(O)}
=\Gamma_{m,s}\ee^{-\jj\kappa(\bm b_m^a)\trans\bmu_s^a}
\label{eq:ideal-element-response}
\end{equation}
and temporarily neglecting the directional gradient of \(\Gamma_{m,s}\), \Cref{eq:dynamic-phase-center} becomes
\begin{equation}
\bm r_{{\rm pc},s}^{\rm tan,a}
=\bm P_{u,s}\Rea\!\left[
\frac{\sum_{m=1}^{M}w_m^{*}\Gamma_{m,s}
\ee^{-\jj\kappa(\bm b_m^a)\trans\bmu_s^a}\bm b_m^a}
{g_s}
\right].
\label{eq:dynamic-phase-center-closed}
\end{equation}
This quantity varies with satellite direction, weights, channel state, and array state; it may diverge as \(g_s\to0\). It connects response-zero singularities and composite-phase conditioning to the conventional phase-center concept, but it is not a fixed mechanical point and cannot replace independent antenna calibration.

\section[G1: Response-Tracking U(1) Geometry and the PTFE]{G1: Response-Tracking \texorpdfstring{$\Uone$}{U(1)} Geometry,\\the Phase-Transport Fundamental Equation, and Transport Boundaries}
\subsection{Two Physical Roles of the Same Group}
Define
\begin{equation}
\Uone=\{\ee^{\jj\alpha}:\alpha\in\R\}.
\label{eq:u1-group}
\end{equation}
For a nonzero complex quantity, multiplication by \(\Uone\) changes only the phase coordinate, not the magnitude. In SN-ASMO, the same group has two related but distinct physical roles.

First, in a multibranch observation of the same satellite, epoch, and coherent data block,
\begin{equation}
\bmz_s=\alpha_s^{\nav}\bmg_s+\bmn_s,
\label{eq:common-phase-family}
\end{equation}
the phase of \(\alpha_s^{\nav}\) acts diagonally on all branches. For the auxiliary task of extracting relative array responses among branches, this common carrier phase is a nuisance parameter. For the primary PLL/RTK task, however, it is a navigation quantity of interest and cannot be removed globally.

Second, the phase of the composite response at an observer state,
\begin{equation}
 h_s^A(e)=\frac{g_s(e)}{\abs{g_s(e)}}\in\Uone,
\qquad e\in\calD_{\rm safe},
\label{eq:observer-u1-element}
\end{equation}
acts multiplicatively on the carrier observation as the phase of the composite array response. It may vary through either \(D_{\bmxi}g_s\) (for example, physical attitude or direction changes) or \(D_{\bmeta}g_s\) (for example, weight or channel-state changes). Only the latter is observer-state-induced; the Abelian group law does not make their statistical models independent.

The weight \(\bmw\), the point \(e\) in the joint total space, the phase fiber, and the group element on that fiber must retain distinct identities. Only when \(g_s(e)\neq0\) is \(h_s^A(e)\) a well-defined phase point at that state. Same-source multiweight observations are subject to one common diagonal \(\Uone\) action, not to an independent \(\Uone^K\) freedom on each component.

\subsection{Response-Derived Connection and the Phase-Transport Fundamental Equation}
Let
\begin{equation}
\calD_s^{\times}=\{e\in\calD_{\rm safe}:g_s(e)\neq0\},
\qquad
h_s^A=\frac{g_s}{\abs{g_s}},
\qquad
\omega_s^A=\Ima\!\left(\frac{\dd g_s}{g_s}\right).
\label{eq:response-induced-form}
\end{equation}
On the phase bundle \(\mathcal P_s=\calD_s^{\times}\times\Uone\), let \(\zeta\in\Uone\) and fix the standard radian normalization \(\omega_{\Uone}=-\jj\zeta^{-1}\dd\zeta\). Impose response-tracking compatibility by requiring the normalized-response section \(e\mapsto(e,h_s^A(e))\) to be horizontal. Under this convention and compatibility rule, define
\begin{equation}
\Omega_{g_s}
=\operatorname{pr}_{\Uone}^{*}\omega_{\Uone}
-\pi_{\rm ph}^{*}\omega_s^A.
\label{eq:response-tracking-connection}
\end{equation}
Once this compatibility rule is fixed, \(\Omega_{g_s}\) is the unique principal \(\Uone\) connection on the stated trivial bundle whose horizontal lifts satisfy it \cite{meng2026ptfe}. A nonzero scalar response alone does not select a connection before the compatibility rule is declared.

\begin{theorem}[SN-ASMO specialization of the phase-transport fundamental equation]
Let \(\gamma:[0,1]\to\calD_s^{\times}\) be a supplied piecewise-\(C^1\) path and \(\widetilde\gamma(t)=(\gamma(t),\zeta(t))\) a phase-bundle lift. Response-compatible phase transport is characterized by
\begin{equation}
\widetilde\gamma^{*}\Omega_{g_s}=0.
\label{eq:ptfe-horizontal}
\end{equation}
Equivalently,
\begin{equation}
\dot\zeta-\jj\,\omega_s^A(\dot\gamma)\zeta=0,
\qquad
\dot{\widetilde\psi}_s^A
=\omega_s^A(\dot\gamma)
=\Ima\!\left(\frac{\dot g_s}{g_s}\right).
\label{eq:ptfe-equivalent-forms}
\end{equation}
The finite phase-transport factor is
\begin{equation}
T_s^A[\gamma]
=\exp\!\left(\jj\int_\gamma\omega_s^A\right)
=h_s^A(\gamma(1))[h_s^A(\gamma(0))]^{-1}
=\frac{g_s(\gamma(0))^{*}g_s(\gamma(1))}
{\abs{g_s(\gamma(0))}\abs{g_s(\gamma(1))}}.
\label{eq:endpoint-transport}
\end{equation}
\label{thm:snasmo-ptfe}
\end{theorem}
\begin{proof}
Pulling back \Cref{eq:response-tracking-connection} along \(\widetilde\gamma\) gives the differential equation in \Cref{eq:ptfe-equivalent-forms}; integration on the nonzero path gives \Cref{eq:endpoint-transport}. This is the direct GNSS specialization of the general PTFE.
\end{proof}

Because \(\omega_s^A\) is the pullback of the real Maurer--Cartan form by the globally defined \(h_s^A\), on every response chart
\begin{equation}
\dd\omega_s^A=0.
\label{eq:flat-response-form}
\end{equation}
Hence the response-tracking connection is flat. Its circular holonomy is also trivial because the graph of \(h_s^A\) is a global horizontal section on each connected component of \(\calD_s^\times\); this conclusion does not follow from flatness alone on an arbitrary non-simply-connected base. A real-valued lift may still retain integer winding. A general non-flat \(\Uone\) connection would be additional model structure and is not derived from the scalar response used in SN-ASMO.

\subsection{Endpoint Phase Relations and Data-Observable Comparisons}
\begin{definition}[Model-defined endpoint phase factor]
For any two nonzero response states \(e_1,e_2\), define
\begin{equation}
T_{12,s}^{A,\rm mod}
=h_s^A(e_2)[h_s^A(e_1)]^{-1}\in\Uone.
\label{eq:model-defined-transport}
\end{equation}
This algebraic endpoint factor is defined by the response model. It is a PTFE path solution only when an explicit piecewise-\(C^1\) path in \(\calD_s^{\times}\) from \(e_1\) to \(e_2\) is supplied. Two nonzero endpoints alone neither guarantee such a path nor justify calling the factor a path transport.
\end{definition}

If the two states correspond to observations at different epochs,
\begin{equation}
 z_{k,s}=\alpha_{k,s}^{\nav}g_{k,s},
\qquad k=1,2,
\end{equation}
then
\begin{equation}
\ph(z_{2,s}z_{1,s}^{*})
=\ph(\alpha_{2,s}^{\nav}(\alpha_{1,s}^{\nav})^{*})
T_{12,s}^{A,\rm mod}.
\label{eq:asynchronous-product}
\end{equation}
Only when the two branches share exactly the same complex navigation factor, or when the ratio of those factors has been removed independently, does the data conjugate product directly equal the array endpoint phase factor.

\begin{definition}[Data-observable same-source phase relation]
Two branches are called strictly synchronized same-source branches if they are formed for the same satellite from the same element samples, time reference, NCOs, data-bit processing, and coherent integration operator. Their noiseless outputs then share one and the same \(\alpha_s^{\nav}\), and a relative endpoint phase factor can be constructed directly from their conjugate product.
\label{def:observable-transport}
\end{definition}

This distinction separates endpoint algebra, path transport, and observability. A model may compare endpoint phases; PTFE additionally requires an admissible nonzero path; and a direct comparison in data requires a common navigation factor, synchronization, nonzero responses, and other experimentally testable conditions. An atomic digital switch with no modeled continuous weight path may use the algebraic inverse endpoint factor, but that operation is not a PTFE path solution.

\subsection{Primacy of Genuine Degrees of Freedom and Over-Quotienting}
A mathematically writable group action is not automatically a statistical nuisance. Whether quotienting is legitimate must be determined jointly by the physical model, synchronization relations, and task objective.

\begin{theorem}[Local rank loss caused by over-quotienting]
Let \(h:\calE\to\calY\) be a response mapping, let \(G\) be the genuine nuisance group, and let \(\widetilde G\supseteq G\) be an artificially enlarged group. Define
\begin{equation}
L=\dd(\pi_G\circ h)_e,
\qquad
R=\dd p_{[h(e)]_G},
\label{eq:quotient-linear-maps}
\end{equation}
where \(p:\calY/G\to\calY/\widetilde G\) is the natural projection. Then
\begin{equation}
\rank(RL)
=\rank L
-\dim\bigl(\Range L\cap\Ker R\bigr).
\label{eq:overquotient-rank-loss}
\end{equation}
\label{thm:overquotient}
\end{theorem}
\begin{proof}
Apply the rank--nullity theorem to the restriction of \(R\) to \(\Range L\).
\end{proof}
If the intersection is nonzero, the newly quotiented directions were distinguishable by the response; over-quotienting therefore removes real information. A smaller quotient space is not automatically a more intrinsic description.

\subsection{Three Types of Common-Scale Freedom}
For \(\bmg\in\C^K\setminus\{0\}\), common nuisance actions and their quotient spaces are listed in \Cref{tab:scale-quotients}.

\begin{table}[htbp]
\centering
\caption{Common-scale nuisance actions and their exact quotient spaces}
\label{tab:scale-quotients}
\small
\begin{tabularx}{\textwidth}{>{\raggedright\arraybackslash}p{0.22\textwidth}>{\raggedright\arraybackslash}p{0.29\textwidth}>{\raggedright\arraybackslash}X}
\toprule
Genuine nuisance & Quotient space & Information removed and retained\\
\midrule
Common complex scale \(c\in\C^{\times}\) & \((\C^K\!\setminus\!\{0\})/\C^{\times}\cong\CP^{K-1}\) & Removes common magnitude and phase; retains branch-amplitude ratios and relative phases\\
Common phase only \(\ee^{\jj\alpha}\) & \((\C^K\!\setminus\!\{0\})/\Uone\) & Removes common phase; total magnitude remains observable\\
Common positive scale only \(a>0\) & \((\C^K\!\setminus\!\{0\})/\R_{+}\cong S^{2K-1}\) & Removes total scale; retains overall phase\\
\bottomrule
\end{tabularx}
\end{table}

Use of \(\CP^{K-1}\) requires a physical reason for a common nonzero complex scale to be unidentifiable. If only the common phase is unknown, a complex-projective quotient additionally and incorrectly removes the total magnitude. If each branch has its own unknown independent phase bias, a single diagonal \(\Uone\) action is insufficient to represent the genuine nuisance.

\subsection{Boundary Between Navigation and Auxiliary Tasks}
Whether a phase quantity may be quotiented out depends on the task:
\begin{enumerate}
\item In the primary carrier and RTK tasks, propagation phase, frequency, and integer ambiguity are quantities of interest and cannot be removed as one common complex scale.
\item In an auxiliary task based on synchronized same-source branches of one satellite, the shared navigation complex factor is a nuisance for extracting interbranch relative responses and may be eliminated by conjugate products, ratios, or a maximal invariant.
\item The composite array-response phase is satellite dependent. Its observer-state-induced part must be identified, transported, and inverted, while its physical-state dependence must remain typed as such; neither part may be mistaken for ordinary scalar propagation phase or expected to vanish automatically under multisatellite differencing.
\end{enumerate}

For a normalized phase vector of \(S\) satellites,
\begin{equation}
\bm h=(h_1,\ldots,h_S)\in\Uone^S=\Torus^S,
\end{equation}
a common receiver phase reference induces a diagonal action, and the inter-satellite phase space is therefore
\begin{equation}
\Torus^S/\Uone_{\rm diag}\cong\Torus^{S-1}.
\label{eq:multisatellite-phase-quotient}
\end{equation}
Selecting a reference satellite is merely a local coordinate choice. It removes a genuinely common phase, but the inter-satellite array phase \(h_s^A/h_r^A\) generally remains.
\section{Same-Source Relative Transport, Absolute Closure, and Complete Noise Propagation}
\subsection{Synchronized Same-Source Observation Family}
For the same satellite \(s\), physical epoch, and coherent data block, \(K\) synchronized weight vectors form
\begin{equation}
\bmz_s=\alpha_s^{\nav}(O)\bmg_s+\bmn_s,
\qquad
\bmg_s=[g_{s,1},\ldots,g_{s,K}]\trans.
\label{eq:same-source-family}
\end{equation}
The defining property of ``same source'' is exact sharing of the common complex navigation factor and the correlator operator, not independence of the noise. The internal reference element does not enter the physical output; it changes only the local coordinate representation of \(\bmc_s^{(O)}\).

For a working branch \(w\) and an anchor branch \(a\), the noiseless model is
\begin{equation}
 z_{w,s}^{(0)}=\alpha_s^{\nav}g_{w,s},
\qquad
 z_{a,s}^{(0)}=\alpha_s^{\nav}g_{a,s}.
\label{eq:anchor-working-noiseless}
\end{equation}
The anchor branch is a weighted statistical branch, not a raw reference element endowed in advance with a privileged physical identity. Relative alignment requires only its synchronized observation; absolute restoration additionally requires its composite phase to be closed by a model, independent calibration, or informed joint-state estimation.

\subsection{Noiseless Same-Source Cancellation and Absolute Closure}
Write the anchor-to-working endpoint phase factor in the original SN-ASMO notation as
\begin{equation}
T_{aw,s}^{A}:=h_{w,s}^{A}(h_{a,s}^{A})^{-1}.
\label{eq:anchor-working-endpoint-factor}
\end{equation}
\begin{theorem}[Noiseless same-source endpoint phase-relation identity]
If \Cref{eq:anchor-working-noiseless} holds exactly and \(\alpha_s^{\nav},g_{a,s},g_{w,s}\neq0\), then
\begin{equation}
\ph(z_{w,s}^{(0)}(z_{a,s}^{(0)})^{*})
=\ph(g_{w,s}g_{a,s}^{*})
=T_{aw,s}^{A}.
\label{eq:same-source-identity}
\end{equation}
The common navigation amplitude and phase cancel exactly, without requiring independent carrier tracking by any raw array element. If an admissible nonzero path between the branch states is supplied, the same factor is the PTFE transport along that path; otherwise it remains an exact algebraic endpoint relation.
\label{thm:same-source-identity}
\end{theorem}
\begin{proof}
Since \(z_{w,s}^{(0)}(z_{a,s}^{(0)})^{*}=\abs{\alpha_s^{\nav}}^2g_{w,s}g_{a,s}^{*}\), normalization to unit modulus gives the result.
\end{proof}

The endpoint relation is sufficient for branch alignment but not for absolute restoration. Since \(T_{aw,s}^{A}=h_{w,s}^{A}(h_{a,s}^{A})^{-1}\),
\begin{equation}
(T_{aw,s}^{A})^{-1}z_{w,s}^{(0)}
=\alpha_s^{\nav}\abs{g_{w,s}}h_{a,s}^{A}.
\label{eq:relative-alignment-to-anchor}
\end{equation}
Thus relative inverse transport maps the working branch into the anchor branch's phase fiber and is sufficient for a phase-continuous weight switch into an existing PLL coordinate. It does not remove the anchor phase.

Only if the absolute composite phase \(h_{a,s}^A\) of the anchor has been closed by a model, independent calibration, or an informed joint-state prior can one additionally form
\begin{equation}
(h_{a,s}^{A})^{-1}(T_{aw,s}^{A})^{-1}z_{w,s}^{(0)}
=\alpha_s^{\nav}\abs{g_{w,s}}.
\label{eq:absolute-restoration-after-anchor}
\end{equation}
Equivalently, the absolute working-branch phase then satisfies
\begin{equation}
 h_{w,s}^A=T_{aw,s}^Ah_{a,s}^A.
\label{eq:absolute-closure-from-anchor}
\end{equation}
Without this anchor closure, \Cref{thm:absolute-unidentifiable} leaves an unknown constant phase: temporal continuity propagates it but cannot determine it. If it contains a noninteger bias, subsequent relative transport does not make that bias disappear.

\subsection{Branch Timing Mismatch and High-Dynamic Residuals}
Suppose that the working branch is formed at \(t+\delta t\) and the anchor branch at \(t\). In addition to the desired same-epoch relative phase, the noiseless product phase contains temporal residuals from both the navigation phase and the working-branch array phase:
\begin{equation}
\Delta\phi_{\rm sync}
=\phi_s^{\nav}(t+\delta t)-\phi_s^{\nav}(t)
+\widetilde\psi_{w,s}^{A}(t+\delta t)
-\widetilde\psi_{w,s}^{A}(t).
\label{eq:timing-mismatch-exact}
\end{equation}
A second-order expansion gives
\begin{equation}
\begin{aligned}
\Delta\phi_{\rm sync}
\approx{}&
\left(2\pi f_{D,s}+\dot{\widetilde\psi}_{w,s}^{A}\right)\delta t\\
&+\frac12
\left(2\pi\dot f_{D,s}+\ddot{\widetilde\psi}_{w,s}^{A}\right)\delta t^2.
\end{aligned}
\label{eq:timing-mismatch-second-order}
\end{equation}
Thus, a common clock and common NCO must still be realized as experimentally verifiable consistency of sample time stamps. If the allowable synchronization-phase error is \(\epsilon_{\rm sync}\) and \(2\pi f_{D,s}+\dot{\widetilde\psi}_{w,s}^{A}\neq0\), a first-order sufficient condition is
\begin{equation}
\abs{\delta t}
\le
\frac{\epsilon_{\rm sync}}
{\abs{2\pi f_{D,s}+\dot{\widetilde\psi}_{w,s}^{A}}}.
\label{eq:timing-tolerance}
\end{equation}
If the first-order coefficient vanishes or nearly cancels, \Cref{eq:timing-mismatch-second-order} must be used instead. Under high dynamics, timing tolerances cannot be designed from geometric Doppler alone; the composite array-phase rate must also be included.

\subsection{Noise and Information Invariance Under a Deterministic \texorpdfstring{$\Uone$}{U(1)} Rotation}
\begin{theorem}[Invariance of additive-noise power]
For any complex random variable \(n\) and deterministic \(T\in\Uone\), let \(\widetilde n=T^{-1}n\). Then, samplewise,
\begin{equation}
\abs{\widetilde n}=\abs n,
\qquad
\E\abs{\widetilde n}^2=\E\abs n^2.
\label{eq:u1-noise-invariance}
\end{equation}
In the vector case, if \(\bm U_T\) is diagonal and unitary, then
\begin{equation}
\bm C_{\widetilde n}
=\bm U_T\herm\bm C_n\bm U_T,
\label{eq:u1-covariance-invariance}
\end{equation}
so the covariance eigenvalues, trace, and rank are preserved.
\label{thm:u1-noise}
\end{theorem}
\begin{proof}
Since \(\abs T=1\), scalar magnitudes are unchanged. In the vector case, covariance transforms by deterministic unitary conjugation, which preserves eigenvalues, trace, and rank.
\end{proof}

\begin{proposition}[Fisher-information invariance under a deterministic bijection]
If \(T\) is known and independent of the parameter to be estimated, and if \(z\mapsto T^{-1}z\) is bijective, then the transformation leaves the Fisher information for that parameter unchanged.
\label{prop:u1-fisher}
\end{proposition}
\begin{proof}
In a real representation, the transformation is a parameter-independent rotation with Jacobian determinant of unit absolute value. The likelihoods therefore differ only by a parameter-independent Jacobian factor, leaving the score and Fisher information unchanged.
\end{proof}

It is therefore essential to distinguish the following mechanisms: array weights can suppress interference in the spatial domain; same-source correlation can cancel common-mode phase noise; temporal averaging can reduce variance; but an exact \(\Uone\) inverse action by itself merely aligns the phase reference and provides no additional suppression of additive-noise power.

\subsection{Local Phase Noise Under a Proper Complex Gaussian Model}
Consider two synchronized branches
\begin{equation}
 z_i=\mu_i+n_i,
\qquad
\mu_i=A_i\ee^{\jj\phi_i},
\qquad
A_i>0,
\quad i\in\{a,w\}.
\label{eq:two-branch-gaussian}
\end{equation}
Let \(\bm n=[n_a,n_w]\trans\) be a zero-mean proper complex random vector, with
\begin{equation}
C_{ij}=\E[n_i n_j^{*}].
\label{eq:complex-noise-cov}
\end{equation}
At high signal-to-noise ratio and away from a response zero,
\begin{equation}
\epsilon_i
=\Arg(z_i)-\phi_i
\approx\frac{\Ima(n_i\ee^{-\jj\phi_i})}{A_i}.
\label{eq:phase-error-linearization}
\end{equation}
The covariance is
\begin{equation}
\cov(\epsilon_i,\epsilon_j)
\approx
\frac{1}{2A_iA_j}
\Rea\!\left[
\ee^{-\jj(\phi_i-\phi_j)}C_{ij}
\right].
\label{eq:phase-error-covariance}
\end{equation}
For improper noise, the pseudo-covariance \(P_{ij}=\E[n_i n_j]\) must also be retained.

Let the true relative phase be \(\delta=\phi_w-\phi_a\), and define
\begin{equation}
\widehat\delta=\Arg(z_wz_a^{*}),
\qquad
e_T=\widehat\delta-\delta.
\end{equation}
To first order,
\begin{equation}
 e_T\approx\epsilon_w-\epsilon_a,
\label{eq:transport-error-first-order}
\end{equation}
and hence
\begin{equation}
\sigma_T^2
\approx
\frac{C_{ww}}{2A_w^2}
+\frac{C_{aa}}{2A_a^2}
-\frac{\Rea(\ee^{-\jj\delta}C_{wa})}{A_wA_a}.
\label{eq:transport-phase-variance}
\end{equation}
Define the phase-domain correlation coefficient by
\begin{equation}
\rho_\phi
=\frac{\cov(\epsilon_w,\epsilon_a)}{\sigma_w\sigma_a}.
\end{equation}
Then
\begin{equation}
\sigma_T^2
\approx\sigma_w^2+\sigma_a^2-2\rho_\phi\sigma_w\sigma_a.
\label{eq:transport-variance-correlation}
\end{equation}
Positive, phase-aligned correlation may reduce the relative-phase variance; for independent noise the cross term is zero; adverse correlation may increase the variance. Same-source does not mean automatically noise-reducing.

Conditioned on given weights, if both branches are formed from the same element-noise vector \(\bm n_e\),
\begin{equation}
 n_w=\bmw_w\herm\bm n_e,
\qquad
 n_a=\bmw_a\herm\bm n_e,
\end{equation}
then
\begin{equation}
 C_{wa\mid W}=\bmw_w\herm\bm R_e\bmw_a.
\label{eq:branch-cross-covariance}
\end{equation}
If the weights are estimated from the same data block, write \(\widehat{\bmw}_i=\bar{\bmw}_i+\delta\bmw_i\). The first-order branch perturbation is
\begin{equation}
\delta z_i
\approx\bar{\bmw}_i\herm\bm n_e
+(\delta\bmw_i)\herm\bm\mu_e,
\label{eq:random-weight-first-order}
\end{equation}
and the second order also contains \((\delta\bmw_i)\herm\bm n_e\). An unconditional covariance must therefore propagate the weight error, element noise, and their cross terms jointly.

\subsection{Exact Phase Replacement Under Same-Statistic Inverse Transport}
\begin{theorem}[Exact same-statistic phase replacement]
If
\begin{equation}
\widehat T_{aw,s}^{A}=\ph(z_{w,s}z_{a,s}^{*})
\label{eq:same-stat-transporter}
\end{equation}
is constructed from the same nonzero pair of complex coherent statistics \((z_{a,s},z_{w,s})\) and immediately applied to the same \(z_{w,s}\), then, without a high-SNR approximation,
\begin{equation}
(\widehat T_{aw,s}^{A})^{-1}z_{w,s}
=\abs{z_{w,s}}\ph(z_{a,s}).
\label{eq:exact-phase-replacement}
\end{equation}
\label{thm:exact-phase-replacement}
\end{theorem}
\begin{proof}
Since \(\ph(z_{w,s}z_{a,s}^{*})=\ph(z_{w,s})\ph(z_{a,s})^{*}\), direct multiplication gives the result.
\end{proof}

This theorem exposes an important statistical boundary: direct same-block inverse transport replaces the phase of the working branch, sample by sample, with the noisy phase of the anchor branch; it does not automatically produce a phase less noisy than that of the anchor. Noise from the working branch still enters the amplitude through \(\abs{z_w}\). Simultaneously retaining the low residual-interference advantage of the working branch and a low phase-noise level requires a high-quality anchor, a joint likelihood, short-window smoothing, or state filtering.

\subsection{Residual Propagation in a General Estimated Inverse-Transport Architecture}
Let the true relative phase be \(\delta\), with transport estimate
\begin{equation}
\widehat\delta=\delta+e_T.
\end{equation}
The working branch used for phase discrimination is
\begin{equation}
 z_w^{(d)}=\mu_w^{(d)}+n_w^{(d)},
\end{equation}
and its mean after ideal relative inverse transport is \(\mu_c=\ee^{-\jj\delta}\mu_w^{(d)}\). The actual inverse-transport quantity is
\begin{equation}
\widetilde z_w
=\ee^{-\jj\widehat\delta}z_w^{(d)}.
\label{eq:general-compensated-complex}
\end{equation}
Expansion yields
\begin{equation}
\begin{aligned}
\widetilde z_w
={}&\mu_c+\ee^{-\jj\delta}n_w^{(d)}
-\jj\mu_c e_T
-\jj\ee^{-\jj\delta}n_w^{(d)}e_T
-\frac12\mu_ce_T^2\\
&+O_p(\abs{\mu_c}\abs{e_T}^3+\abs{n_w^{(d)}}\abs{e_T}^2).
\end{aligned}
\label{eq:residual-complex-expansion}
\end{equation}
If only first-order terms are retained, the remainder can be written as
\begin{equation}
O_p\!\left(\abs{\mu_c}e_T^2+\abs{e_T}\abs{n_w^{(d)}}\right),
\label{eq:first-order-remainder}
\end{equation}
with no independent \(\abs n^2\) term; such a term appears only if \(\Arg(\cdot)\) is subsequently expanded nonlinearly.

The linearized discriminator residual is
\begin{equation}
 e_\phi^{\rm comp}
\approx\epsilon_w^{(d)}-e_T,
\label{eq:compensated-phase-error}
\end{equation}
and therefore
\begin{equation}
\sigma_{\phi,\rm comp}^2
=\sigma_{w,d}^2+\sigma_T^2
-2\cov(\epsilon_w^{(d)},e_T).
\label{eq:compensated-phase-variance}
\end{equation}
The last term cannot be assumed to vanish. Same-block processing, overlapping windows, and shared element samples generally create significant correlation.

More generally, if the transport error is a linear combination of working- and anchor-branch phase errors from overlapping windows, let
\begin{equation}
\bm\epsilon
=[\epsilon_w^{(d)},(\bm\epsilon_w^{(e)})\trans,
(\bm\epsilon_a^{(e)})\trans]\trans,
\qquad
e_\phi^{\rm comp}=\bm b\trans\bm\epsilon.
\end{equation}
Then
\begin{equation}
\sigma_{\phi,\rm comp}^2
=\bm b\trans\bm\Sigma_\epsilon\bm b.
\label{eq:overlap-window-covariance}
\end{equation}
This expression covers same-block, partially overlapping, leave-one-out cross-fitting, and independent-window processing within one formula.

\subsection{Noise--Dynamics Tradeoff in Stateful Inverse Transport}
Let the continuously lifted relative-phase observation be
\begin{equation}
 y_\delta=\delta+\epsilon_w-\epsilon_a,
\label{eq:relative-phase-observation}
\end{equation}
and let a linear operator \(\mathsf K\) produce \(\widehat\delta=\mathsf K y_\delta\). To first order, the residual relative to the anchor-phase fiber \(\phi^{\nav}+\psi_a^A\) is
\begin{equation}
 r_\phi
=(\mathsf I-\mathsf K)\delta
+(\mathsf I-\mathsf K)\epsilon_w
+\mathsf K\epsilon_a.
\label{eq:stateful-residual}
\end{equation}
The three terms represent insufficient tracking of the relative dynamics, untransported working-branch noise, and injection of anchor-branch noise, respectively. This residual is relative to the true propagation phase only after absolute anchor closure and application of \((h_a^A)^{-1}\); if the anchor phase is estimated, its error and all cross-covariances must be added.

If the processes are wide-sense stationary and \(\delta\) is independent of the measurement noise, the frequency-domain residual spectrum is
\begin{equation}
\begin{aligned}
S_r(\omega)={}&
\abs{1-K(\omega)}^2S_\delta(\omega)
+\abs{1-K(\omega)}^2S_{ww}(\omega)\\
&+\abs{K(\omega)}^2S_{aa}(\omega)
+2\Rea\!\left\{[1-K(\omega)]K^{*}(\omega)S_{wa}(\omega)\right\}.
\end{aligned}
\label{eq:stateful-residual-spectrum}
\end{equation}
For \(K=0\), the entire relative composite-response phase is retained. For \(K=1\), the relative dynamics are removed instantaneously, but the output phase is inherited from the anchor branch, consistently with \Cref{eq:exact-phase-replacement}.

Let \(x=\delta+\epsilon_w\) and \(y=\delta+\epsilon_w-\epsilon_a\). The frequency-by-frequency noncausal Wiener lower bound is
\begin{equation}
K_W(\omega)
=\frac{S_{xy}(\omega)}{S_{yy}(\omega)}
=\frac{S_\delta+S_{ww}-S_{wa}}
{S_\delta+S_{ww}+S_{aa}-S_{wa}-S_{aw}},
\label{eq:wiener-compensator}
\end{equation}
and the minimum residual spectrum is
\begin{equation}
S_{r,\min}
=S_{xx}-\frac{\abs{S_{xy}}^2}{S_{yy}}.
\label{eq:wiener-residual}
\end{equation}
This is only a statistical lower bound. A causal implementation must also satisfy latency, continuous-lift, and zero-avoidance requirements. A sensible engineering architecture is therefore to use instantaneous transport at state-transition boundaries to guarantee phase continuity, and state filtering during continuous evolution to manage the noise--dynamics tradeoff.

\subsection{Multiblock Coherent Estimation, Temporal Correlation, and Tracking-Loop Propagation}
For \(L\) blocks after data-bit wipeoff and common-phase derotation, let
\begin{equation}
\bar z_i=\sum_{\ell=1}^{L}c_\ell z_{i,\ell},
\qquad
\sum_{\ell=1}^{L}c_\ell=1.
\end{equation}
The effective complex-noise covariance is
\begin{equation}
\bar C_{ij}
=\sum_{\ell=1}^{L}\sum_{m=1}^{L}
 c_\ell c_m^{*}\E[n_{i,\ell}n_{j,m}^{*}].
\label{eq:multiblock-covariance}
\end{equation}
Only for independent, identically distributed blocks with \(c_\ell=1/L\) does \(\bar C_{ij}=C_{ij}/L\). Under temporal correlation, division by \(L\) is not valid.

For a multisatellite residual-phase vector
\begin{equation}
\bm e_\phi^{\rm comp}
=\bm\epsilon_w^{(d)}-\bm e_T,
\end{equation}
the complete covariance is
\begin{equation}
\bm\Sigma_\phi^{\rm comp}
=\bm\Sigma_{ww}^{(d)}+\bm\Sigma_{TT}
-\bm\Sigma_{wT}-\bm\Sigma_{Tw}.
\label{eq:multisatellite-residual-covariance}
\end{equation}
Inter-satellite correlation must be retained because the satellites may share element samples, residual interference, anchor hardware, NCOs, and estimation windows.

In a linearized PLL or carrier Kalman filter, \Cref{eq:multisatellite-residual-covariance} enters as the measurement covariance \(\bm R_{\phi,k}\):
\begin{align}
\bm P_{k|k-1}&=\bm F_k\bm P_{k-1|k-1}\bm F_k\trans+\bm Q_k,
\label{eq:pll-predict}\\
\bm K_k&=\bm P_{k|k-1}\bm H_k\trans
(\bm H_k\bm P_{k|k-1}\bm H_k\trans+\bm R_{\phi,k})^{-1},
\label{eq:pll-gain}\\
\bm P_{k|k}&=(\bm I-\bm K_k\bm H_k)\bm P_{k|k-1}.
\label{eq:pll-update}
\end{align}
The transport variance cannot be added mechanically at each epoch as independent white noise.

If chained relative transport is used without re-anchoring, the accumulated error
\begin{equation}
\bm\eta_K=\bm\eta_0+\sum_{k=0}^{K-1}\bm e_{T,k}
\end{equation}
has covariance
\begin{equation}
\bm\Sigma_{\eta,K}
=\bm\Sigma_{\eta,0}
+\sum_k\bm\Sigma_{e,k}
+\sum_{k\neq\ell}\cov(\bm e_{T,k},\bm e_{T,\ell}).
\label{eq:transport-chain-covariance}
\end{equation}
Directly anchoring the current working branch avoids random-walk accumulation. Chained continuity requires periodic re-anchoring and explicit propagation of the cross-covariances.

\subsection{Low-Amplitude and Non-Gaussian Boundaries, and Safe Switching}
When \(\abs{\mu_a}\) or \(\abs{\mu_w}\) approaches zero, the complex-phase distribution may become biased, heavy-tailed, nearly uniform, or branch-ambiguous. Safe statistics should satisfy at least
\begin{equation}
\abs{\bar z_a}\ge A_{a,\min},
\qquad
\abs{\bar z_w}\ge A_{w,\min},
\qquad
\abs{\bar z_w\bar z_a^{*}}\ge\Gamma_{\min}.
\label{eq:phase-statistic-gates}
\end{equation}
Define the transport coherence by
\begin{equation}
\chi_{wa}
=\frac{
\abs{\sum_\ell z_{w,\ell}z_{a,\ell}^{*}}}
{\sqrt{\sum_\ell\abs{z_{w,\ell}}^2
\sum_\ell\abs{z_{a,\ell}}^2}}.
\label{eq:transport-coherence}
\end{equation}
Below a prescribed threshold, the system should hold its state, inflate the covariance, change anchor, or return to a joint likelihood; a Gaussian phase approximation should not be continued.

For a \textbf{continuous weight path} \(\bmw(\tau)\), one must require
\begin{equation}
\min_{\tau\in[0,1]}\abs{g_s(\bmw(\tau))}\ge\varepsilon_s,
\quad \forall s\in\calS_{\rm tr}.
\label{eq:continuous-switch-safety}
\end{equation}
For a genuinely \textbf{atomic digital switch}, the system does not traverse an arbitrary interpolation path. For response-nonzero admissibility of that atomic switch, it is sufficient that both the old and new endpoints are safe and that the weights, satellite-wise transport states, and continuous-lift counters are committed at the same sample boundary:
\begin{equation}
\abs{g_s(\bmw_{\rm old})}\ge\varepsilon_s,
\qquad
\abs{g_s(\bmw_{\rm new})}\ge\varepsilon_s.
\label{eq:atomic-switch-safety}
\end{equation}
Applying the continuous-path condition mechanically to an atomic switch over-restricts feasible weights; conversely, an unsynchronized commit creates an artificial phase step.
\section{G2: Maximal Invariants, Complex-Projective Geometry, and Task Boundaries}
\subsection{Same-Source Multiweight Auxiliary Model}
For the same satellite and the same element samples, apply \(K\) synchronized auxiliary weight vectors to obtain
\begin{equation}
\bm y_s
=\alpha_s^{\nav}\bmg_s(\bmvartheta)+\bmn_s\in\C^K,
\label{eq:g2-model}
\end{equation}
where \(\bmvartheta=\bmvartheta(\bmxi,\bmeta)\in\R^p\) is the explicitly declared target subcoordinate for this auxiliary task, and \(\alpha_s^{\nav}\in\C^{\times}\) is the shared GNSS complex factor. Depending on the experiment, \(\bmvartheta\) may contain physical-axis coordinates, observer-axis coordinates, or a stated joint subset; its type must be declared before a Fisher-information or identifiability claim is interpreted. In the auxiliary relative-response task, \(\alpha_s^{\nav}\) is a common-complex-scale nuisance; in the primary carrier task, it carries the propagation phase of interest and cannot be removed as a whole.

\subsection{Maximal Invariant Under a Common Complex Scale}
For \(\bmg\in\C^K\setminus\{0\}\), define the rank-one orthogonal projector
\begin{equation}
\bm\Pi(\bmg)
=\frac{\bmg\bmg\herm}{\bmg\herm\bmg}.
\label{eq:projective-invariant}
\end{equation}
It satisfies
\begin{equation}
\bm\Pi(c\bmg)=\bm\Pi(\bmg),
\qquad c\in\C^{\times},
\label{eq:projective-invariance}
\end{equation}
and
\begin{equation}
\bm\Pi(\bmg_1)=\bm\Pi(\bmg_2)
\Longleftrightarrow
\bmg_2=c\bmg_1,
\quad c\in\C^{\times}.
\label{eq:maximal-invariant-equivalence}
\end{equation}
Thus, for a noiseless response or a parameterized mean model, \(\bm\Pi\) is a maximal invariant under the common-complex-scale action and parameterizes a complex line in \(\CP^{K-1}\).

\begin{remark}[A maximal invariant is not automatically a sufficient statistic]
The quantity \(\bm\Pi(\bmg)\) is a maximal invariant of the noiseless response or parameterized mean. Under a fixed additive-noise model, however, the distribution of the random variable \(\bm\Pi(\bm y)\) is generally non-Gaussian and is not automatically sufficient for \(\bmvartheta\). Statistical inference must begin with the joint likelihood and elimination of genuine nuisance parameters; the noise structure cannot be ignored merely because a group invariant has been constructed.
\label{rem:maximal-not-sufficient}
\end{remark}

\subsection{Local Ratio Coordinates and Atlas Switching}
On the local chart \(g_r\neq0\), define
\begin{equation}
\varpi_k^{(r)}=\frac{g_k}{g_r},
\qquad k\neq r.
\label{eq:projective-ratios}
\end{equation}
For \(K=2\), the unique complex coordinate is
\begin{equation}
\varpi^{(1)}
=\frac{g_2}{g_1}
=\varrho_{21}\ee^{\jj\Delta\psi_{21}},
\label{eq:two-branch-projective-coordinate}
\end{equation}
where
\begin{equation}
\varrho_{21}=\frac{\abs{g_2}}{\abs{g_1}},
\qquad
\Delta\psi_{21}=\Arg(g_2)-\Arg(g_1).
\end{equation}
This coordinate retains both relative amplitude and relative phase. Pre-discriminator inverse transport uses only its \(\Uone\) component, whereas branch-quality evaluation and identifiability analysis still require the amplitude component.

If the current denominator \(g_r=0\), the ratio chart fails. As long as \(\bmg\neq0\), another nonzero component can be selected. Thus, \(g_r=0\) is a chart singularity, whereas \(\bmg=0\) is a genuine failure of the joint projective state.

\subsection{Local Bridge Between G1 and G2}
On a region where \(g_kg_r\neq0\),
\begin{equation}
\dd\log\varpi_k^{(r)}
=\frac{\dd g_k}{g_k}-\frac{\dd g_r}{g_r}.
\label{eq:g1-g2-log-bridge}
\end{equation}
The real part is the differential of the relative log-amplitude, and the imaginary part is the relative-phase differential. In vector form,
\begin{equation}
\dd\bm\varpi^{(r)}
=\diag(\bm\varpi^{(r)})\dd\log\bm\varpi^{(r)}.
\label{eq:g1-g2-vector-bridge}
\end{equation}
If all local-coordinate components are nonzero, \(\diag(\bm\varpi^{(r)})\) is invertible, and the ratio-coordinate Jacobian and the complex-log-difference Jacobian have the same local rank. Their statistical condition numbers may nevertheless differ; near a response zero, complex-log coordinates strongly amplify noise.

\subsection{Two Distinct Quotient Levels for Multiple Satellites and Multiple Weights}
A multisatellite phase vector under a common receiver phase forms \(\Torus^{S-1}\), whereas the multiweight response of one satellite under a common GNSS complex factor forms \(\CP^{K-1}\). The two must not be conflated. Specifically,
\begin{equation}
\Torus^S/\Uone_{\rm diag}
\end{equation}
describes inter-satellite relations after removal of the common multisatellite phase, whereas
\begin{equation}
(\C^K\setminus\{0\})/\C^{\times}
\end{equation}
describes the relative responses among multiple weights for a single satellite after removal of their shared complex factor. Reference-satellite coordinates and reference-weight coordinates are both local chart choices and do not alter the intrinsic objects.

\section{Noise Whitening, Nuisance Projection, the Intrinsic Jacobian, and Identifiability}
\subsection{Raw Response Jacobian and Analysis Conditions}
Let \(\bmvartheta=\bmvartheta(\bmxi,\bmeta)\in\R^p\) be the task-declared physical, observer, or joint subcoordinate to be estimated, and let \(\bmg_s(\bmvartheta)\in\C^K\) be the synchronized multiweight response. The raw complex Jacobian is
\begin{equation}
\bm D_s
=\frac{\partial\bmg_s}{\partial\bmvartheta\trans}
\in\C^{K\times p}.
\label{eq:raw-response-jacobian}
\end{equation}
In the linear special case with a fixed weight matrix \(\bm W\),
\begin{equation}
\bm D_s
=\bm W\herm
\frac{\partial\bmc_s^{(O)}}{\partial\bmvartheta\trans}.
\label{eq:fixed-weight-jacobian}
\end{equation}
If the weights themselves depend on \(\bmvartheta\), the total derivative also contains
\begin{equation}
\left[
\frac{\partial\bm W\herm}{\partial\vartheta_1}\bmc_s^{(O)},
\ldots,
\frac{\partial\bm W\herm}{\partial\vartheta_p}\bmc_s^{(O)}
\right].
\label{eq:weight-dependent-jacobian}
\end{equation}
One must therefore state explicitly whether the analysis concerns identifiability of the physical state conditioned on the observer state, or joint identifiability of the physical and observer states. This section is conditioned on given weights; random weights are treated according to \Cref{ass:conditional-weight}.

\subsection{Effective Output Subspace and Whitening}
Let the output-noise covariance be
\begin{equation}
\bm C_s=\E[\bmn_s\bmn_s\herm]\succeq0.
\label{eq:output-noise-covariance}
\end{equation}
If \(\bm C_s\succ0\), one may take \(\bm L_s=\bm C_s^{-1/2}\). If \(\bm C_s\) is singular, decompose its positive-eigenvalue subspace as
\begin{equation}
\bm C_s
=\bm U_{s,+}\bm\Lambda_{s,+}\bm U_{s,+}\herm,
\qquad
\bm\Lambda_{s,+}\succ0,
\label{eq:positive-eigenspace}
\end{equation}
and define the effective whitening operator
\begin{equation}
\bm L_s
=\bm\Lambda_{s,+}^{-1/2}\bm U_{s,+}\herm,
\qquad
\bm L_s\bm C_s\bm L_s\herm=\bm I_{r_s}.
\label{eq:effective-whitener}
\end{equation}
Set
\begin{equation}
\widetilde{\bmg}_s=\bm L_s\bmg_s,
\qquad
\widetilde{\bm D}_s=\bm L_s\bm D_s.
\label{eq:whitened-response}
\end{equation}
Discarding \(\Ker\bm C_s\) is lossless only when the covariance singularity is caused by redundant weights, or when the signal and Jacobian directions lie in the positive-noise subspace. If the noise null space contains deterministic information, it must be handled separately as an exact constraint.

\subsection{Tangent Space of the Genuine Nuisance Orbit}
Define the real embedding of a complex vector by
\begin{equation}
\calR(\bm v)
=\begin{bmatrix}\Rea\bm v\\\Ima\bm v\end{bmatrix}.
\label{eq:real-embedding}
\end{equation}
Let \(\bm N_s\) be a real basis for the tangent space of the genuine nuisance orbit at \(\widetilde{\bmg}_s\), and define the orthogonal-complement projector
\begin{equation}
\bm P_{N_s}^{\perp}
=\bm I-\bm N_s(\bm N_s\trans\bm N_s)^{\dagger}\bm N_s\trans.
\label{eq:nuisance-projector}
\end{equation}
The intrinsic Jacobian is
\begin{equation}
\bm J_s^{\rm int}
=\bm P_{N_s}^{\perp}\calR(\widetilde{\bm D}_s).
\label{eq:intrinsic-jacobian}
\end{equation}
Its image consists of output variations that cannot be explained by motion along the nuisance orbit, and its kernel consists of state perturbations that the current observation block cannot distinguish.

Common orbit tangent spaces include
\begin{equation}
\bm N_{\phi,s}=\calR(\jj\widetilde{\bmg}_s)
\qquad\text{(common phase only)},
\label{eq:phase-nuisance-tangent}
\end{equation}
and
\begin{equation}
\bm N_{\C,s}
=[\calR(\widetilde{\bmg}_s),\ \calR(\jj\widetilde{\bmg}_s)]
\qquad\text{(common complex scale)}.
\label{eq:complex-scale-tangent}
\end{equation}
Only genuinely shared common directions may be removed. If an independent phase nuisance is assigned incorrectly to every branch, all interbranch relative-phase and transport information is deleted.

\subsection{Complex Projection for a Common Complex Scale}
For a common-complex-scale nuisance, the complex orthogonal projector is
\begin{equation}
\bm P_{\widetilde g_s}^{\perp}
=\bm I-
\frac{\widetilde{\bmg}_s\widetilde{\bmg}_s\herm}
{\widetilde{\bmg}_s\herm\widetilde{\bmg}_s}.
\label{eq:complex-projector}
\end{equation}
Define
\begin{equation}
\bm D_s^{\perp}
=\bm P_{\widetilde g_s}^{\perp}\widetilde{\bm D}_s,
\qquad
\bm J_s^{\rm proj}=\calR(\bm D_s^{\perp}).
\label{eq:projected-complex-jacobian}
\end{equation}
One complex projection removes the two real directions \(\widetilde{\bmg}_s\) and \(\jj\widetilde{\bmg}_s\); under a common-complex-scale model, it therefore has the same kernel and rank as \Cref{eq:intrinsic-jacobian}.

\begin{theorem}[Differential of the maximal invariant]
For \(\bm\Pi(\bmg)=\bmg\bmg\herm/(\bmg\herm\bmg)\) and a state perturbation \(\delta\bmvartheta\),
\begin{equation}
\dd\bm\Pi[\delta\bmvartheta]
=\frac{
\bm P_g^{\perp}\bm D\delta\bmvartheta\,\bmg\herm
+\bmg(\bm D\delta\bmvartheta)\herm\bm P_g^{\perp}}
{\bmg\herm\bmg},
\label{eq:maximal-invariant-differential}
\end{equation}
and
\begin{equation}
\norm{\dd\bm\Pi[\delta\bmvartheta]}_F^2
=\frac{2\norm{\bm P_g^{\perp}\bm D\delta\bmvartheta}^2}
{\bmg\herm\bmg}.
\label{eq:maximal-invariant-norm}
\end{equation}
\label{thm:max-invariant-diff}
\end{theorem}
\begin{proof}
Differentiate the normalized rank-one projector directly and use \(\bm P_g^{\perp}\bmg=0\).
\end{proof}
Equation~\eqref{eq:maximal-invariant-norm} shows that comparison of scale-independent observation strength must include the normalization \(\norm{\bmg}^{-2}\), rather than comparing only the unnormalized norm of the projected Jacobian.

\subsection{Local Identifiability and a Dimensional Resource Bound}
\begin{definition}[Local intrinsic identifiability]
If, at \(\bmvartheta_0\),
\begin{equation}
\rank\bm J_s^{\rm int}=p,
\label{eq:local-identifiability-rank}
\end{equation}
then \(\bmvartheta\) is said to be locally differentially identifiable for the given observation block and specified genuine nuisance.
\end{definition}
This condition does not guarantee global uniqueness; a stable local immersion additionally requires constant rank in a neighborhood.

If the effective complex-output rank is \(r_s\), the maximum real intrinsic dimension of one block satisfies
\begin{equation}
\dim_{\R}\calY_{\rm int}
\le
\begin{cases}
2(r_s-1),&\text{common-complex-scale nuisance},\\
2r_s-1,&\text{common phase only or common positive scale only}.
\end{cases}
\label{eq:intrinsic-dimension-bound}
\end{equation}
For multiple satellites and blocks, when each block has an independent common complex scale, a necessary resource condition is
\begin{equation}
 p\le\sum_{s=1}^{S}\sum_{\ell=1}^{L}2(r_{s,\ell}-1),
\label{eq:resource-necessary-condition}
\end{equation}
but sufficiency still requires the joint Jacobian condition
\begin{equation}
\rank
\begin{bmatrix}
\bm J_{1,1}^{\rm int}\\
\vdots\\
\bm J_{S,L}^{\rm int}
\end{bmatrix}
=p.
\label{eq:joint-identifiability}
\end{equation}
If one common complex gain is shared across multiple epochs, the observations must first be stacked jointly and the shared nuisance projected out only once. Subtracting two real dimensions independently at every epoch is incorrect.

When \(K=1\) and the common complex gain of a single instantaneous complex output is completely unknown, the quotient space is \(\CP^0\). This means only that a single-branch instantaneous auxiliary observation cannot by itself separate the GNSS complex factor from the array response; it does not mean that the primary GNSS carrier phase contains no information. Temporal continuity alone propagates, but does not determine, the unknown constant complex factor. Additional samples with an identified dynamical model, or a second synchronized branch, may restore task-specific relative identifiability; absolute closure still requires a model, calibration, or informative anchor.

\section{Fisher Information, Fubini--Study Geometry, and Increments from Correlated Observations}
\subsection{State Metrics and Coordinate-Independent Observation Strength}
Singular values of a raw Jacobian change with the use of metres, radians, millimetres, or other parameter units. Let \(\bm G_\vartheta\succ0\) be a reference Riemannian metric on the state space, and let \(\bm G_y\succ0\) be a metric on the intrinsic observation space. Define the generalized eigenproblem
\begin{equation}
(\bm J^{\rm int})\trans\bm G_y\bm J^{\rm int}\bm v
=\lambda\bm G_\vartheta\bm v.
\label{eq:generalized-observation-spectrum}
\end{equation}
Equivalently, the \(\lambda_i\) are the squared singular values of
\begin{equation}
\bm G_y^{1/2}\bm J^{\rm int}\bm G_\vartheta^{-1/2}.
\label{eq:metric-normalized-jacobian}
\end{equation}
The weakest observable direction is characterized by
\begin{equation}
\lambda_{\min}
=\min_{\delta\bmvartheta\neq0}
\frac{\norm{\bm J^{\rm int}\delta\bmvartheta}_{\bm G_y}^2}
{\norm{\delta\bmvartheta}_{\bm G_\vartheta}^2}.
\label{eq:minimum-generalized-information}
\end{equation}
Under reparameterization, \(\bm G_\vartheta\) must transform tensorially in order for the numerical values to remain coordinate independent. After whitening, the natural choice is \(\bm G_y=\bm I\); without whitening, an inverse covariance may be used, but the noise must not be whitened twice.

\subsection{Fisher Schur Complement for a General Nuisance}
Partition the parameter into the state of interest \(\bmvartheta\) and a nuisance \(\bm\nu\). The full Fisher information is
\begin{equation}
\bm F
=\begin{bmatrix}
\bm F_{\vartheta\vartheta}&\bm F_{\vartheta\nu}\\
\bm F_{\nu\vartheta}&\bm F_{\nu\nu}
\end{bmatrix}.
\label{eq:fisher-blocks}
\end{equation}
Under regularity conditions and
\begin{equation}
\Range(\bm F_{\nu\vartheta})
\subseteq\Range(\bm F_{\nu\nu}),
\label{eq:generalized-schur-condition}
\end{equation}
the effective Fisher information for the target state is the generalized Schur complement
\begin{equation}
\bm F_\vartheta^{\rm eff}
=\bm F_{\vartheta\vartheta}
-\bm F_{\vartheta\nu}\bm F_{\nu\nu}^{\dagger}\bm F_{\nu\vartheta}.
\label{eq:fisher-schur}
\end{equation}
It is equivalent to projecting the score for the state of interest onto the orthogonal complement of the nuisance-score subspace\cite{kay1993,amari2016}. Only parameters established physically and statistically to be genuine nuisances may be eliminated. Over-elimination corresponds to the information loss in \Cref{eq:overquotient-rank-loss}.

\subsection{Common Complex-Gain Model and the Pullback of the Fubini--Study Metric}
Consider the whitened proper complex Gaussian model
\begin{equation}
\widetilde{\bm y}
=\alpha\widetilde{\bmg}(\bmvartheta)+\widetilde{\bmn},
\qquad
\widetilde{\bmn}\sim\CN(\bm0,\bm I),
\label{eq:white-common-gain-model}
\end{equation}
where the covariance is locally independent of \(\bmvartheta\), the response satisfies \(\widetilde{\bmg}(\bmvartheta)\neq\bm0\), and the common complex gain \(\alpha\in\C^{\times}\) is a nuisance in the auxiliary relative-response task. Let
\begin{equation}
\widetilde{\bm D}
=\frac{\partial\widetilde{\bmg}}{\partial\bmvartheta\trans}.
\end{equation}

\begin{theorem}[Effective Fisher information and the Fubini--Study metric]
After elimination of the common complex gain, the effective Fisher information is
\begin{equation}
\bm F_\vartheta^{\rm eff}
=2\abs{\alpha}^2
\Rea\!\left(
\widetilde{\bm D}\herm
\bm P_{\widetilde g}^{\perp}
\widetilde{\bm D}
\right).
\label{eq:effective-fisher-projective}
\end{equation}
Under the normalization convention adopted here, the pullback to the state space of the Fubini--Study metric on \(\CP^{K-1}\) is
\begin{equation}
\bm G_{\rm FS}
=\frac{
\Rea\!\left(
\widetilde{\bm D}\herm
\bm P_{\widetilde g}^{\perp}
\widetilde{\bm D}
\right)}
{\widetilde{\bmg}\herm\widetilde{\bmg}},
\label{eq:fs-pullback}
\end{equation}
and the two quantities satisfy
\begin{equation}
\bm F_\vartheta^{\rm eff}
=2\abs{\alpha}^2\norm{\widetilde{\bmg}}^2\bm G_{\rm FS}.
\label{eq:fisher-fs-equivalence}
\end{equation}
\label{thm:fisher-fs}
\end{theorem}
\begin{proof}
For a proper complex Gaussian mean model, the real Fisher inner product is twice the real part of the complex inner product between mean Jacobians. The two real nuisance directions of the common complex gain are \(\widetilde{\bmg}\) and \(\jj\widetilde{\bmg}\); Schur elimination is therefore equivalent to left multiplication by \(\bm P_{\widetilde g}^{\perp}\). Comparison with \Cref{eq:maximal-invariant-norm} gives the stated relation.
\end{proof}

The theorem unifies three levels. The quantity \(\bm P_{\widetilde g}^{\perp}\widetilde{\bm D}\) determines geometrically identifiable directions; \(\bm G_{\rm FS}\) describes projective shape variation after removal of the common complex scale; and \(\abs\alpha^2\norm{\widetilde\bmg}^2\) converts geometric sensitivity into statistical observation strength.

\subsection{State-Dependent Covariance, Random Weights, and the Low-SNR Boundary}
For
\begin{equation}
\bm y\sim\CN(\bm\mu(\bmvartheta),\bm C(\bmvartheta)),
\end{equation}
if the covariance depends on the state, the Fisher information for real parameters \(\vartheta_i,\vartheta_j\) is
\begin{equation}
\begin{aligned}
F_{ij}={}&2\Rea\!\left[
\frac{\partial\bm\mu\herm}{\partial\vartheta_i}
\bm C^{-1}
\frac{\partial\bm\mu}{\partial\vartheta_j}
\right]\\
&+\tr\!\left(
\bm C^{-1}\frac{\partial\bm C}{\partial\vartheta_i}
\bm C^{-1}\frac{\partial\bm C}{\partial\vartheta_j}
\right).
\end{aligned}
\label{eq:state-dependent-cov-fisher}
\end{equation}
Consequently, \Cref{eq:effective-fisher-projective} applies only to a mean model with locally fixed covariance and conditioned on given weights. If the weights are estimated adaptively from the same data block, one should use \(p(\bm y,\widehat{\bm W}\mid\bmvartheta)\), or eliminate the weight errors jointly as nuisance parameters.

At low SNR, under heavy-tailed interference, near a response zero, with a multimodal likelihood, or under weak identifiability, the Fisher information is only a local regular-model lower bound and may not predict finite-sample errors accurately\cite{vallisneri2008}. SN-ASMO does not use a Fisher matrix as a substitute for the complete likelihood, Monte Carlo analysis, or integer-boundary risk.

\subsection{Conditional Fisher Increment from Correlated New Observations}
In this subsection \(\bm y_0,\bm y_a,\bm J\), and \(\bm Q\) are real-valued, either natively or after realification of complex data. A direct proper-complex formulation must use the Hermitian, real-part, and factor-of-two convention of \Cref{eq:effective-fisher-projective} consistently. Let an existing observation \(\bm y_0\) and an added auxiliary observation \(\bm y_a\) be jointly Gaussian, with mean Jacobian and covariance
\begin{equation}
\bm J
=\begin{bmatrix}\bm J_0\\\bm J_a\end{bmatrix},
\qquad
\bm Q
=\begin{bmatrix}
\bm Q_0&\bm Q_{0a}\\
\bm Q_{a0}&\bm Q_a
\end{bmatrix},
\label{eq:joint-correlated-observation}
\end{equation}
and assume that the covariance is locally independent of the parameter and that \(\bm Q\succ0\). Equivalently, \(\bm Q_0\succ0\) and the Schur complement \(\bm Q_{a|0}\succ0\). The conditional covariance and conditional Jacobian are
\begin{equation}
\bm Q_{a|0}
=\bm Q_a-\bm Q_{a0}\bm Q_0^{-1}\bm Q_{0a},
\label{eq:conditional-covariance}
\end{equation}
and
\begin{equation}
\bm J_{a|0}
=\bm J_a-\bm Q_{a0}\bm Q_0^{-1}\bm J_0.
\label{eq:conditional-jacobian}
\end{equation}

\begin{theorem}[Conditional Fisher increment]
Under the preceding conditions, the joint Fisher information decomposes as
\begin{equation}
\bm F_{\rm total}
=\bm F_0+\bm F_{a|0},
\qquad
\bm F_{a|0}
=\bm J_{a|0}\trans\bm Q_{a|0}^{-1}\bm J_{a|0}\succeq0.
\label{eq:conditional-fisher-increment}
\end{equation}
\label{thm:conditional-fisher}
\end{theorem}
\begin{proof}
Factor the joint Gaussian density as \(p(\bm y_0\mid\bmvartheta)p(\bm y_a\mid\bm y_0,\bmvartheta)\). The expected cross product of the two scores is zero, and the conditional Gaussian mean Jacobian and covariance are given by \Cref{eq:conditional-jacobian,eq:conditional-covariance}.
\end{proof}

When \(\bm Q_{a0}=0\), \Cref{eq:conditional-fisher-increment} reduces to addition of information from independent observations. If raw samples are shared, the marginal sum \(\bm F_0+\bm F_a\) generally double-counts information; the conditional increment or the directly computed joint information must be used. If the covariance depends on the parameter, the conditional decomposition must also include the corresponding covariance-derivative terms.
\section[Propagation of Positioning Information from Precise Spatial-Manifold Observations]{Propagation of Positioning Information\\from Precise Spatial-Manifold Observations}
\subsection{Two Distinct Sources of Positioning Gain}
At least two fundamentally different mechanisms allow spatial-manifold information to enter a positioning system.

\begin{enumerate}
\item \textbf{Observation-quality gain.} A single common weight vector suppresses interference, and pre-discriminator restoration preserves the standard carrier structure of the primary output. The backend still has one scalar observation per satellite. This mechanism changes mainly the observation covariance and availability; it does not add independent physical observation dimensions.
\item \textbf{Intrinsic state-information gain.} If synchronized multiweight auxiliary branches, array-attitude or calibration estimates, or projective-response statistics are retained and passed to a joint estimator, they can contribute additional conditional Fisher information about attitude, channel state, array deformation, or multipath state, and can improve position indirectly through state coupling.
\end{enumerate}

If auxiliary branches are used only for internal weight selection and discarded after the selection, their Jacobians cannot subsequently be added as independent positioning observations; doing so would double-count information contained in the same raw samples.

\subsection{CRLB Monotonicity Under a Legitimate Information Increment}
Suppose that the standard observations provide information \(\bm F_0\succ0\) about a parameter vector \(\bm x\), and that the legitimate conditional increment supplied by the auxiliary observation relative to the standard observation is \(\bm F_{a|0}\succeq0\). By \Cref{eq:conditional-fisher-increment},
\begin{equation}
\bm F_1=\bm F_0+\bm F_{a|0}\succeq\bm F_0.
\label{eq:position-fisher-increment}
\end{equation}

\begin{theorem}[CRLB monotonicity under a legitimate information increment]
If \(\bm F_0\succ0\) and \(\bm F_{a|0}\succeq0\), then
\begin{equation}
\bm F_1^{-1}\preceq\bm F_0^{-1}.
\label{eq:crlb-monotonicity}
\end{equation}
For any position-selection matrix \(\bm S_p\),
\begin{equation}
\bm S_p\bm F_1^{-1}\bm S_p\trans
\preceq
\bm S_p\bm F_0^{-1}\bm S_p\trans.
\label{eq:position-crlb-monotonicity}
\end{equation}
\label{thm:crlb-monotonicity}
\end{theorem}
\begin{proof}
Matrix inversion is order-reversing on the cone of positive-definite matrices, and congruence transformations preserve the semidefinite order.
\end{proof}

The premise is that the added information has already been orthogonalized under the conditional model. Marginal information from observations that share samples cannot be added directly.

\subsection{Positioning-Gain Factor and Unit Audit}
Let \(\bm F_{p,0}^{\rm eff}\) and \(\bm F_{p,1}^{\rm eff}\) denote the effective position Fisher information matrices after elimination of clock, troposphere, attitude, calibration, and other nuisance parameters. If a scalar \(\gamma_p\ge1\) exists such that
\begin{equation}
\bm F_{p,1}^{\rm eff}
\succeq\gamma_p\bm F_{p,0}^{\rm eff},
\label{eq:position-gain-factor}
\end{equation}
then
\begin{equation}
\bm C_{p,1}^{\rm CRLB}
\preceq\gamma_p^{-1}\bm C_{p,0}^{\rm CRLB},
\label{eq:position-gain-covariance}
\end{equation}
and
\begin{equation}
\frac{\sqrt{\tr\bm C_{p,1}^{\rm CRLB}}}
{\sqrt{\tr\bm C_{p,0}^{\rm CRLB}}}
\le\gamma_p^{-1/2}.
\label{eq:position-rmse-gain}
\end{equation}
The factor \(\gamma_p\) must be verified from the effective information matrices; it cannot be inferred directly from the number of elements, output SINR, or an unnormalized Jacobian norm.

\subsection{Far-Field Translation and Attitude Sensitivities}
Consider the ideal phase of element \(m\) relative to the physical reference point,
\begin{equation}
\phi_{m,s}^{A}
=-\kappa(\bmu_s^n)\trans\bm C_a^n\bm b_m^a.
\label{eq:position-attitude-phase}
\end{equation}
Let the satellite position be \(\bm p_s\), the reference-point position be \(\bm p\), and \(\rho_s=\norm{\bm p_s-\bm p}\). Then
\begin{equation}
\frac{\partial\bmu_s^n}{\partial\bm p\trans}
=-\frac{1}{\rho_s}
(\bm I-\bmu_s^n(\bmu_s^n)\trans),
\label{eq:los-position-jacobian}
\end{equation}
so that
\begin{equation}
\nabla_{\bm p}\phi_{m,s}^{A}
=\frac{\kappa}{\rho_s}
(\bm I-\bmu_s^n(\bmu_s^n)\trans)
\bm C_a^n\bm b_m^a.
\label{eq:phase-position-gradient}
\end{equation}
Its order is \(O(\kappa B/\rho_s)\), and the first-order sensitivity along the line-of-sight radial direction is zero.

Under the attitude convention in \Cref{eq:attitude-left-perturb},
\begin{equation}
\delta\phi_{m,s}^{A}
=\kappa(\bmu_s^n\times\bm C_a^n\bm b_m^a)\trans\delta\bmtheta^n,
\label{eq:phase-attitude-gradient}
\end{equation}
which is of order \(O(\kappa B)\). A compact far-field array is therefore much more directly sensitive to attitude and internal array states than to absolute translation. Substantial positioning gains arise mainly from interference-noise reduction, decoupling of attitude, lever arm, and calibration, multipath constraints, and joint geometry---not from mischaracterizing the directional array response as a strong millimetre-level absolute-translation observation.

\subsection{Single-Output Compression and the Data-Processing Boundary}
The single-output mapping
\begin{equation}
\bm x\longmapsto\bmw\herm\bm x
\label{eq:single-output-compression}
\end{equation}
compresses \(M\)-dimensional element data into one complex scalar. If the backend retains only this scalar, the data-processing inequality implies that it cannot contain more Fisher information about the same parameter than the complete element data. SN-ASMO does not violate this principle. Rather, it seeks to
\begin{enumerate}
\item make the retained scalar closer to a task-relevant sufficient statistic through weight optimization;
\item preserve relative-response information through a small number of synchronized auxiliary branches; and
\item prevent useful carrier information from being contaminated by residual composite-response phase through exact quotienting of genuine nuisances and pre-discriminator inverse transport.
\end{enumerate}

\subsection{From Per-Satellite Covariance to Position Fisher Information}
Let \(\bm Q_y(\bmw)\) be the covariance of inverse-transported per-satellite code and carrier observations, let \(\bm H_x\) be the design matrix, and let \(\bm P_x\succeq0\) represent prior or dynamical information. Then
\begin{equation}
\bm F_x(\bmw)
=\bm H_x\trans\bm Q_y^{-1}(\bmw)\bm H_x+\bm P_x.
\label{eq:position-fisher-weight}
\end{equation}
If the weights change only \(\bm Q_y\), the gain comes from improved observation quality. If intrinsic array observations \(\bm y_a\) are also retained, they must be incorporated through the joint covariance or a conditional Fisher increment; the same prompt information must not be counted twice.

The basic theory can provide task-oriented objective functions to an upper-level active-observation controller, for example
\begin{equation}
J_p(\bmw)=\log\det\bm C_p^{\rm CRLB}(\bmw),
\qquad
J_E(\bmw)=-\lambda_{\min}(\bm F_p^{\rm eff}(\bmw)).
\label{eq:position-weight-objectives}
\end{equation}
Any weight evaluation must simultaneously enforce the PI constraint, interference suppression, tracking thresholds for critical satellites, transport coherence, the composite-phase condition numbers, and response-zero safety. A minimum-output-power weight need not be positioning-optimal, because it may drive a geometrically important satellite toward a deep null.
\section[Carrier Observations, Bias--Covariance Propagation, and Preservation of the RTK Integer Structure]{Carrier Observations, Bias--Covariance Propagation,\\and Preservation of the RTK Integer Structure}
\subsection{Standard Carrier Observation After Inverse Transport and Absolute Closure}
The final composite array-response phase should be inversely transported before the complex coherent statistic enters the carrier discriminator. Relative inverse transport alone aligns the working output to an anchor fiber as in \Cref{eq:relative-alignment-to-anchor}; the standard navigation form below additionally requires absolute anchor closure. Under that condition, the discriminator input for receiver \(i\) and satellite \(s\) is
\begin{equation}
\widetilde z_{i,s}
=\abs{g_{i,s}}\alpha_{i,s}^{\nav}(O_i)+\widetilde n_{i,s}.
\label{eq:pre-discriminator-restored}
\end{equation}
Because \(g_{i,s}\) is invariant under a change of internal-reference-element chart, \Cref{eq:pre-discriminator-restored} is likewise independent of the internal reference element.

This paper fixes a positive carrier-phase convention. The receiver's continuous carrier observation \(L_{i,s}\) is expressed in cycles, and
\begin{equation}
\widetilde\phi_{i,s}^{\nav}=2\pi L_{i,s}^{\nav}.
\label{eq:phase-cycle-convention}
\end{equation}
The restored standard undifferenced carrier model is
\begin{equation}
\lambda\widetilde L_{i,s}
=\rho_{i,s}(O_i)
+c_0(\delta t_i-\delta t_s)
+T_{i,s}-I_{i,s}
+\lambda N_{i,s}
+B_{i,s}^{\hw}
+\widetilde\epsilon_{i,s}.
\label{eq:standard-carrier-model}
\end{equation}
Here, \(N_{i,s}\in\Z\) is the arc-constant integer arising when the modulo-cycle carrier phase is lifted and the tracking arc is initialized, whereas \(B_{i,s}^{\hw}\in\R\) is a noninteger hardware bias expressed in metres. The latter is not an arbitrary container for a failed integer closure. If the solution coordinates are defined at a different mechanical point, they must be transferred to \(O_i\) using a known lever arm; this has no relation to the internal reference element.

\subsection{\texorpdfstring{$\Uone$}{U(1)} Branches, Continuous Lifts, and Cycle Counting}
Exact inverse phase transport uses the group element
\begin{equation}
(h_{i,s}^{A})^{-1}=\ee^{-\jj\psi_{i,s}^{A}}.
\end{equation}
Because
\begin{equation}
\ee^{-\jj(\psi_{i,s}^{A}+2\pi k)}
=\ee^{-\jj\psi_{i,s}^{A}},
\qquad k\in\Z,
\label{eq:u1-branch-independence}
\end{equation}
a pure \(\Uone\) inverse action is independent of the principal-value phase branch; a fixed \(2\pi k\) term need not be ``absorbed into the ambiguity.''

The actual integer-cycle structure is associated with the covering map
\begin{equation}
0\longrightarrow2\pi\Z
\longrightarrow\R
\xrightarrow{\exp(\jj\cdot)}\Uone
\longrightarrow1.
\label{eq:universal-cover-sequence}
\end{equation}
The continuous PLL phase state lives in \(\R\), and cycle counting is the action of the deck-transformation group \(2\pi\Z\). If the continuous-lift counter used for inverse transport, the weight commit, and the PLL cycle count are synchronized, the integer ambiguity is unchanged. An unsynchronized jump in any of these integer states appears as an artificial cycle slip.

\begin{theorem}[Cycle continuity under synchronized relative inverse transport]
Suppose \(g_{\rm old},g_{\rm new}\neq0\), and the old and new outputs are strictly synchronized same-source branches at the same sample boundary. Let the applied relative factor be the exact \(T_{aw,s}^{A}\), with \(a={\rm old}\) and \(w={\rm new}\), and choose its real lift consistently with the existing PLL tracking arc. If the weights, inverse factor, lift counter, and PLL state are committed atomically, relative inverse transport maps the new branch into the old PLL phase fiber without changing the integer cycle count.
\label{thm:relative-cycle-continuity}
\end{theorem}
\begin{proof}
On the real lift,
\begin{equation}
\widetilde\phi_{\rm new}^{\obs}
=\widetilde\phi^{\nav}+\widetilde\psi_{\rm new}^{A}.
\end{equation}
Subtracting the consistently lifted relative increment
\(\widetilde\psi_{\rm new}^{A}-\widetilde\psi_{\rm old}^{A}\) gives
\(\widetilde\phi^{\nav}+\widetilde\psi_{\rm old}^{A}\). No \(2\pi k\) deck transformation is introduced, so the PLL cycle count is unchanged. Failure of atomicity invalidates the last step and can create an artificial cycle slip.
\end{proof}
For an estimated factor, atomicity alone is insufficient: integer preservation additionally requires the residual phase to remain continuous and not to trigger the declared cycle-slip or loss-of-lock gate.

\begin{theorem}[Standard carrier form after absolute phase closure]
Suppose, in addition to the hypotheses of \Cref{thm:relative-cycle-continuity}, that the absolute composite phase is closed by a model, calibrated same-source anchor, or informed joint-state estimate; the closure supplies a lift consistent with the current tracking arc and is inversely transported continuously before the discriminator. If no unmodeled front-end nonlinearity, loss of lock, data-bit error, or unsynchronized cycle-count jump occurs, the output retains the standard structure in \Cref{eq:standard-carrier-model}; \(N_{i,s}\) is unchanged. A noninteger residual enters only as an explicit real-valued bias or random error.
\label{thm:integer-preservation}
\end{theorem}
\begin{proof}
The lifted observation satisfies \(\widetilde\phi_{i,s}^{\obs}=\widetilde\phi_{i,s}^{\nav}+\widetilde\psi_{i,s}^{A}\). Absolute inverse transport subtracts the same continuous lift \(\widetilde\psi_{i,s}^{A}\), leaving \(\widetilde\phi_{i,s}^{\nav}\) plus the stated residual. Changing a principal branch by \(2\pi k\) leaves the \(\Uone\) factor unchanged and does not alter the separately maintained deck-transformation count. Hence the arc-initialized integer \(N_{i,s}\) is preserved.
\end{proof}

\subsection{Unit-Consistent Propagation from Discriminator Residuals to Carrier Covariance}
Let the per-satellite pre-discriminator residual-phase covariance of receiver \(i\), \(\bm\Sigma_{\phi,i}^{\rm comp}\), be expressed in \(\mathrm{rad}^2\). Over a linearized tracking window, let \(\bm{\mathcal L}_i\) denote the linear operator from the sequence of discriminator residuals to the carrier-phase estimation error, and let the oscillator, dynamics-model, and quantization-error covariance \(\bm\Sigma_{\rm dyn,i}^{\rm rad}\) also be expressed in \(\mathrm{rad}^2\). The undifferenced carrier covariance in cycles squared is then
\begin{equation}
\bm\Sigma_{i,L}^{\rm comp}
=\frac{1}{(2\pi)^2}
\left(
\bm{\mathcal L}_i\bm\Sigma_{\phi,i}^{\rm comp}\bm{\mathcal L}_i\trans
+\bm\Sigma_{\rm dyn,i}^{\rm rad}
\right).
\label{eq:carrier-cycle-covariance}
\end{equation}
If the dynamics covariance is already expressed in cycles squared, the correct form is
\begin{equation}
\bm\Sigma_{i,L}^{\rm comp}
=\frac{1}{(2\pi)^2}
\bm{\mathcal L}_i\bm\Sigma_{\phi,i}^{\rm comp}\bm{\mathcal L}_i\trans
+\bm\Sigma_{\rm dyn,i}^{\rm cyc}.
\label{eq:carrier-cycle-covariance-alt}
\end{equation}
The two forms must not be mixed.

\subsection{Differencing Operators and the Double-Difference Model with Real-Valued Hardware Bias}
Adopt the inter-satellite difference
\begin{equation}
\Delta_{sr}f_q=f_{q,s}-f_{q,r},
\end{equation}
the between-receiver single difference
\begin{equation}
\nabla_{RB}f_s=f_{R,s}-f_{B,s},
\end{equation}
and the RTK double difference
\begin{equation}
\nabla_{RB}\Delta_{sr}f
=(f_{R,s}-f_{R,r})-(f_{B,s}-f_{B,r}).
\label{eq:double-difference-order}
\end{equation}
The ordering \(\nabla_{RB}\Delta_{sr}\) is used consistently throughout.

The restored carrier double-difference model is
\begin{equation}
\lambda\nabla_{RB}\Delta_{sr}\widetilde L
=\nabla_{RB}\Delta_{sr}\rho
+B_{sr}^{\rm atm}
+B_{sr}^{\hw}
+\lambda\nabla_{RB}\Delta_{sr}N
+\widetilde\epsilon_{sr}^{\DD},
\label{eq:double-difference-carrier}
\end{equation}
where
\begin{equation}
\nabla_{RB}\Delta_{sr}N\in\Z.
\end{equation}
The term \(B_{sr}^{\hw}\) may be set to zero only after hardware calibration or when a common value follows rigorously from identical synchronized architectures. Otherwise, it is a real-valued nuisance or deterministic bias and cannot be absorbed into the integer ambiguity.

If the composite array-response phase has not been inversely transported completely, the equivalent single-station inter-satellite bias in cycles and the corresponding RTK double-difference bias are
\begin{equation}
 b_{q,sr}^{A}
=\frac{1}{2\pi}\Delta_{sr}\widetilde\psi_q^A,
\label{eq:single-station-array-cycle-bias}
\end{equation}
and
\begin{equation}
 b_{A,RB}^{sr}
=\frac{1}{2\pi}\nabla_{RB}\Delta_{sr}\widetilde\psi^A,
\qquad
 d_{A,RB}^{sr}=\lambda b_{A,RB}^{sr}.
\label{eq:double-difference-array-bias}
\end{equation}
Phase in radians, bias in cycles, and bias in distance must not be interchanged.

\subsection{Double-Difference Covariance and Inter-Station Correlation}
Let \(\bm\Sigma_{R,L}^{\rm comp}\) and \(\bm\Sigma_{B,L}^{\rm comp}\) be the undifferenced carrier covariances of the rover and base station, let \(\bm\Sigma_{RB}\) be their cross-covariance, and let \(\bm D_{sr}\) be the inter-satellite differencing matrix. Then
\begin{equation}
\bm Q_{\DD}
=\bm D_{sr}
\left(
\bm\Sigma_{R,L}^{\rm comp}
+\bm\Sigma_{B,L}^{\rm comp}
-\bm\Sigma_{RB}-\bm\Sigma_{RB}\trans
\right)
\bm D_{sr}\trans.
\label{eq:double-difference-covariance}
\end{equation}
For a short baseline, the internal receiver noises may often be approximated as independent. If the receivers share an external frequency standard, interference source, network-correction state, or synchronization device, however, \(\bm\Sigma_{RB}\) must be retained. Correlation introduced by the reference satellite must be preserved in matrix form and cannot be reconstructed by concatenating independent single-satellite variances.

\subsection{Float-Ambiguity Information and Deterministic-Bias Propagation}
Write a multisatellite, multiepoch, multifrequency model as
\begin{equation}
\bm y
=\bm H_x\bm x
+\lambda\bm H_N\bm N
+\bm b_y
+\bm\epsilon,
\qquad
\bm\epsilon\sim\calN(\bm0,\bm Q_y),
\quad
\bm N\in\Z^n.
\label{eq:rtk-linear-bias-model}
\end{equation}
The vector \(\bm b_y\) collects residual array phase, uncalibrated hardware bias, timing mismatch, and other deterministic noninteger terms. Let
\begin{equation}
\bm W_y=\bm Q_y^{-1},
\end{equation}
and let \(\bm P_x\succeq0\) represent prior or dynamical information for the real-valued parameters. The bias formula below assumes that this prior is centered at the true \(\bm x\), or is unbiased in expectation. If its mean is displaced, the corresponding prior-mean term must be propagated together with \(\bm b_y\). The weighted residual operator after elimination of \(\bm x\) is
\begin{equation}
\bm M_x
=\bm W_y
-\bm W_y\bm H_x
(\bm H_x\trans\bm W_y\bm H_x+\bm P_x)^{-1}
\bm H_x\trans\bm W_y.
\label{eq:real-parameter-eliminator}
\end{equation}
Assume the relevant real-parameter normal matrix is nonsingular and \(\bm I_N\succ0\). The float-ambiguity information and covariance are
\begin{equation}
\bm I_N
=\lambda^2\bm H_N\trans\bm M_x\bm H_N,
\qquad
\bm Q_{\widehat N}=\bm I_N^{-1}.
\label{eq:float-ambiguity-information}
\end{equation}
The mean bias in the float ambiguity induced by a deterministic observation bias is
\begin{equation}
\bm\beta_N
=\bm I_N^{-1}
\lambda\bm H_N\trans\bm M_x\bm b_y.
\label{eq:float-ambiguity-bias}
\end{equation}
Thus, an array residual propagates along two parallel chains:
\begin{equation}
\bm\Sigma_\phi
\longrightarrow\bm Q_y
\longrightarrow\bm Q_{\widehat N}
\longrightarrow\text{ADOP/lattice distance/random risk},
\label{eq:covariance-propagation-chain}
\end{equation}
and
\begin{equation}
\bm b_\phi
\longrightarrow\bm b_y
\longrightarrow\bm\beta_N
\longrightarrow\text{Voronoi displacement/false-fix risk}.
\label{eq:bias-propagation-chain}
\end{equation}
Propagation of covariance alone is insufficient for an integer-integrity analysis.

\begin{theorem}[Monotonicity of the float-ambiguity covariance under observation-covariance improvement]
Assume that the design matrices and prior information are fixed, the joint model has full rank, and \(\bm Q_{y,0},\bm Q_{y,1}\succ0\). If
\begin{equation}
\bm Q_{y,1}\preceq\bm Q_{y,0},
\label{eq:observation-covariance-improvement}
\end{equation}
then the joint parameter covariance and its ambiguity principal submatrix do not increase:
\begin{equation}
\bm Q_{\widehat N,1}\preceq\bm Q_{\widehat N,0}.
\label{eq:ambiguity-covariance-monotonicity}
\end{equation}
\label{thm:ambiguity-cov-monotonicity}
\end{theorem}
\begin{proof}
The covariance order implies the reverse information order, \(\bm Q_{y,1}^{-1}\succeq\bm Q_{y,0}^{-1}\). Hence the full joint normal matrix for \((\bm x,\bm N)\) increases in the semidefinite order. Matrix inversion reverses that order, and taking the ambiguity principal submatrix preserves it, which gives \Cref{eq:ambiguity-covariance-monotonicity}.
\end{proof}
This theorem concerns covariance only. If a weight improves the covariance while increasing \(\bm\beta_N\), the correct-fix probability may still decrease.

\subsection{Integer Voronoi Geometry and Biased False-Fix Risk}
Define the integer Voronoi cell under the ILS metric by
\begin{equation}
\calV_{\bm0}
=\left\{\bm v\in\R^n:
\bm v\trans\bm Q_{\widehat N}^{-1}\bm v
\le
(\bm v-\bm z)\trans\bm Q_{\widehat N}^{-1}(\bm v-\bm z),
\ \forall\bm z\in\Z^n\setminus\{\bm0\}
\right\}.
\label{eq:integer-voronoi-cell}
\end{equation}
If the float-ambiguity error is
\begin{equation}
\widehat{\bm N}-\bm N
=\bm\beta_N+\bm\varepsilon_N,
\qquad
\bm\varepsilon_N\sim\calN(\bm0,\bm Q_{\widehat N}),
\label{eq:biased-float-ambiguity}
\end{equation}
then, for an unvalidated nearest-lattice-point ILS decision, the true integer is among the nearest candidates exactly on the event
\begin{equation}
\bm\beta_N+\bm\varepsilon_N\in\calV_{\bm0}.
\label{eq:correct-fix-event}
\end{equation}
A positive-definite Gaussian model assigns zero probability to the Voronoi boundary, so the nearest candidate is unique almost surely. A validation test subsequently partitions candidate selection into correct fixing, float retention, and wrong fixing.

In the absence of random error, the true integer lies in the strict interior of its Voronoi region and is the unique ILS decision if and only if, for every \(\bm z\neq\bm0\),
\begin{equation}
\abs{\bm z\trans\bm Q_{\widehat N}^{-1}\bm\beta_N}
<\frac12
\bm z\trans\bm Q_{\widehat N}^{-1}\bm z.
\label{eq:voronoi-bias-condition}
\end{equation}
Define
\begin{equation}
 d_{\bm z}^2
=\bm z\trans\bm Q_{\widehat N}^{-1}\bm z,
\qquad
 d_{\min}^2
=\min_{\bm z\in\Z^n\setminus\{0\}}d_{\bm z}^2.
\label{eq:integer-lattice-distance}
\end{equation}
A concise sufficient condition is
\begin{equation}
\norm{\bm\beta_N}_{\bm Q_{\widehat N}^{-1}}
<\frac12d_{\min}.
\label{eq:bias-norm-sufficient}
\end{equation}

For a competing integer direction \(\bm z\), the biased pairwise boundary-crossing probability is
\begin{equation}
P_{\bm z}(\bm\beta_N)
=Q\!\left(
\frac{\tfrac12d_{\bm z}^2
-\bm z\trans\bm Q_{\widehat N}^{-1}\bm\beta_N}
{d_{\bm z}}
\right),
\label{eq:biased-pairwise-error}
\end{equation}
where \(Q(\cdot)\) is the right-tail probability of a standard Gaussian random variable. This is a one-sided crossing probability toward the specified competitor \(\bm z\), not the total false-fix probability. It nevertheless shows how a bias projection along a shortest lattice direction can sharply increase a dominant pairwise error mechanism even when covariance and ADOP improve.
For the unvalidated nearest-lattice ILS decision, the union bound gives
\begin{equation}
\Pr\{\text{wrong ILS candidate}\}
\le
\min\!\left\{1,
\sum_{\bm z\in\Z^n\setminus\{\bm0\}}
P_{\bm z}(\bm\beta_N)\right\}.
\label{eq:wrong-candidate-union-bound}
\end{equation}
Practical evaluation may truncate the sum only after bounding the omitted tail.

\subsection{ADOP, Bootstrap Success Rate, and Their Limitations}
The ambiguity dilution of precision is
\begin{equation}
\mathrm{ADOP}
=\det(\bm Q_{\widehat N})^{1/(2n)}.
\label{eq:adop}
\end{equation}
After LAMBDA integer decorrelation, if the conditional standard deviations are \(\sigma_{i|I}\), the bootstrap success rate is
\begin{equation}
P_B
=\prod_{i=1}^{n}
\left[2\Phi\!\left(\frac{1}{2\sigma_{i|I}}\right)-1\right],
\label{eq:bootstrap-success}
\end{equation}
where \(\Phi\) is the standard Gaussian cumulative distribution function\cite{teunissen1995,teunissen1998}. ADOP and \(P_B\) mainly describe zero-mean covariance geometry. When \(\bm\beta_N\neq0\), they must be used jointly with \Cref{eq:voronoi-bias-condition,eq:biased-pairwise-error}.

An RTK decision must distinguish correct fixing, retention of a float solution, and false fixing. Let \(P_{\rm correct}\), \(P_{\rm float}\), and \(P_{\rm wrong}\) denote these mutually exclusive outcome probabilities. One may define the risk
\begin{equation}
\calR(\bmw)
=P_{\rm correct}(\bmw)L_c
+P_{\rm float}(\bmw)L_f
+P_{\rm wrong}(\bmw)L_w,
\qquad
L_w\gg L_f\ge L_c,
\label{eq:rtk-decision-risk}
\end{equation}
or constrain the false-fix probability directly. Increasing the nominal fix rate while enlarging the false-fix tail is not a performance improvement.

\subsection{RTK Evaluation Interface for a Single Common Weight Vector}
For a candidate common weight vector \(\bmw\), estimate the per-satellite residual mean \(\bm b_\phi(\bmw)\) only after absolute anchor closure, and estimate the relative covariance \(\bm\Sigma_\phi(\bmw)\) from the same-source auxiliary branches with all cross-covariances retained. Then use \Cref{eq:carrier-cycle-covariance,eq:double-difference-covariance,eq:float-ambiguity-bias} to obtain \(\bm Q_{\widehat N}(\bmw)\) and \(\bm\beta_N(\bmw)\). Relative branches alone cannot determine the absolute residual mean. The basic theory can provide
\begin{equation}
J_N(\bmw)=\log\det\bm Q_{\widehat N}(\bmw),
\qquad
J_\beta(\bmw)=\norm{\bm\beta_N(\bmw)}_{\bm Q_{\widehat N}^{-1}},
\label{eq:ambiguity-weight-objectives}
\end{equation}
together with a biased false-fix upper bound. Candidate weights should also satisfy
\begin{align}
\bmc_{\PI}\herm\bmw&=1,
\label{eq:weight-pi-constraint}\\
\bmw\herm\bm R_{J+N}\bmw&\le\xi_J,
\label{eq:weight-interference-constraint}\\
C/N_{0,s}(\bmw)&\ge\Gamma_{s,\min},
\quad s\in\calS_{\rm tr},
\label{eq:weight-cn0-constraint}\\
\abs{g_s(\bmw)}&\ge\varepsilon_s,
\quad
\chi_{wa,s}(\bmw,\bmw_a)\ge\chi_{\min},
\label{eq:weight-phase-safety}\\
\kappa_{w,s}(\bmw)&\le\bar\kappa_{w,s}.
\label{eq:weight-conditioning-constraint}
\end{align}
All satellites share the same common weight vector. Per-satellite coherent statistics are used only to evaluate its combined effect on the satellites and on the integer model. The specific optimizer belongs to an upper-level active-observation control theory.

\subsection{Distinction from Multi-Antenna ARTK}
Multi-antenna ARTK retains carrier observations from multiple physical antennas, inter-antenna integers, and known rigid geometry, and can form a high-precision virtual antenna by data reduction\cite{li2025artk}. In the present architecture, element data are combined into the primary output before acquisition or tracking, and auxiliary branches do not establish independent physical integers. The information sources are different: ARTK exploits complete multi-antenna observation redundancy, whereas SN-ASMO exploits single-output interference suppression, same-source composite-phase restoration, and task-oriented observation-quality evaluation. The two can be combined hierarchically, but they cannot be counted twice within the same statistical model.
\section{Materials and Methods}
The observable quantities and test protocols below separate algebraic unit tests from empirical predictions. They do not report experiments not conducted; they define the repeated trials, controls, uncertainty records, and failure criteria required to test the theory.

\subsection{Algebraic and Software Unit Tests}
The following identities can be checked sample by sample in double-precision simulations and in FPGA or software-receiver unit tests.

\begin{enumerate}
\item \textbf{Internal-reference-element invariance:} randomly generate nonzero \(\bmc_s^{(O)}\) and \(\bmw\), and for every valid pair \(\mu,\nu\) verify
\begin{equation}
\frac{\abs{\beta_{\mu,s}^{(O)}\gamma_{\mu,s}
-\beta_{\nu,s}^{(O)}\gamma_{\nu,s}}}
{\abs{g_s}}
\le\epsilon_{\rm num}.
\label{eq:test-chart-invariance}
\end{equation}
\item \textbf{Paired phase-coordinate invariance:} for a random admissible physical-reference phase field \(\chi=\chi(\bmxi)\), independent of \(\bmeta\), verify
\begin{equation}
\abs{
\ee^{-\jj\chi}\alpha_s^{\nav}
\ee^{\jj\chi}g_s
-\alpha_s^{\nav}g_s}
\le\epsilon_{\rm num}\abs{\alpha_s^{\nav}g_s}.
\label{eq:test-paired-gauge}
\end{equation}
\item \textbf{Exact same-statistic replacement:} for arbitrary nonzero complex numbers \(z_a,z_w\), verify the relative error of \Cref{eq:exact-phase-replacement}. This test is independent of SNR.
\item \textbf{Tangent-space consistency:} generate random \(\delta\bmtheta^n,\delta\bmu_s^n\), verify \Cref{eq:los-attitude-chain} by finite differences, and check that \((\bmu_s^a)\trans\delta\bmu_s^a=O(\norm\delta^2)\).
\item \textbf{Complex-projection consistency:} compare the kernel, rank, and projection residual of the real nuisance projection in \Cref{eq:intrinsic-jacobian} with those of the complex projection in \Cref{eq:projected-complex-jacobian}.
\end{enumerate}

If any identity fails systematically after numerical error has been controlled, the reference convention, complex-conjugation direction, or attitude coordinates used in the implementation are inconsistent with the theory.

\subsection{Static-Interference and Weight-Variation Experiment}
The first objective is to test \Cref{prop:gnss-dual-axis-witness} directly. Use repeated matched trials or reproducible RF replay. Within each paired trial, hold the physical condition, the non-weight formation coordinates, and one coherent element-sample block fixed while applying two synchronized admissible weights digitally. Estimate the conditional means across matched trials and test whether
\begin{equation}
\widehat{\E}[z_{s,0}-z_{s,1}\mid\bmxi_0]
\approx
\widehat\alpha_s^{\nav}(\bmxi_0)
(\bmw_0-\bmw_1)\herm\widehat\bmc_s^{(O)}(\bmxi_0)\neq0
\label{eq:empirical-dual-axis-witness}
\end{equation}
within the declared calibration and sampling uncertainty. This same-block comparison, rather than an inter-epoch sample difference, is the empirical witness that the observer-state axis cannot be deleted.

Next hold the receiving platform and interference direction fixed. Use live-sky signals, a GNSS simulator, or recorded replay to generate coherent multisatellite element data, while slowly varying the interference power or switching the PI weights according to a prescribed sequence. Record synchronously
\begin{equation}
\{\bmw[n],\ \widehat\bmc_s^{(O)}[n],\ z_{w,s}[n],\ z_{a,s}[n],\ C/N_{0,s}[n]\},
\end{equation}
where \(\widehat\bmc_s^{(O)}\) is a calibration estimate with reported uncertainty, not a directly observed true response. The model-predicted phase is
\begin{equation}
\psi_{s,\rm mod}^A[n]
=\Arg(\bmw[n]\herm\widehat\bmc_s^{(O)}[n]),
\label{eq:static-predicted-phase}
\end{equation}
and the same-source measured relative phase is
\begin{equation}
\psi_{s,\rm rel}^A[n]
=\Arg(z_{w,s}[n]z_{a,s}^{*}[n]).
\label{eq:static-measured-relative-phase}
\end{equation}
After closure of the anchor phase, compare the residual
\begin{equation}
 r_{\psi,s}[n]
=\operatorname{wrap}_{\pi}
(\psi_{s,\rm obs}^A[n]-\psi_{s,\rm mod}^A[n]).
\label{eq:static-phase-residual}
\end{equation}
The acceptance condition is not a prespecified universal numerical threshold. Rather, the residual mean should be consistent with array-calibration uncertainty, its variance should be consistent with the joint-covariance prediction in \Cref{eq:transport-phase-variance}, and a change of internal-reference-element chart should not change the result. If different internal reference elements produce statistically significant differences in the composite phase, then \Cref{thm:chart-invariance} is falsified at the implementation level.

\subsection{High-Dynamic Turntable and Attitude-Phase-Dynamics Experiment}
Mount the array on a three-axis turntable with measured angular velocity and attitude. Maintain traceable satellite directions using a GNSS simulator or live sky, and hold the interference field fixed or measurable. Compute the prediction of \Cref{eq:element-phase-dynamics} from attitude, ephemerides, and baselines, and estimate the phase rate from coherent element observations and composite responses.

For the relative baseline of element \(m\), define the normalized dynamics residual
\begin{equation}
 r_{m\mu,s}^{\rm dyn}(t)
=\dot\phi_{m\mu,s}^{\rm meas}(t)
+\kappa\left[
(\dot{\bmu}_s^n)\trans\bm r_{m\mu}^n
+(\bm r_{m\mu}^n\times\bmu_s^n)\trans\bm\omega_{a/n}^n
\right].
\label{eq:dynamic-test-residual}
\end{equation}
Pure translation, pure rotation, and mixed motion should be tested separately to verify the separation between propagation Doppler and relative array-phase dynamics. If the inverse-transported carrier remains systematically correlated with turntable angular velocity beyond calibration uncertainty while the normal propagation geometry is correctly modeled, then the composite-response dynamics model is not closed.

\subsection{Injected Branch-Timing-Mismatch Experiment}
In a digital receiver, impose a known sample delay \(\delta t\) on either the working or anchor branch while leaving all other processing unchanged. For a sequence of positive and negative delays, fit the residual phase to
\begin{equation}
 r_{\rm sync}(\delta t)
=a_1\delta t+a_2\delta t^2+o(\delta t^2).
\label{eq:timing-fit}
\end{equation}
The theory predicts
\begin{equation}
 a_1=2\pi f_{D,s}+\dot{\widetilde\psi}_{w,s}^{A},
\qquad
 a_2=\frac12(2\pi\dot f_{D,s}+\ddot{\widetilde\psi}_{w,s}^{A}).
\label{eq:timing-fit-prediction}
\end{equation}
Verifying only the geometric-Doppler coefficient under static conditions, while neglecting the high-dynamic array-phase rate, is insufficient to validate \Cref{eq:timing-mismatch-second-order}.

\subsection{Near-Null Response and Phase-Conditioning Experiment}
Use a controlled weight trajectory to drive \(\abs{g_s}\) for one satellite progressively toward, but not through, its safety threshold. Record phase variance, sensitivity to weight perturbations, and tracking status. In the regular high-SNR local-Gaussian regime, before circular-phase collapse, the theory predicts
\begin{equation}
\Var(\widehat\psi_s^A)
\propto\frac{1}{\abs{g_s}^2},
\qquad
\kappa_{w,s},\kappa_{c,s}
\propto\frac{1}{\abs{g_s}},
\label{eq:near-null-prediction}
\end{equation}
and the norm of the dynamic equivalent phase center in \Cref{eq:dynamic-phase-center-closed} may increase rapidly. If the measured phase distribution is visibly heavy-tailed or nearly uniform, Gaussian fitting should be stopped and coherence, circular statistics, and loss-of-lock rate should be reported instead.

\subsection{Random-Weight and Conditional-Fisher-Information Experiment}
Design two processing modes:
\begin{enumerate}
\item \textbf{Data separation:} estimate the weights from an independent training block and estimate the state from a test block;
\item \textbf{Same-block adaptation:} derive both the weights and the state from the same noisy block.
\end{enumerate}
Compare the empirical error covariance with three predictions: the Fisher lower bound conditioned on given weights, a marginal approximation that ignores weight randomness, and a joint weight--state model. If the empirical error under same-block adaptation is systematically larger than the conditional Fisher prediction, while the joint model explains the discrepancy, the necessity of \Cref{ass:conditional-weight} is confirmed.

For a standard observation \(\bm y_0\) and an auxiliary observation \(\bm y_a\) sharing raw samples, compare
\begin{equation}
\bm F_0+\bm F_a
\quad\text{and}\quad
\bm F_0+\bm F_{a|0}.
\end{equation}
Monte Carlo error covariance or an empirical CRLB should confirm that only the conditional increment avoids double counting.

\subsection{RF-Replay and Atomic-Weight-Commit Experiment}
Use a synchronized multichannel RF recording and retain the sample ID of every element. For each sample batch, compute the active- and candidate-weight prompts in parallel and, at the same sample boundary, commit
\begin{equation}
(\bmw_{\rm active},\ \widehat T_{\rm active\to candidate,s},\ k_s^{\rm lift})
\longmapsto
(\bmw_{\rm candidate},\ \bm I,\ k_s^{\rm lift,new}).
\label{eq:atomic-commit-state}
\end{equation}
Compare this implementation with an intentionally non-atomic implementation displaced by one or more samples. The theory predicts no additional phase step under a correct atomic commit, whereas the displaced implementation should exhibit a residual dependent on Doppler and array dynamics according to \Cref{eq:timing-mismatch-second-order}.

\subsection{RTK Bias--Covariance and Integer-Risk Experiment}
Inject a known array residual \(\bm b_y\) and correlated noise \(\bm Q_y\) into base and rover observations, or obtain these quantities from real array replay. Verify three levels of prediction:
\begin{enumerate}
\item the mean float-ambiguity bias agrees with \Cref{eq:float-ambiguity-bias};
\item the float covariance agrees with \Cref{eq:float-ambiguity-information}; and
\item the trend and order of empirical boundary-crossing probability along competing integer directions agree with \Cref{eq:biased-pairwise-error}.
\end{enumerate}
Report the fix rate, false-fix rate, float-retention rate, ADOP, \(d_{\min}\), \(\norm{\bm\beta_N}_{\bm Q^{-1}}\), and the most dangerous integer direction. A fix rate alone cannot test the theory of bias risk.

\subsection{Core Auditable Claims}
The chart and deterministic-rotation rows below are algebraic/software identities; the remaining rows are model-conditional empirical predictions. Keeping these classes separate prevents a numerical identity check from being presented as an independent physical validation.
\begin{table}[htbp]
\centering
\caption{Core SN-ASMO claims and their audit or falsification conditions}
\label{tab:falsifiability}
\small
\begin{tabularx}{\textwidth}{>{\raggedright\arraybackslash}p{0.24\textwidth}>{\raggedright\arraybackslash}X>{\raggedright\arraybackslash}X}
\toprule
Theoretical proposition & Observable prediction & Observation constituting falsification after assumptions are controlled\\
\midrule
Nonphysicality of the internal reference element & Chart switching leaves \(g_s\), the restored phase, and RTK results unchanged & A valid chart switch changes the composite output systematically\\
GNSS dual-axis strong witness & At fixed \(\bmxi_0\) and a common sample block, admissible synchronized weights with \((\bmw_0-\bmw_1)\herm\bmc_s^{(O)}\neq0\) produce distinct conditional means & The outputs remain equal beyond uncertainty despite a verified nonzero predicted response difference\\
Synchronized same-source cancellation & The common navigation complex factor vanishes from the conjugate product & A term correlated with the common navigation phase remains under strict synchronization and common sourcing\\
Deterministic \(\Uone\) invariance & Noise-sample magnitude and covariance spectrum are preserved; Fisher information is preserved for a known parameter-independent rotation & Such a rotation changes noise power, or changes Fisher information under the stated conditions\\
Near-null degradation & Phase variance and condition numbers grow as \(\abs g^{-1}\) or \(\abs g^{-2}\) & No sensitivity growth occurs near a response zero under the same noise model\\
Conditional Fisher increment & Added information from correlated branches equals conditional, not marginally added, information & The conditional decomposition fails under regular joint-Gaussian conditions\\
Integer-bias risk & Float bias along a shortest lattice direction increases false-fix probability & A controlled bias change is systematically inconsistent with the empirical decision boundary\\
\bottomrule
\end{tabularx}
\end{table}

An experimental report must state the model-validity conditions, front-end linear range, synchronization accuracy, data dependence of the weight estimate, calibration uncertainty, and all data-exclusion rules. Only after these quantities have been observed and documented does a discrepancy between prediction and data have a well-defined falsifying meaning.
\section{Discussion, Theoretical Boundaries, and Conclusions}
\subsection{Domain of Validity}
SN-ASMO does not force all array-GNSS problems into one fixed-dimensional global manifold. It provides rigorous interfaces among observation formation, reference covariance, group actions, statistical information, and the integer structure. The following boundaries are constitutive parts of the theory, not ancillary disclaimers.

\begin{enumerate}
\item \textbf{Physical invariance does not imply numerical invariance of every observation.} The observer state may change the recorded complex statistic and the information it retains; the theory preserves the physical propagation content through explicit covariance and transport laws.
\item \textbf{High-dynamic platform motion has two physical pathways.} It changes range and Doppler through \(\alpha_s^{\nav}\) and changes the composite response through \(D_{\bmxi}g_s\). Inverse transport acts on the declared response phase; it does not invert the navigation-propagation factor.
\item \textbf{Ordinary satellite motion is not an independent ``erroneous phase source.''} Orbital motion enters primarily through the normal propagation phase; its slow change of direction of arrival is merely a physical parameter variation of the point on the array manifold.
\item \textbf{Relative observability does not imply absolute identifiability.} Under the unconstrained common-complex-factor auxiliary model, purely relative same-source branches determine only a projective class. Absolute restoration must be closed by a model, a calibrated anchor, or an informed joint state.
\item \textbf{A group action must correspond to a genuine degree of freedom.} The primary carrier phase and integer ambiguity are quantities of interest. Over-quotienting removes genuine information according to \Cref{eq:overquotient-rank-loss}.
\item \textbf{An exact \(\Uone\) inverse action does not create noise reduction.} Statistical improvement comes from spatial suppression, correlated common-mode cancellation, temporal averaging, or improved joint estimation. Estimated phase errors must be propagated explicitly.
\item \textbf{Local geometry cannot be globalized unconditionally.} Constant rank, complex logarithms, projective atlases, continuous lifts, and Fisher information are valid only on their respective regular domains.
\item \textbf{A single output does not create physical dimensions.} If auxiliary statistics are not retained, the gain is primarily a covariance improvement. Only statistics retained legitimately and modeled jointly can contribute new conditional information.
\item \textbf{A narrowband carrier theory does not replace a broadband propagation theory.} Group delay, code-correlation shape, beam squint, and near-field wavefronts require an expanded observation space.
\item \textbf{Front-end nonlinearity is an irreversible boundary.} Information lost through saturation, clipping, or irreversible quantization cannot be recovered by a downstream phase-group action.
\item \textbf{Fisher information is not a complete performance characterization.} Low SNR, heavy tails, multimodality, integer boundaries, and finite samples still require the complete likelihood, Monte Carlo analysis, and physical experiments.
\item \textbf{Covariance improvement does not guarantee improvement in correct fixing.} If the deterministic float bias increases along a dangerous lattice direction, false-fix risk may still rise.
\end{enumerate}

\subsection{Precise Statement of Theoretical Novelty}
Previous research has discussed and validated carrier-phase bias caused by array spatial processing, run-time estimation and compensation, low-distortion processing, and RTK under interference\cite{church2009,obrien2008,xie2021,li2023,wang2024,wang2025,bamberg2023}. Accordingly, this paper does not claim any single compensation method, the \(\Uone\) group, projective space, Fisher information, or RTK integer theory as an isolated invention.

The contribution is not any one of these established ingredients in isolation, but their typed composition within precise satellite-navigation array observation. Starting from the protocol-supplied roles, the fixed-physical-state GNSS witness exhibits the irreducible formation-side dependence, and the behavioral quotient gives that dependence its minimal observer-state identity. SN-ASMO then derives the admissible dual-axis interface, the complete response differential, and the response-tracking PTFE. It states the boundaries among path transport, endpoint algebra, relative alignment, and absolute closure, and carries residual mean and covariance to RTK cycle-continuity and integer-risk conditions. Each step is tied to an explicit regularity, synchronization, or identifiability condition.

\subsection{Conditional Compatibility and RTK Performance}
\begin{theorem}[Sufficient conditions for compatible anti-jamming, high-dynamic reception, and RTK preservation]
Assume that a common weight sequence exists that satisfies the PI, interference-suppression, \(C/N_0\), response-amplitude, transport-coherence, and conditioning constraints in \Cref{eq:weight-pi-constraint,eq:weight-interference-constraint,eq:weight-cn0-constraint,eq:weight-phase-safety,eq:weight-conditioning-constraint}. Assume also:
\begin{enumerate}
\item the front end remains linear; the stated narrowband model, coherent-window approximation, and synchronization-error budget remain valid under the platform dynamics;
\item every continuous update follows an admissible path in the nonzero safe domain, and both endpoints of any genuinely atomic switch are nonzero;
\item same-source branches share the declared samples, clock, NCOs, data-bit processing, and integration operator, and the exact transport factor \(T_{aw,s}^{A}\) is used;
\item the weight, inverse factor, phase-lift counter, and PLL state are committed atomically; an absolute phase closure supplies a lift consistent with the current tracking arc and is inversely transported continuously before the discriminator; and no loss of lock, data-bit error, or unsynchronized cycle-count jump occurs.
\end{enumerate}
Then interference suppression, carrier-phase continuity under the admitted dynamics, and preservation of the standard carrier model and integer structure are mutually compatible within the same receiver chain.
\label{thm:conditional-simultaneous-rtk}
\end{theorem}
\begin{proof}
The weight constraints provide interference suppression and usable satellite amplitudes. On continuous updates, the nonzero-path hypothesis invokes \Cref{thm:snasmo-ptfe}; on atomic switches, the stated exact same-source endpoint factor applies. Synchronization, exact inverse transport, and atomicity then invoke \Cref{thm:relative-cycle-continuity}, while the absolute closure and tracking hypotheses invoke \Cref{thm:integer-preservation}. Thus the interference-suppression operation does not alter the standard carrier parameterization or its integer structure under the admitted dynamics.
\end{proof}

The theorem is conditional on a common sequence satisfying the listed constraints; it neither constructs that sequence nor asserts that the joint feasible set is nonempty for arbitrary interference-to-signal ratio or platform dynamics.

\begin{corollary}[Conditional high-precision RTK performance]
Under the hypotheses of \Cref{thm:conditional-simultaneous-rtk}, assume in addition that \(\bm Q_y\succ0\), the joint RTK design is full rank, and \(\bm I_N\succ0\). If the propagated position bias, position covariance, integer bias \(\bm\beta_N\), ambiguity covariance \(\bm Q_{\widehat N}\), and validation risk satisfy the declared application and integrity thresholds, then the same receiver chain meets the corresponding high-precision RTK criteria under the admitted interference and dynamics.
\label{cor:conditional-high-precision-rtk}
\end{corollary}
\begin{proof}
Under the preserved standard carrier model, the linear statistical model propagates residual means and covariances through \Cref{eq:float-ambiguity-information,eq:float-ambiguity-bias}. The stated threshold conditions then imply the declared accuracy and integrity criteria by their definitions.
\end{proof}

Under the Gaussian float-error model in \Cref{eq:biased-float-ambiguity}, condition on fixed \(\bm\beta_N\) and \(\bm Q_{\widehat N}\succ0\), and let \(0<\delta<1\). A model-checkable sufficient condition is
\begin{equation}
\norm{\bm\beta_N}_{\bm Q_{\widehat N}^{-1}}
+\sqrt{\chi^2_{n,1-\delta}}
<\frac12d_{\min},
\label{eq:probabilistic-voronoi-sufficient}
\end{equation}
where \(\chi^2_{n,1-\delta}\) is the \(1-\delta\) quantile of a chi-squared variable with \(n\) degrees of freedom. On the event
\(\norm{\bm\varepsilon_N}_{\bm Q_{\widehat N}^{-1}}\le\sqrt{\chi^2_{n,1-\delta}}\), the triangle inequality and \Cref{eq:bias-norm-sufficient} place the float error strictly inside the true Voronoi cell. Hence the true integer is the unique unvalidated ILS solution with probability at least \(1-\delta\). A subsequent validation gate requires its own acceptance analysis. If \(\bm\beta_N\) or \(\bm Q_{\widehat N}\) is estimated from data, its estimation uncertainty must be propagated before this condition is used.

\subsection{Principal Conclusions}
The fixed-physical-state weight witness shows that the complete GNSS formation map cannot factor through \(\calX\) alone. The behavioral quotient yields the minimal observer-state space \(\calC_{\rm obs}\), and \(\calH\) on \(\calE\) records the joint dependence without asserting independence or unrestricted product admissibility. The differential of \(\calH\) separates the physical dependence of both \(\alpha_s^{\nav}\) and \(g_s\) from observer-state dependence through \(D_{\bmeta}g_s\).

On \(g_s\neq0\), the normalized response determines the response one-form. With the standard radian normalization and response-tracking compatibility fixed, the form determines the compatible connection and the PTFE. Internal-reference-element chart changes leave the complete response invariant, whereas physical-reference-point changes are paired covariantly with the propagation factor. PTFE transport requires an admissible nonzero path; a ratio of nonzero endpoint phases alone is an algebraic endpoint relation.

After only physically and statistically justified nuisance freedoms are removed, the intrinsic Jacobian and, under the common-complex-gain auxiliary model, the pullback Fubini--Study metric are governed by the same horizontal projection. A known parameter-independent \(\Uone\) rotation preserves noise power and Fisher information but does not itself improve either. Correlated branches contribute through conditional, rather than marginally added, information.

For RTK, a circular phase branch, its continuous real-valued lift, and the PLL integer cycle count are distinct objects. Synchronized exact relative inverse transport preserves branch-to-branch cycle continuity, whereas the standard carrier observation additionally requires absolute phase closure. Residual array error must therefore be propagated through both mean and covariance; improved ADOP alone cannot establish a lower false-fix risk.

\subsection{Physical Interpretation}
The satellite--platform--medium dynamics define an objective propagation process, while array geometry, channels, synchronization, and weights determine the complex statistic recorded by the receiver. A change of observer state need not leave that statistic numerically invariant and may also alter the information retained. Objectivity is therefore expressed through the stated covariance and preservation laws, not through equality of all observer-dependent outputs.

Quotienting removes only freedoms shown to be unidentifiable for the declared task; phase transport compares response phases only along admissible nonzero paths. This separation prevents the observer-state-induced component of the response phase from being mistaken for propagation range while preserving physical motion, continuous carrier phase, and the RTK integer structure.

\subsection{Future Theoretical Problems}
Future work may extend the framework toward a joint geometry of phase, group delay, and code-correlation shape for broadband space--time arrays; hierarchical group actions for shared multifrequency and multisatellite freedoms; intrinsic statistical bounds under low SNR, non-Gaussian interference, and multimodal likelihoods; robust quotient structures under multipath, blockage, pattern, and mutual-coupling uncertainty; and hierarchical interfaces between SN-ASMO and inertial navigation, multi-antenna ARTK, integrity monitoring, and SN-ASOC active-observation control.

In summary, SN-ASMO establishes conditional compatibility among interference suppression, high-dynamic array reception, and preservation of carrier-phase RTK structure. Under the stated feasible-weight, front-end, synchronization, nonzero-response, phase-closure, and tracking assumptions, it preserves the standard carrier observation and integer parameterization. When the propagated position and integer metrics additionally satisfy the declared accuracy and integrity thresholds, the resulting chain meets the corresponding high-precision RTK criteria; the theory does not assert feasibility for arbitrary interference or dynamics.
\appendix
\section{Supplementary Mathematical Derivations}
\subsection[Cocycle Condition and Differential Invariance Under Chart Changes]{Cocycle Condition and Differential Invariance\\Under Internal-Reference-Element Chart Changes}
Let the transition factor between three valid internal reference elements \(\mu,\nu,\ell\) be
\begin{equation}
 t_{\nu\mu,s}=a_{\nu,s}^{[\mu]}.
\end{equation}
By definition,
\begin{equation}
 t_{\ell\mu,s}=t_{\ell\nu,s}t_{\nu\mu,s},
\label{eq:cocycle-condition}
\end{equation}
so the local internal-reference-element atlas satisfies the cocycle condition. Moreover,
\begin{equation}
\beta_{\nu,s}^{(O)}
=\beta_{\mu,s}^{(O)}t_{\nu\mu,s},
\qquad
\gamma_{\nu,s}=t_{\nu\mu,s}^{-1}\gamma_{\mu,s},
\end{equation}
and hence
\begin{equation}
 g_s=\beta_{\mu,s}^{(O)}\gamma_{\mu,s}
=\beta_{\nu,s}^{(O)}\gamma_{\nu,s}.
\end{equation}
Taking complex-logarithmic differentials gives
\begin{equation}
\frac{\dd\beta_{\nu,s}^{(O)}}{\beta_{\nu,s}^{(O)}}
+\frac{\dd\gamma_{\nu,s}}{\gamma_{\nu,s}}
=
\frac{\dd\beta_{\mu,s}^{(O)}}{\beta_{\mu,s}^{(O)}}
+\frac{\dd\gamma_{\mu,s}}{\gamma_{\mu,s}}
=\frac{\dd g_s}{g_s}.
\label{eq:cocycle-log-invariance}
\end{equation}
Thus, both the composite log-amplitude differential and the phase one-form are invariant under a chart change.

\subsection{Coherent-Window Freezing Remainder}
In the expected satellite term of \Cref{eq:exact-correlator}, expand about the window center \(t_n\):
\begin{equation}
F_s(t)
=F_s(t_n)+\dot F_s(t_n)(t-t_n)
+\frac12\ddot F_s(\xi_t)(t-t_n)^2,
\end{equation}
where \(\xi_t\) lies between \(t\) and \(t_n\). Multiplying by \(q_{s,n}^{*}(t)\) and integrating, and using \(m_1=0\), yields
\begin{equation}
\int q_{s,n}^{*}(t)F_s(t)\,\dd t
=m_0F_s(t_n)+r_{s,n}^{\rm frz}.
\end{equation}
The triangle inequality immediately gives \Cref{eq:freeze-remainder-bound}. If the correlator kernel does not have zero first moment, the term \(m_1\dot F_s(t_n)\) must also be retained.

\subsection{Topological Interpretation of the Global Phase-Lift Criterion}
\label{subsec:lift-proof}
The normalized map
\begin{equation}
 h=\frac{g}{\abs g}:D\to S^1
\end{equation}
admits a real-valued lift \(\widetilde\psi:D\to\R\) if and only if the induced homomorphism on the fundamental group,
\begin{equation}
 h_*:\pi_1(D)\to\pi_1(S^1)\cong\Z,
\end{equation}
is zero. For a closed loop \(\gamma\), the corresponding integer is precisely
\begin{equation}
 h_*([\gamma])
=\frac{1}{2\pi}\oint_\gamma\Ima\!\left(\frac{\dd g}{g}\right).
\end{equation}
This gives \Cref{eq:global-lift-condition}. If \(D\) is simply connected, then \(\pi_1(D)=0\), and the lift exists automatically\cite{lee2013,kobayashi1963}.

\subsection{Phase Covariance Under Proper Complex Noise}
Let
\begin{equation}
 z_i=A_i\ee^{\jj\phi_i}+n_i,
\qquad
 x_i=n_i\ee^{-\jj\phi_i}.
\end{equation}
When \(\abs{n_i}/A_i\) is sufficiently small,
\begin{equation}
\Arg(z_i)
=\phi_i+
\Ima\log\!\left(1+\frac{x_i}{A_i}\right)
=\phi_i+\frac{\Ima x_i}{A_i}
+O_p\!\left(\frac{\abs{n_i}^2}{A_i^2}\right).
\label{eq:arg-expansion-appendix}
\end{equation}
For a jointly proper complex random vector, \(\E[x_ix_j]=0\), and
\begin{equation}
\E[x_ix_j^{*}]
=\ee^{-\jj(\phi_i-\phi_j)}C_{ij}.
\end{equation}
Using
\begin{equation}
\cov(\Ima x_i,\Ima x_j)
=\frac12\Rea\E[x_ix_j^{*}],
\end{equation}
one obtains \Cref{eq:phase-error-covariance}, and subsequently \Cref{eq:transport-phase-variance}.

\subsection{Closed-Form Expression for the Dynamic Equivalent Phase Center}
The ideal composite response is
\begin{equation}
 g_s
=\sum_{m=1}^{M}w_m^{*}\Gamma_{m,s}
\ee^{-\jj\kappa(\bm b_m^a)\trans\bmu_s^a}.
\end{equation}
Its directional gradient is
\begin{equation}
\nabla_{\bmu}g_s
=-\jj\kappa
\sum_{m=1}^{M}w_m^{*}\Gamma_{m,s}
\ee^{-\jj\kappa(\bm b_m^a)\trans\bmu_s^a}\bm b_m^a.
\end{equation}
Since
\begin{equation}
\nabla_{\bmu}\widetilde\psi_s^A
=\Ima\!\left(\frac{\nabla_{\bmu}g_s}{g_s}\right)
=-\kappa\Rea\!\left[
\frac{\sum_mw_m^{*}\Gamma_{m,s}
\ee^{-\jj\kappa(\bm b_m^a)\trans\bmu_s^a}\bm b_m^a}
{g_s}
\right],
\end{equation}
and \(d_s^A=\widetilde\psi_s^A/\kappa\), \Cref{eq:dynamic-phase-center-closed} follows. If \(\Gamma_{m,s}\) depends on direction, its gradient terms must also be included.

\subsection{Equivalence of Complex Projection and Real Nuisance Projection}
Let \(\bmg\in\C^r\setminus\{0\}\). The real orbit tangent space of the common complex scale is
\begin{equation}
\calN=\Span_{\R}\{\calR(\bmg),\calR(\jj\bmg)\}.
\end{equation}
For any \(\bm v\in\C^r\),
\begin{equation}
\calR(\bm P_g^{\perp}\bm v)
=\bm P_{\calN}^{\perp}\calR(\bm v).
\label{eq:real-complex-projection-equivalence}
\end{equation}
This follows because realification of the complex orthogonal decomposition \(\C^r=\Span_{\C}\{\bmg\}\oplus\Span_{\C}\{\bmg\}^{\perp}\) removes exactly the two real directions \(\bmg\) and \(\jj\bmg\).

\subsection{Differential of the Maximal Invariant}
Let \(s=\bmg\herm\bmg\), \(\bm\Pi=\bmg\bmg\herm/s\), and \(\delta\bmg=\bm D\delta\bmvartheta\). Direct differentiation gives
\begin{equation}
\delta\bm\Pi
=\frac{\delta\bmg\bmg\herm+\bmg\delta\bmg\herm}{s}
-\frac{\bmg\bmg\herm}{s^2}
(\bmg\herm\delta\bmg+\delta\bmg\herm\bmg).
\end{equation}
Rearranging,
\begin{equation}
\delta\bm\Pi
=\frac{
\bm P_g^{\perp}\delta\bmg\bmg\herm
+\bmg\delta\bmg\herm\bm P_g^{\perp}}
{s},
\end{equation}
which is \Cref{eq:maximal-invariant-differential}. The Frobenius inner product of the two rank-one terms is zero because \(\bm P_g^{\perp}\bmg=0\), yielding \Cref{eq:maximal-invariant-norm}.

\subsection{Explicit Derivation of the Conditional Fisher Increment}
For the joint Gaussian model in \Cref{eq:joint-correlated-observation}, the conditional distribution is
\begin{equation}
\bm y_a\mid\bm y_0
\sim\calN\!\left(
\bm\mu_a+\bm Q_{a0}\bm Q_0^{-1}(\bm y_0-\bm\mu_0),
\bm Q_{a|0}
\right).
\end{equation}
Differentiation with respect to the parameter gives the conditional-mean Jacobian
\begin{equation}
\bm J_a-\bm Q_{a0}\bm Q_0^{-1}\bm J_0
=\bm J_{a|0}.
\end{equation}
The joint log-likelihood separates into a marginal term and a conditional term. The conditional expectation of the conditional score given \(\bm y_0\) is zero, so the Fisher cross term between the two parts is zero. This yields \Cref{eq:conditional-fisher-increment}.

\subsection{Biased ILS Decision Boundary}
Let the float error be \(\bm v=\bm\beta_N+\bm\varepsilon_N\). The ILS decision boundary between the true integer \(\bm N\) and a competing integer \(\bm N+\bm z\) satisfies
\begin{equation}
\bm v\trans\bm Q^{-1}\bm v
=(\bm v-\bm z)\trans\bm Q^{-1}(\bm v-\bm z),
\end{equation}
or equivalently
\begin{equation}
2\bm z\trans\bm Q^{-1}\bm v
=\bm z\trans\bm Q^{-1}\bm z.
\label{eq:ils-boundary}
\end{equation}
Requiring the true-integer side for both \(\bm z\) and \(-\bm z\) gives \Cref{eq:voronoi-bias-condition}. Moreover,
\begin{equation}
\bm z\trans\bm Q^{-1}\bm\varepsilon_N
\sim\calN(0,d_{\bm z}^2),
\end{equation}
so the probability of crossing toward \(\bm N+\bm z\) is exactly \Cref{eq:biased-pairwise-error}.

\section{Principal Symbols and Unit Conventions}
\renewcommand{\arraystretch}{1.18}
\begin{longtable}{p{0.23\textwidth}p{0.68\textwidth}}
\caption{Unified SN-ASMO notation}\label{tab:symbols}\\
\toprule
Symbol & Meaning\\
\midrule
\endfirsthead
\toprule
Symbol & Meaning\\
\midrule
\endhead
\bottomrule
\endfoot
\(\calX,\bmxi\) & Physical-state manifold and a particular physical state\\
\(\calC_{\rm obs},\bmeta\) & Behavior-minimal observer-state space and a particular observer state\\
\(\calE,\pi,\calE_{\bmxi}\) & Admissible joint domain \(\calE\subseteq\calX\times\calC_{\rm obs}\), projection onto the physical state, and admissible observer fiber\\
\(\calZ,\bmz,\calH\) & Observation space, observation, and joint observation-formation mapping\\
\(\mathscr{T}_U\) & Local trivialization of the joint total space over a base-space neighborhood \(U\)\\
\(\Phi\) & Standard Gaussian cumulative distribution function\\
\(\mathsf{Hor}_e,\mathsf{Ver}_e\) & Horizontal and vertical subspaces of the tangent space of the joint total space\\
\(O\) & Physical reference point defining geometric range and ordinary propagation phase\\
\(\mu,\nu,\ell\) & Internal-reference-element indices or a PI constraint-channel index, distinguished explicitly by context\\
\(\bmu_s^n,\bmu_s^a\) & Unit direction of arrival of satellite \(s\) in the navigation and array frames\\
\(\bm C_a^n,\bm C_n^a\) & Attitude matrix from the array frame to the navigation frame, and its inverse\\
\(\bm b_m^a,\bm r_m^n\) & Position of element \(m\) relative to the physical reference point in the array and navigation frames\\
\(\kappa,\lambda,c_0\) & Wavenumber \(2\pi/\lambda\), carrier wavelength, and vacuum speed of light\\
\(\bmc_s^{(O)}\) & Complete element-response vector relative to physical reference point \(O\)\\
\(\bma_s^{[\mu]}\) & Local response coordinates normalized by element \(\mu\)\\
\(\beta_{\mu,s}^{(O)}\) & Common complex factor associated with the internal reference element\\
\(\bmw,\bm W\) & Single weight vector and multiweight matrix\\
\(g_s,\bmg_s\) & Single-branch composite response and joint multiweight response vector\\
\(\varrho_s^A,h_s^A\) & Positive magnitude and \(\Uone\) phase element of the composite response\\
\(\psi_s^A,\widetilde\psi_s^A\) & Principal composite array-response phase and its real-valued continuous lift, in radians\\
\(\phi_s^{\nav}\) & Ordinary navigation propagation phase, in radians\\
\(\alpha_s^{\nav}\) & Shared navigation complex factor \(A_s\ee^{\jj\phi_s^{\nav}}\)\\
\(T_{12,s}^{A,\rm mod}\) & Model-defined algebraic endpoint phase factor; a PTFE transport only when an admissible nonzero path is supplied\\
\(\omega_s^A\) & Response-induced phase one-form \(\Ima(\dd g_s/g_s)\)\\
\(\Omega_{g_s}\) & Response-tracking \(\Uone\) connection whose horizontal equation is the PTFE\\
\(\kappa_{w,s},\kappa_{c,s}\) & Composite-phase condition numbers for weight and manifold perturbations\\
\(d_s^A\) & Equivalent range associated with composite array-response phase, \(\widetilde\psi_s^A/\kappa\), in metres\\
\(\bm r_{{\rm pc},s}^{\rm tan,a}\) & Tangential component of the weight-dependent dynamic equivalent phase center, in metres\\
\(\bm\Pi(\bmg)\) & Rank-one maximal invariant under a common complex scale\\
\(\bmvartheta\) & Task-declared physical, observer, or joint subcoordinate to be identified\\
\(\bm J^{\rm int}\) & Intrinsic Jacobian after whitening and projection of genuine nuisances\\
\(\bm F^{\rm eff}\) & Effective Fisher information after elimination of genuine nuisances\\
\(\bm G_{\rm FS}\) & Pullback of the Fubini--Study metric to the state space\\
\(L_{i,s}\) & Receiver carrier-phase observation, in cycles\\
\(N_{i,s}\) & Arc-initialized integer ambiguity associated with the continuous carrier lift, an element of \(\Z\)\\
\(B_{i,s}^{\hw}\) & Noninteger hardware bias, in metres\\
\(\Delta_{sr}\) & Inter-satellite difference: current satellite \(s\) minus reference satellite \(r\)\\
\(\nabla_{RB}\) & Between-receiver single difference: rover \(R\) minus base \(B\)\\
\(b_A,d_A\) & Equivalent array-phase bias in cycles and distance\\
\(\bm Q_{\widehat N}\) & Float integer-ambiguity covariance, in cycles squared\\
\(\bm\beta_N\) & Deterministic mean bias of the float ambiguity, in cycles\\
\(d_{\bm z},d_{\min}\) & Distance to a competing integer direction and shortest integer-lattice distance under the ILS metric\\
\end{longtable}

\section*{Resource Availability}
\subsection*{Lead Contact}
Further information and requests should be directed to the lead contact, Xianwei Meng (\href{mailto:mxianwei@hfut.edu.cn}{mxianwei@hfut.edu.cn}).

\subsection*{Materials Availability}
No new physical materials were generated in this theoretical study.

\subsection*{Data and Code Availability}
This theoretical study generated no new experimental dataset or research software. The manuscript specifies the observable quantities, assumptions, and reproducible validation protocols required to test its claims. Any future experimental data and implementation code should be released with the corresponding empirical study, subject to applicable security and licensing constraints.

\subsection*{Supplemental Information}
The supplementary mathematical derivations required to audit the main results are included in the Appendix of this manuscript; no separate supplemental dataset accompanies this version.

\section*{Author Contributions}
X.M. conceived the theory, developed the mathematical framework, performed the derivations and logical audit, designed the validation protocols, and wrote and revised the manuscript.

\section*{Declaration of Interests}
The author declares no competing interests.

\section*{Acknowledgments and Funding}
The author reports no external funding for this work and declares no additional acknowledgments.

\section*{Declaration of AI and AI-Assisted Technologies in the Writing Process}
During the preparation of this work, the author used Deepseek to assist with language refinement, structural organization, and consistency checking. The author independently developed the theory, verified all derivations, claims, and references, reviewed and edited the resulting text, and takes full responsibility for the content of the manuscript.


\begin{thebibliography}{99}
\bibitem{meng2026dualaxis}
X. Meng,
``From primitive observation records to the dual-axis axiom: Set-theoretic foundations of observational geometry,''
ResearchGate Preprint, rev. R3.1, 2026, doi: 10.13140/RG.2.2.15815.05281.

\bibitem{meng2026ptfe}
X. Meng,
``Mathematical and theoretical foundations of the basic phase-transport equation: A unified differential, covariant, and path-functional theory,''
ResearchGate Preprint, rev. R4.1, 2026, doi: 10.13140/RG.2.2.19170.49609.

\bibitem{delorenzo2006}
D. S. De Lorenzo, J. Rife, P. Enge, and D. M. Akos,
``Navigation accuracy and interference rejection for an adaptive GPS antenna array,''
in \emph{Proc. ION GNSS 2006}, Fort Worth, TX, USA, 2006, pp. 763--773.

\bibitem{church2007}
C. Church, I. Gupta, and A. O'Brien,
``Adaptive antenna induced biases in GNSS receivers,''
in \emph{Proc. 63rd Annual Meeting of the Institute of Navigation}, Cambridge, MA, USA, 2007, pp. 204--212.

\bibitem{church2009}
C. M. Church and I. J. Gupta,
``Estimation of adaptive antenna induced code and carrier phase bias in GNSS receivers,''
\emph{NAVIGATION}, vol. 56, no. 3, pp. 151--160, 2009.

\bibitem{obrien2008}
A. J. O'Brien and I. J. Gupta,
``Mitigation of adaptive antenna-induced biases in GNSS receivers,''
in \emph{Proc. 2008 National Technical Meeting of the Institute of Navigation}, San Diego, CA, USA, 2008, pp. 177--185.

\bibitem{xie2021}
Y. Xie, Y. Chen, F. Chen, and F. Wang,
``Carrier phase correction for GNSS space-time array processing based on frequency jerk detection,''
\emph{IET Radar, Sonar \& Navigation}, vol. 15, no. 3, pp. 240--249, 2021, doi: 10.1049/rsn2.12033.

\bibitem{li2023}
S. Li, F. Wang, X. Tang, S. Ni, and H. Lin,
``Anti-jamming GNSS antenna array receiver with reduced phase distortions using a robust phase compensation technique,''
\emph{Remote Sensing}, vol. 15, no. 17, Art. no. 4344, 2023, doi: 10.3390/rs15174344.

\bibitem{wang2024}
Y. Wang,
``Distortion-less carrier phase tracking space-time adaptive processor based on power inversion criterion for GNSS anti-jamming receiver,''
\emph{IET Radar, Sonar \& Navigation}, vol. 18, no. 5, pp. 754--764, 2024, doi: 10.1049/rsn2.12515.

\bibitem{wang2025}
Y. Wang, X. Ye, S. Chen, and Z. Liu,
``Study on the influence of space-time adaptive processor on single point position and real-time kinematic for GNSS antenna array anti-jamming receiver,''
\emph{IET Radar, Sonar \& Navigation}, vol. 19, Art. no. e70035, 2025, doi: 10.1049/rsn2.70035.

\bibitem{li2025artk}
B. Li, Z. Zhang, and W. Miao,
``ARTK: Antenna-array aided RTK,''
in \emph{GNSS Real-Time Kinematic Positioning}. Singapore: Springer, 2025, pp. 255--276, doi: 10.1007/978-981-96-9116-6\_12.

\bibitem{bamberg2023}
T. Bamberg, A. Konovaltsev, and M. Meurer,
``Enabling RTK positioning under jamming: Mitigation of carrier-phase distortions induced by blind spatial filtering,''
\emph{NAVIGATION}, vol. 70, no. 1, 2023, doi: 10.33012/navi.556.

\bibitem{compton1979}
R. T. Compton,
``The power-inversion adaptive array: Concept and performance,''
\emph{IEEE Transactions on Aerospace and Electronic Systems}, vol. AES-15, no. 6, pp. 803--814, 1979, doi: 10.1109/TAES.1979.308765.

\bibitem{lee2013}
J. M. Lee,
\emph{Introduction to Smooth Manifolds}, 2nd ed. New York, NY, USA: Springer, 2013, doi: 10.1007/978-1-4419-9982-5.

\bibitem{kobayashi1963}
S. Kobayashi and K. Nomizu,
\emph{Foundations of Differential Geometry}, vol. I. New York, NY, USA: Wiley, 1963.

\bibitem{kay1993}
S. M. Kay,
\emph{Fundamentals of Statistical Signal Processing, Volume I: Estimation Theory}. Upper Saddle River, NJ, USA: Prentice Hall, 1993.

\bibitem{amari2016}
S. Amari,
\emph{Information Geometry and Its Applications}. Tokyo, Japan: Springer, 2016, doi: 10.1007/978-4-431-55978-8.

\bibitem{vallisneri2008}
M. Vallisneri,
``Use and abuse of the Fisher information matrix in the assessment of gravitational-wave parameter-estimation prospects,''
\emph{Physical Review D}, vol. 77, Art. no. 042001, 2008, doi: 10.1103/PhysRevD.77.042001.

\bibitem{teunissen1995}
P. J. G. Teunissen,
``The least-squares ambiguity decorrelation adjustment: A method for fast GPS integer ambiguity estimation,''
\emph{Journal of Geodesy}, vol. 70, pp. 65--82, 1995, doi: 10.1007/BF00863419.

\bibitem{teunissen1998}
P. J. G. Teunissen,
``Success probability of integer GPS ambiguity rounding and bootstrapping,''
\emph{Journal of Geodesy}, vol. 72, pp. 606--612, 1998, doi: 10.1007/s001900050199.

\bibitem{antex2026}
O. Montenbruck,
\emph{ANTEX: The Antenna Exchange Format, Version 2.0}. International GNSS Service, Issue Apr. 2, 2026.

\end{thebibliography}
\end{document}